\documentclass{lmcs}
\pdfoutput=1
\usepackage[utf8]{inputenc}

\usepackage{lastpage}
\lmcsdoi{19}{4}{30}
\lmcsheading{}{\pageref{LastPage}}{}{}%
{Aug.~30,~2022}{Oct.~22,~2024}{}

\keywords{Cryptography, composable security, category theory}

\usepackage{hyperref}

\usepackage{xspace,amsmath,amssymb,caption,subcaption,stmaryrd}

\newcommand{\cat}[1]{\mathbf{#1}}

\newcommand{\id}[1][]{\ensuremath{\mathrm{id}_{#1}}}
\newcommand{\tuple}[1]{\mathopen{\langle}#1\mathclose{\rangle}}
\newcommand{\ie}{\text{i.e.,}\xspace}
\newcommand{\eg}{\text{e.g.}\xspace}
\newcommand{\N}{\mathbb{N}}
\newcommand{\R}{\mathbb{R}}
\newcommand{\Z}{\mathbb{Z}}
\DeclareMathOperator{\cod}{cod}
\DeclareMathOperator{\dom}{dom}
\newcommand{\Met}{\cat{Met}}
\newcommand{\Set}{\cat{Set}}
\newcommand{\Equ}{\cat{Equ}}
\newcommand{\Eff}{\cat{Eff}}
\newcommand{\A}{\mathcal{A}}
\newcommand{\CC}{\cat{C}}
\newcommand{\DD}{\cat{D}}
\newcommand{\res}{\int}
\newcommand{\resmet}{\res^\Met}
\newcommand{\reseq}{\res^\Equ}
\newcommand{\CF}{\CC_F}
\newcommand{\vect}[1]{\bar{#1}}
\newcommand{\f}{\vect{f}}

\newcommand{\rand}{\$} 
\newcommand{\sem}[1]{\llbracket {#1}\rrbracket}
\newcommand{\inprod}[2]{\langle #1 \,,\, #2 \rangle}

\DeclareMathOperator{\ncomb}{n-comb}
\DeclareMathOperator{\ncombcart}{n-comb!}

\usepackage{tikz}
\usetikzlibrary{decorations.pathreplacing}
\usetikzlibrary{decorations.markings}
\usetikzlibrary{calc}
\usetikzlibrary{arrows}
\usetikzlibrary{arrows.meta}
\usetikzlibrary{shapes.geometric}
\tikzset{>={To[length=2.5pt,width=4pt]}}

\usetikzlibrary{intersections}

\newenvironment{pic}[1][]
{\begin{aligned}\begin{tikzpicture}[font=\tiny,#1]}
{\end{tikzpicture}\end{aligned}}

\pgfdeclarelayer{foreground}
\pgfdeclarelayer{background}
\pgfdeclarelayer{morphismlayer}
\pgfdeclarelayer{dashedmorphismlayer}
\pgfdeclarelayer{edgelayer}
\pgfdeclarelayer{nodelayer}
\pgfsetlayers{background,main,morphismlayer,dashedmorphismlayer,foreground,edgelayer,nodelayer}

\def\thickness{0.7pt}
\makeatletter
\pgfkeys{%
  /tikz/on layer/.code={
    \pgfonlayer{#1}\begingroup
    \aftergroup\endpgfonlayer
    \aftergroup\endgroup
  },
  /tikz/node on layer/.code={
    \gdef\node@@on@layer{%
      \setbox\tikz@tempbox=\hbox\bgroup\pgfonlayer{#1}\unhbox\tikz@tempbox\endpgfonlayer\pgfsetlinewidth{\thickness}\egroup}
    \aftergroup\node@on@layer
  },
  /tikz/end node on layer/.code={
    \endpgfonlayer\endgroup\endgroup
  }
}
\def\node@on@layer{\aftergroup\node@@on@layer}
\makeatother

\tikzstyle{braid}=[double=black, line width=3*\thickness, double distance=\thickness, draw=white, text=black]
\tikzset{every picture/.style={line width=\thickness, draw=black}}
\tikzstyle{pure}=[line width=.7pt]
\tikzstyle{string}=[line width=\thickness]
\tikzstyle{scalar}=[circle, inner sep=0pt, minimum width=15pt, draw, line width=\thickness, fill=white]
\tikzstyle{dot}=[circle, draw=black, fill=black!25, inner sep=.4ex, line width=\thickness, node on layer=foreground]
\tikzstyle{blackdot}=[circle, draw=black, fill=black!75, inner sep=.4ex, line width=\thickness, node on layer=foreground, text=white]
\tikzstyle{whitedot}=[circle, draw=black, fill=white, inner sep=.4ex, line width=\thickness, node on layer=foreground]
\tikzstyle{mixedmorphism}=[morphism, minimum width=30pt, minimum height=16pt, draw, font=\small, inner sep=0pt, fill=white, line width=\thickness,rounded corners=1ex]

\tikzstyle{triangle} = [regular polygon, regular polygon sides=3, draw=black, fill=black!20,scale=0.4, node on layer=foreground]

\tikzstyle{whitetriangle}=[triangle, fill=white]
\tikzstyle{greytriangle}=[triangle, fill=black!20]
\tikzstyle{darkgreytriangle}=[triangle, fill=black!50]
\tikzstyle{blacktriangle}=[triangle, fill=black]

\tikzstyle{invertedtriangle} = [triangle,scale=-1]
\tikzstyle{whiteinvertedtriangle}=[invertedtriangle, fill=white]
\tikzstyle{greyinvertedtriangle}=[invertedtriangle, fill=black!20]
\tikzstyle{darkgreyinvertedtriangle}=[invertedtriangle, fill=black!50]
\tikzstyle{blackinvertedtriangle}=[invertedtriangle, fill=black]

\tikzset{functor1/.style={gray!50}}
\tikzset{functor2/.style={gray!50}}
\tikzstyle{thick}=[line width=\thickness]
\tikzstyle{tiny}=[font=\tiny]
\tikzset{circlelabel/.style={draw, thick, circle, inner sep=-5pt,
 fill=white, minimum width=14pt, fill opacity=1, node on layer=foreground, font=\scriptsize}}
\tikzset{shade 1 transparent/.style={fill=black!20, fill opacity=0.8}}
\tikzset{shade 2 transparent/.style={fill=black!35, fill opacity=0.8}}
\tikzset{shade 3 transparent/.style={fill=black!50, fill opacity=0.8}}
\tikzset{shade 4 transparent/.style={fill=black!65, fill opacity=0.8}}
\tikzset{shade 1/.style={fill=black!16, fill opacity=1}}
\tikzset{shade 2/.style={fill=black!28, fill opacity=1}}
\tikzset{shade 3/.style={fill=black!40, fill opacity=1}}
\tikzset{shade 4/.style={fill=black!52, fill opacity=1}}

\def\strarr{line width=0.7pt, length=4pt, width=5pt, color=black}

\tikzset{arrow/.style={decoration={
    markings,
    mark=at position #1 with \arrow{>[\strarr]}},
    postaction=decorate},
    reverse arrow/.style={decoration={
    markings,
    mark=at position #1 with {{\arrow{<[\strarr]}}}},
    postaction=decorate}
}

\tikzset{overbrace/.style={
     decoration={brace},
     decorate}
}
\tikzset{underbrace/.style={
     decoration={brace, mirror},
     decorate}
}

\newif\ifblack\pgfkeys{/tikz/black/.is if=black}
\newif\ifwedge\pgfkeys{/tikz/wedge/.is if=wedge}
\newif\ifvflip\pgfkeys{/tikz/vflip/.is if=vflip}
\newif\ifhflip\pgfkeys{/tikz/hflip/.is if=hflip}
\newif\ifhvflip\pgfkeys{/tikz/hvflip/.is if=hvflip}
\newif\ifconnectsw\pgfkeys{/tikz/connect sw/.is if=connectsw}
\newif\ifconnectse\pgfkeys{/tikz/connect se/.is if=connectse}
\newif\ifconnectn\pgfkeys{/tikz/connect n/.is if=connectn}
\newif\ifconnects\pgfkeys{/tikz/connect s/.is if=connects}
\newif\ifconnectnw\pgfkeys{/tikz/connect nw/.is if=connectnw}
\newif\ifconnectne\pgfkeys{/tikz/connect ne/.is if=connectne}
\newif\ifconnectnwf\pgfkeys{/tikz/connect nw >/.is if=connectnwf}
\newif\ifconnectnef\pgfkeys{/tikz/connect ne >/.is if=connectnef}
\newif\ifconnectswf\pgfkeys{/tikz/connect sw >/.is if=connectswf}
\newif\ifconnectsef\pgfkeys{/tikz/connect se >/.is if=connectsef}
\newif\ifconnectnf\pgfkeys{/tikz/connect n >/.is if=connectnf}
\newif\ifconnectsf\pgfkeys{/tikz/connect s >/.is if=connectsf}
\newif\ifconnectnwr\pgfkeys{/tikz/connect nw </.is if=connectnwr}
\newif\ifconnectner\pgfkeys{/tikz/connect ne </.is if=connectner}
\newif\ifconnectswr\pgfkeys{/tikz/connect sw </.is if=connectswr}
\newif\ifconnectser\pgfkeys{/tikz/connect se </.is if=connectser}
\newif\ifconnectnr\pgfkeys{/tikz/connect n </.is if=connectnr}
\newif\ifconnectsr\pgfkeys{/tikz/connect s </.is if=connectsr}
\tikzset{keylengthnw/.initial=\connectheight}
\tikzset{keylengthn/.initial =\connectheight}
\tikzset{keylengthne/.initial=\connectheight}
\tikzset{keylengthsw/.initial=\connectheight}
\tikzset{keylengths/.initial =\connectheight}
\tikzset{keylengthse/.initial=\connectheight}
\tikzset{connect nw length/.style={connect nw=true, keylengthnw={#1}}}
\tikzset{connect n length/.style ={connect n =true, keylengthn ={#1}}}
\tikzset{connect ne length/.style={connect ne=true, keylengthne={#1}}}
\tikzset{connect sw length/.style={connect sw=true, keylengthsw={#1}}}
\tikzset{connect s length/.style ={connect s =true, keylengths ={#1}}}
\tikzset{connect se length/.style={connect se=true, keylengthse={#1}}}
\tikzset{connect nw < length/.style={connect nw <=true, keylengthnw={#1}}}
\tikzset{connect n < length/.style ={connect n <=true,  keylengthn ={#1}}}
\tikzset{connect ne < length/.style={connect ne <=true, keylengthne={#1}}}
\tikzset{connect sw < length/.style={connect sw <=true, keylengthnw={#1}}}
\tikzset{connect s < length/.style ={connect s <=true,  keylengths ={#1}}}
\tikzset{connect se < length/.style={connect se <=true, keylengthse={#1}}}
\tikzset{connect nw > length/.style={connect nw >=true, keylengthnw={#1}}}
\tikzset{connect n > length/.style ={connect n >=true,  keylengthn ={#1}}}
\tikzset{connect ne > length/.style={connect ne >=true, keylengthne={#1}}}
\tikzset{connect sw > length/.style={connect sw >=true, keylengthsw={#1}}}
\tikzset{connect s > length/.style ={connect s >=true,  keylengths ={#1}}}
\tikzset{connect se > length/.style={connect se >=true, keylengthse={#1}}}

\newlength\morphismheight
\newlength\wedgewidth
\newlength\minimummorphismwidth
\newlength\stateheight
\newlength\minimumstatewidth
\newlength\connectheight
\tikzset{width/.initial=\minimummorphismwidth}
\tikzset{colour/.initial=white}

\makeatletter
\pgfarrowsdeclare{thickarrow}{thickarrow}
{
  \pgfutil@tempdima=-0.84pt%
  \advance\pgfutil@tempdima by-1.3\pgflinewidth%
  \pgfutil@tempdimb=-1.7pt%
  \advance\pgfutil@tempdimb by.625\pgflinewidth%
  \pgfarrowsleftextend{+\pgfutil@tempdima}
  \pgfarrowsrightextend{+\pgfutil@tempdimb}
}
{
  \pgfmathparse{\pgfgetarrowoptions{thickarrow}}%
  \pgfsetlinewidth{1.25 pt}
  \pgfutil@tempdima=0.28pt%
  \advance\pgfutil@tempdima by.3\pgflinewidth%
  \pgfsetlinewidth{0.8\pgflinewidth}
  \pgfsetdash{}{+0pt}
  \pgfsetroundcap
  \pgfsetroundjoin
  \pgfpathmoveto{\pgfqpoint{-3\pgfutil@tempdima}{4\pgfutil@tempdima}}
  \pgfpathcurveto
  {\pgfqpoint{-2.75\pgfutil@tempdima}{2.5\pgfutil@tempdima}}
  {\pgfqpoint{0pt}{0.25\pgfutil@tempdima}}
  {\pgfqpoint{0.75\pgfutil@tempdima}{0pt}}
  \pgfpathcurveto
  {\pgfqpoint{0pt}{-0.25\pgfutil@tempdima}}
  {\pgfqpoint{-2.75\pgfutil@tempdima}{-2.5\pgfutil@tempdima}}
  {\pgfqpoint{-3\pgfutil@tempdima}{-4\pgfutil@tempdima}}
  \pgfusepathqstroke
}
\pgfarrowsdeclare{reversethickarrow}{reversethickarrow}
{
  \pgfutil@tempdima=-0.84pt%
  \advance\pgfutil@tempdima by-1.3\pgflinewidth%
  \pgfutil@tempdimb=0.2pt%
  \advance\pgfutil@tempdimb by.625\pgflinewidth%
  \pgfarrowsleftextend{+\pgfutil@tempdima}
  \pgfarrowsrightextend{+\pgfutil@tempdimb}
}
{
  \pgftransformxscale{-1}
  \pgfmathparse{\pgfgetarrowoptions{thickarrow}}%
  \ifpgfmathunitsdeclared%
    \pgfmathparse{\pgfmathresult pt}%
  \else%
    \pgfmathparse{\pgfmathresult*\pgflinewidth}%
  \fi%
  \let\thickness=\pgfmathresult
  \pgfsetlinewidth{1.25 pt}
  \pgfutil@tempdima=0.28pt%
  \advance\pgfutil@tempdima by.3\pgflinewidth%
  \pgfsetlinewidth{0.8\pgflinewidth}
  \pgfsetdash{}{+0pt}
  \pgfsetroundcap
  \pgfsetroundjoin
  \pgfpathmoveto{\pgfqpoint{-3\pgfutil@tempdima}{4\pgfutil@tempdima}}
  \pgfpathcurveto
  {\pgfqpoint{-2.75\pgfutil@tempdima}{2.5\pgfutil@tempdima}}
  {\pgfqpoint{0pt}{0.25\pgfutil@tempdima}}
  {\pgfqpoint{0.75\pgfutil@tempdima}{0pt}}
  \pgfpathcurveto
  {\pgfqpoint{0pt}{-0.25\pgfutil@tempdima}}
  {\pgfqpoint{-2.75\pgfutil@tempdima}{-2.5\pgfutil@tempdima}}
  {\pgfqpoint{-3\pgfutil@tempdima}{-4\pgfutil@tempdima}}
  \pgfusepathqstroke
}
\makeatother

\tikzset{diredge/.style={decoration={
  markings,
  mark=at position 0.525 with {\arrow{#1}}},postaction={decorate}}}
\tikzset{
    diredge/.default=>
}

\tikzset{diredgestart/.style={decoration={
  markings,
  mark=at position 4pt with {\arrow{#1}}},postaction={decorate}}}
\tikzset{
    diredgestart/.default=<
}

\tikzset{diredgeend/.style={decoration={
  markings,
  mark=at position 1 with {\arrow{#1}}},postaction={decorate}}}
\tikzset{
    diredgeend/.default=>
}

\makeatletter
\pgfdeclareshape{ground}
{
    \savedanchor\centerpoint
    {
        \pgf@x=0pt
        \pgf@y=0pt
    }
    \anchor{center}{\centerpoint}
    \anchorborder{\centerpoint}
    \anchor{north}
    {
        \pgf@x=0pt
        \pgf@y=0.5*0.33*\stateheight
    }
    \anchor{south}
    {
        \pgf@x=0pt
        \pgf@y=0pt
    }
    \saveddimen\overallwidth
    {
        \pgfkeysgetvalue{/pgf/minimum width}{\minwidth}
        \pgf@x=\minimumstatewidth
        \ifdim\pgf@x<\minwidth
            \pgf@x=\minwidth
        \fi
    }
    \backgroundpath
    {
        \begin{pgfonlayer}{foreground}
        \pgfsetstrokecolor{black}
        \pgfsetlinewidth{1.25pt}
        \ifhflip
            \pgftransformyscale{-1}
        \fi
        \pgftransformscale{0.5}
        \pgfpathmoveto{\pgfpoint{-0.5*\overallwidth}{0}}
        \pgfpathlineto{\pgfpoint{0.5*\overallwidth}{0}}
        \pgfpathmoveto{\pgfpoint{-0.33*\overallwidth}{0.33*\stateheight}}
        \pgfpathlineto{\pgfpoint{0.33*\overallwidth}{0.33*\stateheight}}
        \pgfpathmoveto{\pgfpoint{-0.16*\overallwidth}{0.66*\stateheight}}
        \pgfpathlineto{\pgfpoint{0.16*\overallwidth}{0.66*\stateheight}}
        \pgfpathmoveto{\pgfpoint{-0.02*\overallwidth}{\stateheight}}
        \pgfpathlineto{\pgfpoint{0.02*\overallwidth}{\stateheight}}
        \pgfusepath{stroke}
        \end{pgfonlayer}
    }
}
\tikzset{forward arrow style/.style={every to/.style, decoration={
    markings,
    mark=at position 0.5*\pgfdecoratedpathlength+2pt with \arrow{>[\strarr]}},
    postaction=decorate}}
\tikzset{reverse arrow style/.style={every to/.style, decoration={
    markings,
    mark=at position 0.5*\pgfdecoratedpathlength+2pt with \arrow{<[\strarr]}},
    postaction=decorate}}
\pgfdeclareshape{morphismshape}
{
    \savedanchor\centerpoint
    {
        \pgf@x=0pt
        \pgf@y=0pt
    }
    \anchor{center}{\centerpoint}
    \anchorborder{\centerpoint}
    \saveddimen\savedlengthnw
    {
        \pgfkeysgetvalue{/tikz/keylengthnw}{\len}
        \pgf@x=\len
    }
    \saveddimen\savedlengthn
    {
        \pgfkeysgetvalue{/tikz/keylengthn}{\len}
        \pgf@x=\len
    }
    \saveddimen\savedlengthne
    {
        \pgfkeysgetvalue{/tikz/keylengthne}{\len}
        \pgf@x=\len
    }
    \saveddimen\savedlengthsw
    {
        \pgfkeysgetvalue{/tikz/keylengthsw}{\len}
        \pgf@x=\len
    }
    \saveddimen\savedlengths
    {
        \pgfkeysgetvalue{/tikz/keylengths}{\len}
        \pgf@x=\len
    }
    \saveddimen\savedlengthse
    {
        \pgfkeysgetvalue{/tikz/keylengthse}{\len}
        \pgf@x=\len
    }
    \saveddimen\overallwidth
    {
        \pgfkeysgetvalue{/tikz/width}{\minwidth}
        \pgf@x=\wd\pgfnodeparttextbox
        \ifdim\pgf@x<\minwidth
            \pgf@x=\minwidth
        \fi
    }
    \savedanchor{\upperrightcorner}
    {
        \pgf@y=.5\ht\pgfnodeparttextbox
        \advance\pgf@y by -.5\dp\pgfnodeparttextbox
        \pgf@x=.5\wd\pgfnodeparttextbox
    }
    \anchor{north}
    {
        \pgf@x=0pt
        \pgf@y=0.5\morphismheight
    }
    \anchor{north east}
    {
        \pgf@x=\overallwidth
        \multiply \pgf@x by 2
        \divide \pgf@x by 5
        \pgf@y=0.5\morphismheight
    }
    \anchor{east}
    {
        \pgf@x=\overallwidth
        \divide \pgf@x by 2
        \advance \pgf@x by 5pt
        \pgf@y=0pt
    }
    \anchor{west}
    {
        \pgf@x=-\overallwidth
        \divide \pgf@x by 2
        \advance \pgf@x by -5pt
        \pgf@y=0pt
    }
    \anchor{north west}
    {
        \pgf@x=-\overallwidth
        \multiply \pgf@x by 2
        \divide \pgf@x by 5
        \pgf@y=0.5\morphismheight
    }
    \anchor{connect nw}
    {
        \pgf@x=-\overallwidth
        \multiply \pgf@x by 2
        \divide \pgf@x by 5
        \pgf@y=0.5\morphismheight
        \advance\pgf@y by \savedlengthnw
    }
    \anchor{connect ne}
    {
        \pgf@x=\overallwidth
        \multiply \pgf@x by 2
        \divide \pgf@x by 5
        \pgf@y=0.5\morphismheight
        \advance\pgf@y by \savedlengthne
    }
    \anchor{connect sw}
    {
        \pgf@x=-\overallwidth
        \multiply \pgf@x by 2
        \divide \pgf@x by 5
        \pgf@y=-0.5\morphismheight
        \advance\pgf@y by -\savedlengthsw
    }
    \anchor{connect se}
    {
        \pgf@x=\overallwidth
        \multiply \pgf@x by 2
        \divide \pgf@x by 5
        \pgf@y=-0.5\morphismheight
        \advance\pgf@y by -\savedlengthse
    }
    \anchor{connect n}
    {
        \pgf@x=0pt
        \pgf@y=0.5\morphismheight
        \advance\pgf@y by \savedlengthn
    }
    \anchor{connect s}
    {
        \pgf@x=0pt
        \pgf@y=-0.5\morphismheight
        \advance\pgf@y by -\savedlengths
    }
    \anchor{south east}
    {
        \pgf@x=\overallwidth
        \multiply \pgf@x by 2
        \divide \pgf@x by 5
        \pgf@y=-0.5\morphismheight
    }
    \anchor{south west}
    {
        \pgf@x=-\overallwidth
        \multiply \pgf@x by 2
        \divide \pgf@x by 5
        \pgf@y=-0.5\morphismheight
    }
    \anchor{south}
    {
        \pgf@x=0pt
        \pgf@y=-0.5\morphismheight
    }
    \anchor{text}
    {
        \upperrightcorner
        \pgf@x=-\pgf@x
        \pgf@y=-\pgf@y
    }
    \backgroundpath
    {
    \begin{scope}
        \pgfkeysgetvalue{/tikz/fill}{\morphismfill}
        \pgfsetstrokecolor{black}
        \pgfsetlinewidth{\thickness}
        \begin{scope}
            \begin{pgfonlayer}{morphismlayer}
        \pgfsetstrokecolor{black}
        \pgfsetfillcolor{\pgfkeysvalueof{/tikz/colour}}
        \pgfsetlinewidth{\thickness}
                \ifhflip
                    \pgftransformyscale{-1}
                \fi
                \ifvflip
                    \pgftransformxscale{-1}
                \fi
                \ifhvflip
                    \pgftransformxscale{-1}
                    \pgftransformyscale{-1}
                \fi
                \pgfpathmoveto{\pgfpoint
                    {-0.5*\overallwidth-5pt}
                    {0.5*\morphismheight}}
                \pgfpathlineto{\pgfpoint
                    {0.5*\overallwidth+5pt}
                    {0.5*\morphismheight}}
                \ifwedge
                    \pgfpathlineto{\pgfpoint
                        {0.5*\overallwidth + \wedgewidth}
                        {-0.5*\morphismheight}}
                \else
                    \pgfpathlineto{\pgfpoint
                        {0.5*\overallwidth + 5pt}
                        {-0.5*\morphismheight}}
                \fi
                \pgfpathlineto{\pgfpoint
                    {-0.5*\overallwidth-5pt}
                    {-0.5*\morphismheight}}
                \pgfpathclose
                \pgfusepath{fill,stroke}
            \end{pgfonlayer}
        \end{scope}
        \ifconnectnw
            \pgfpathmoveto{\pgfpoint
                {-0.4*\overallwidth}
                {0.5*\morphismheight}}
            \pgfpathlineto{\pgfpoint
                {-0.4*\overallwidth}
                {0.5*\morphismheight+\savedlengthnw}}
            \pgfusepath{stroke}
        \fi
        \ifconnectne
            \pgfpathmoveto{\pgfpoint
                {0.4*\overallwidth}
                {0.5*\morphismheight}}
            \pgfpathlineto{\pgfpoint
                {0.4*\overallwidth}
                {0.5*\morphismheight+\savedlengthne}}
            \pgfusepath{stroke}
        \fi
        \ifconnectsw
            \pgfpathmoveto{\pgfpoint
                {-0.4*\overallwidth}
                {-0.5*\morphismheight}}
            \pgfpathlineto{\pgfpoint
                {-0.4*\overallwidth}
                {-0.5*\morphismheight-\savedlengthsw}}
            \pgfusepath{stroke}
        \fi
        \ifconnectse
            \pgfpathmoveto{\pgfpoint
                {0.4*\overallwidth}
                {-0.5*\morphismheight}}
            \pgfpathlineto{\pgfpoint
                {0.4*\overallwidth}
                {-0.5*\morphismheight-\savedlengthse}}
            \pgfusepath{stroke}
        \fi
        \ifconnectn
            \pgfpathmoveto{\pgfpoint
                {0pt}
                {0.5*\morphismheight}}
            \pgfpathlineto{\pgfpoint
                {0pt}
                {0.5*\morphismheight+\savedlengthn}}
            \pgfusepath{stroke}
        \fi
        \ifconnects
            \pgfpathmoveto{\pgfpoint
                {0pt}
                {-0.5*\morphismheight}}
            \pgfpathlineto{\pgfpoint
                {0pt}
                {-0.5*\morphismheight-\savedlengths}}
            \pgfusepath{stroke}
        \fi
        \ifconnectnwf
            \draw [forward arrow style] (-0.4*\overallwidth,0.5*\morphismheight)
                to (-0.4*\overallwidth,0.5*\morphismheight+\savedlengthnw);
        \fi
        \ifconnectnef
            \draw [forward arrow style] (0.4*\overallwidth,0.5*\morphismheight)
                to (0.4*\overallwidth,0.5*\morphismheight+\savedlengthne);
        \fi
        \ifconnectswf
            \draw [forward arrow style] (-0.4*\overallwidth,-0.5*\morphismheight-\savedlengthsw)
                to (-0.4*\overallwidth,-0.5*\morphismheight);
        \fi
        \ifconnectsef
            \draw [forward arrow style] (0.4*\overallwidth,-0.5*\morphismheight-\savedlengthse)
                to (0.4*\overallwidth,-0.5*\morphismheight);
        \fi
        \ifconnectnf
            \draw [forward arrow style] (0,0.5*\morphismheight)
                to (0,0.5*\morphismheight+\savedlengthn);
        \fi
        \ifconnectsf
            \draw [forward arrow style] (0,-0.5*\morphismheight-\savedlengths)
                to (0,-0.5*\morphismheight);
        \fi
        \ifconnectnwr
            \draw [reverse arrow style] (-0.4*\overallwidth,0.5*\morphismheight)
                to (-0.4*\overallwidth,0.5*\morphismheight+\savedlengthnw);
        \fi
        \ifconnectner
            \draw [reverse arrow style] (0.4*\overallwidth,0.5*\morphismheight)
                to (0.4*\overallwidth,0.5*\morphismheight+\savedlengthne);
        \fi
        \ifconnectswr
            \draw [reverse arrow style] (-0.4*\overallwidth,-0.5*\morphismheight-\savedlengthsw)
                to (-0.4*\overallwidth,-0.5*\morphismheight);
        \fi
        \ifconnectser
            \draw [reverse arrow style] (0.4*\overallwidth,-0.5*\morphismheight-\savedlengthse)
                to (0.4*\overallwidth,-0.5*\morphismheight);
        \fi
        \ifconnectnr
            \draw [reverse arrow style] (0,0.5*\morphismheight)
                to (0,0.5*\morphismheight+\savedlengthn);
        \fi
        \ifconnectsr
            \draw [reverse arrow style] (0,-0.5*\morphismheight-\savedlengths)
                to (0,-0.5*\morphismheight);
        \fi
    \end{scope}
    }
}
\tikzset{morphism/.style={morphismshape, node on layer=foreground}}
\pgfdeclareshape{dashedmorphismshape}
{
    \savedanchor\centerpoint
    {
        \pgf@x=0pt
        \pgf@y=0pt
    }
    \anchor{center}{\centerpoint}
    \anchorborder{\centerpoint}
    \saveddimen\savedlengthnw
    {
        \pgfkeysgetvalue{/tikz/keylengthnw}{\len}
        \pgf@x=\len
    }
    \saveddimen\savedlengthn
    {
        \pgfkeysgetvalue{/tikz/keylengthn}{\len}
        \pgf@x=\len
    }
    \saveddimen\savedlengthne
    {
        \pgfkeysgetvalue{/tikz/keylengthne}{\len}
        \pgf@x=\len
    }
    \saveddimen\savedlengthsw
    {
        \pgfkeysgetvalue{/tikz/keylengthsw}{\len}
        \pgf@x=\len
    }
    \saveddimen\savedlengths
    {
        \pgfkeysgetvalue{/tikz/keylengths}{\len}
        \pgf@x=\len
    }
    \saveddimen\savedlengthse
    {
        \pgfkeysgetvalue{/tikz/keylengthse}{\len}
        \pgf@x=\len
    }
    \saveddimen\overallwidth
    {
        \pgfkeysgetvalue{/tikz/width}{\minwidth}
        \pgf@x=\wd\pgfnodeparttextbox
        \ifdim\pgf@x<\minwidth
            \pgf@x=\minwidth
        \fi
    }
    \savedanchor{\upperrightcorner}
    {
        \pgf@y=.5\ht\pgfnodeparttextbox
        \advance\pgf@y by -.5\dp\pgfnodeparttextbox
        \pgf@x=.5\wd\pgfnodeparttextbox
    }
    \anchor{north}
    {
        \pgf@x=0pt
        \pgf@y=0.5\morphismheight
    }
    \anchor{north east}
    {
        \pgf@x=\overallwidth
        \multiply \pgf@x by 2
        \divide \pgf@x by 5
        \pgf@y=0.5\morphismheight
    }
    \anchor{east}
    {
        \pgf@x=\overallwidth
        \divide \pgf@x by 2
        \advance \pgf@x by 5pt
        \pgf@y=0pt
    }
    \anchor{west}
    {
        \pgf@x=-\overallwidth
        \divide \pgf@x by 2
        \advance \pgf@x by -5pt
        \pgf@y=0pt
    }
    \anchor{north west}
    {
        \pgf@x=-\overallwidth
        \multiply \pgf@x by 2
        \divide \pgf@x by 5
        \pgf@y=0.5\morphismheight
    }
    \anchor{connect nw}
    {
        \pgf@x=-\overallwidth
        \multiply \pgf@x by 2
        \divide \pgf@x by 5
        \pgf@y=0.5\morphismheight
        \advance\pgf@y by \savedlengthnw
    }
    \anchor{connect ne}
    {
        \pgf@x=\overallwidth
        \multiply \pgf@x by 2
        \divide \pgf@x by 5
        \pgf@y=0.5\morphismheight
        \advance\pgf@y by \savedlengthne
    }
    \anchor{connect sw}
    {
        \pgf@x=-\overallwidth
        \multiply \pgf@x by 2
        \divide \pgf@x by 5
        \pgf@y=-0.5\morphismheight
        \advance\pgf@y by -\savedlengthsw
    }
    \anchor{connect se}
    {
        \pgf@x=\overallwidth
        \multiply \pgf@x by 2
        \divide \pgf@x by 5
        \pgf@y=-0.5\morphismheight
        \advance\pgf@y by -\savedlengthse
    }
    \anchor{connect n}
    {
        \pgf@x=0pt
        \pgf@y=0.5\morphismheight
        \advance\pgf@y by \savedlengthn
    }
    \anchor{connect s}
    {
        \pgf@x=0pt
        \pgf@y=-0.5\morphismheight
        \advance\pgf@y by -\savedlengths
    }
    \anchor{south east}
    {
        \pgf@x=\overallwidth
        \multiply \pgf@x by 2
        \divide \pgf@x by 5
        \pgf@y=-0.5\morphismheight
    }
    \anchor{south west}
    {
        \pgf@x=-\overallwidth
        \multiply \pgf@x by 2
        \divide \pgf@x by 5
        \pgf@y=-0.5\morphismheight
    }
    \anchor{south}
    {
        \pgf@x=0pt
        \pgf@y=-0.5\morphismheight
    }
    \anchor{text}
    {
        \upperrightcorner
        \pgf@x=-\pgf@x
        \pgf@y=-\pgf@y
    }
    \backgroundpath
    {
    \begin{scope}
        \pgfkeysgetvalue{/tikz/fill}{\morphismfill}
        \pgfsetstrokecolor{black}
        \pgfsetlinewidth{\thickness}
        \begin{scope}
            \begin{pgfonlayer}{dashedmorphismlayer}
        \pgfsetstrokecolor{black}
        \pgfsetdash{{3pt}{3pt}}{0pt}
        \pgfsetfillcolor{\pgfkeysvalueof{/tikz/colour}}
        \pgfsetlinewidth{\thickness}
                \ifhflip
                    \pgftransformyscale{-1}
                \fi
                \ifvflip
                    \pgftransformxscale{-1}
                \fi
                \ifhvflip
                    \pgftransformxscale{-1}
                    \pgftransformyscale{-1}
                \fi
                \pgfpathmoveto{\pgfpoint
                    {-0.5*\overallwidth-5pt}
                    {0.5*\morphismheight}}
                \pgfpathlineto{\pgfpoint
                    {0.5*\overallwidth+5pt}
                    {0.5*\morphismheight}}
                \ifwedge
                    \pgfpathlineto{\pgfpoint
                        {0.5*\overallwidth + \wedgewidth}
                        {-0.5*\morphismheight}}
                \else
                    \pgfpathlineto{\pgfpoint
                        {0.5*\overallwidth + 5pt}
                        {-0.5*\morphismheight}}
                \fi
                \pgfpathlineto{\pgfpoint
                    {-0.5*\overallwidth-5pt}
                    {-0.5*\morphismheight}}
                \pgfpathclose
                \pgfusepath{fill,stroke}
            \end{pgfonlayer}
        \end{scope}
        \ifconnectnw
            \pgfpathmoveto{\pgfpoint
                {-0.4*\overallwidth}
                {0.5*\morphismheight}}
            \pgfpathlineto{\pgfpoint
                {-0.4*\overallwidth}
                {0.5*\morphismheight+\savedlengthnw}}
            \pgfusepath{stroke}
        \fi
        \ifconnectne
            \pgfpathmoveto{\pgfpoint
                {0.4*\overallwidth}
                {0.5*\morphismheight}}
            \pgfpathlineto{\pgfpoint
                {0.4*\overallwidth}
                {0.5*\morphismheight+\savedlengthne}}
            \pgfusepath{stroke}
        \fi
        \ifconnectsw
            \pgfpathmoveto{\pgfpoint
                {-0.4*\overallwidth}
                {-0.5*\morphismheight}}
            \pgfpathlineto{\pgfpoint
                {-0.4*\overallwidth}
                {-0.5*\morphismheight-\savedlengthsw}}
            \pgfusepath{stroke}
        \fi
        \ifconnectse
            \pgfpathmoveto{\pgfpoint
                {0.4*\overallwidth}
                {-0.5*\morphismheight}}
            \pgfpathlineto{\pgfpoint
                {0.4*\overallwidth}
                {-0.5*\morphismheight-\savedlengthse}}
            \pgfusepath{stroke}
        \fi
        \ifconnectn
            \pgfpathmoveto{\pgfpoint
                {0pt}
                {0.5*\morphismheight}}
            \pgfpathlineto{\pgfpoint
                {0pt}
                {0.5*\morphismheight+\savedlengthn}}
            \pgfusepath{stroke}
        \fi
        \ifconnects
            \pgfpathmoveto{\pgfpoint
                {0pt}
                {-0.5*\morphismheight}}
            \pgfpathlineto{\pgfpoint
                {0pt}
                {-0.5*\morphismheight-\savedlengths}}
            \pgfusepath{stroke}
        \fi
        \ifconnectnwf
            \draw [forward arrow style] (-0.4*\overallwidth,0.5*\morphismheight)
                to (-0.4*\overallwidth,0.5*\morphismheight+\savedlengthnw);
        \fi
        \ifconnectnef
            \draw [forward arrow style] (0.4*\overallwidth,0.5*\morphismheight)
                to (0.4*\overallwidth,0.5*\morphismheight+\savedlengthne);
        \fi
        \ifconnectswf
            \draw [forward arrow style] (-0.4*\overallwidth,-0.5*\morphismheight-\savedlengthsw)
                to (-0.4*\overallwidth,-0.5*\morphismheight);
        \fi
        \ifconnectsef
            \draw [forward arrow style] (0.4*\overallwidth,-0.5*\morphismheight-\savedlengthse)
                to (0.4*\overallwidth,-0.5*\morphismheight);
        \fi
        \ifconnectnf
            \draw [forward arrow style] (0,0.5*\morphismheight)
                to (0,0.5*\morphismheight+\savedlengthn);
        \fi
        \ifconnectsf
            \draw [forward arrow style] (0,-0.5*\morphismheight-\savedlengths)
                to (0,-0.5*\morphismheight);
        \fi
        \ifconnectnwr
            \draw [reverse arrow style] (-0.4*\overallwidth,0.5*\morphismheight)
                to (-0.4*\overallwidth,0.5*\morphismheight+\savedlengthnw);
        \fi
        \ifconnectner
            \draw [reverse arrow style] (0.4*\overallwidth,0.5*\morphismheight)
                to (0.4*\overallwidth,0.5*\morphismheight+\savedlengthne);
        \fi
        \ifconnectswr
            \draw [reverse arrow style] (-0.4*\overallwidth,-0.5*\morphismheight-\savedlengthsw)
                to (-0.4*\overallwidth,-0.5*\morphismheight);
        \fi
        \ifconnectser
            \draw [reverse arrow style] (0.4*\overallwidth,-0.5*\morphismheight-\savedlengthse)
                to (0.4*\overallwidth,-0.5*\morphismheight);
        \fi
        \ifconnectnr
            \draw [reverse arrow style] (0,0.5*\morphismheight)
                to (0,0.5*\morphismheight+\savedlengthn);
        \fi
        \ifconnectsr
            \draw [reverse arrow style] (0,-0.5*\morphismheight-\savedlengths)
                to (0,-0.5*\morphismheight);
        \fi
    \end{scope}
    }
}
\tikzset{dashedmorphism/.style={dashedmorphismshape, node on layer=foreground}}
\pgfdeclareshape{swish right}
{
    \savedanchor\centerpoint
    {
        \pgf@x=0pt
        \pgf@y=0pt
    }
    \anchor{center}{\centerpoint}
    \anchorborder{\centerpoint}
    \anchor{north}
    {
        \pgf@x=\minimummorphismwidth
        \divide\pgf@x by 5
        \pgf@y=\morphismheight
        \divide\pgf@y by 2
        \advance\pgf@y by \connectheight
    }
    \anchor{south}
    {
        \pgf@x=-\minimummorphismwidth
        \divide\pgf@x by 5
        \pgf@y=-\morphismheight
        \divide\pgf@y by 2
        \advance\pgf@y by -\connectheight
    }
    \backgroundpath
    {
        \pgfsetstrokecolor{black}
        \pgfsetlinewidth{\thickness}
        \pgfpathmoveto{\pgfpoint
            {-0.2*\minimummorphismwidth}
            {-0.5*\morphismheight-\connectheight}}
        \pgfpathcurveto
            {\pgfpoint{-0.2*\minimummorphismwidth}{0pt}}
            {\pgfpoint{0.2*\minimummorphismwidth}{0pt}}
            {\pgfpoint
                {0.2*\minimummorphismwidth}
                {0.5*\morphismheight+\connectheight}}
        \pgfusepath{stroke}
    }
}
\pgfdeclareshape{swish left}
{
    \savedanchor\centerpoint
    {
        \pgf@x=0pt
        \pgf@y=0pt
    }
    \anchor{center}{\centerpoint}
    \anchorborder{\centerpoint}
    \anchor{north}
    {
        \pgf@x=-\minimummorphismwidth
        \divide\pgf@x by 5
        \pgf@y=\morphismheight
        \divide\pgf@y by 2
        \advance\pgf@y by \connectheight
    }
    \anchor{south}
    {
        \pgf@x=\minimummorphismwidth
        \divide\pgf@x by 5
        \pgf@y=-\morphismheight
        \divide\pgf@y by 2
        \advance\pgf@y by -\connectheight
    }
    \backgroundpath
    {
        \pgfsetstrokecolor{black}
        \pgfsetlinewidth{\thickness}
        \pgfpathmoveto{\pgfpoint
            {0.2*\minimummorphismwidth}
            {-0.5*\morphismheight-\connectheight}}
        \pgfpathcurveto
            {\pgfpoint{0.2*\minimummorphismwidth}{0pt}}
            {\pgfpoint{-0.2*\minimummorphismwidth}{0pt}}
            {\pgfpoint
                {-0.2*\minimummorphismwidth}
                {0.5*\morphismheight+\connectheight}}
        \pgfusepath{stroke}
    }
}
\pgfdeclareshape{state}
{
    \savedanchor\centerpoint
    {
        \pgf@x=0pt
        \pgf@y=0pt
    }
    \anchor{center}{\centerpoint}
    \anchorborder{\centerpoint}
    \saveddimen\overallwidth
    {
        \pgf@x=3\wd\pgfnodeparttextbox
        \ifdim\pgf@x<\minimumstatewidth
            \pgf@x=\minimumstatewidth
        \fi
    }
    \savedanchor{\upperrightcorner}
    {
        \pgf@x=.5\wd\pgfnodeparttextbox
        \pgf@y=.5\ht\pgfnodeparttextbox
        \advance\pgf@y by -.5\dp\pgfnodeparttextbox
    }
    \anchor{north}
    {
        \pgf@x=0pt
        \pgf@y=0pt
    }
    \anchor{south}
    {
        \pgf@x=0pt
        \pgf@y=0pt
    }
    \anchor{A}
    {
        \pgf@x=-\overallwidth
        \divide\pgf@x by 4
        \pgf@y=0pt
    }
    \anchor{B}
    {
        \pgf@x=\overallwidth
        \divide\pgf@x by 4
        \pgf@y=0pt
    }
    \anchor{text}
    {
        \upperrightcorner
        \pgf@x=-\pgf@x
        \ifhflip
            \pgf@y=-\pgf@y
            \advance\pgf@y by 0.4\stateheight
        \else
            \pgf@y=-\pgf@y
            \advance\pgf@y by -0.4\stateheight
        \fi
    }
    \backgroundpath
    {
       \begin{pgfonlayer}{foreground}
        \pgfsetstrokecolor{black}
        \pgfsetlinewidth{\thickness}
        \pgfpathmoveto{\pgfpoint{-0.5*\overallwidth}{0}}
        \pgfpathlineto{\pgfpoint{0.5*\overallwidth}{0}}
        \ifhflip
            \pgfpathlineto{\pgfpoint{0}{\stateheight}}
        \else
            \pgfpathlineto{\pgfpoint{0}{-\stateheight}}
        \fi
        \pgfpathclose
        \ifblack
            \pgfsetfillcolor{black!50}
            \pgfusepath{fill,stroke}
        \else
            \pgfusepath{stroke}
        \fi
       \end{pgfonlayer}
    }
}
\makeatother

\newcommand{\tinymult}[1][dot]{
\smash{\raisebox{-1pt}{\hspace{-2pt}\ensuremath{\begin{pic}[scale=0.33,yscale=-1]
    \node (0) at (0,0) {};
    \node[#1, inner sep=1.5pt] (1) at (0,0.55) {};
    \node (2) at (-0.5,1) {};
    \node (3) at (0.5,1) {};
    \draw (0.center) to (1.center);
    \draw (1.center) to [out=left, in=down, out looseness=1.5] (2.center);
    \draw (1.center) to [out=right, in=down, out looseness=1.5] (3.center);
\end{pic}
}\hspace{-3pt}}}}
\newcommand{\tinyunit}[1][dot]{
\smash{\raisebox{-1pt}{\ensuremath{\hspace{0pt}\begin{pic}[scale=0.33,yscale=-1]
        \node (0) at (0,0) {};
        \node (1) at (0,1) {};
        \node[#1, inner sep=1.5pt] (d) at (0,0.55) {};
        \draw (0.center) to (d.center);
    \end{pic}
    \hspace{-1pt}}}}}

\newcommand{\tinyhandle}[1][dot]{\ensuremath{\smash{\raisebox{-1pt}{\ensuremath{\hspace{-2pt}\begin{pic}[scale=0.33]
        \node (0) at (0,0) {};
        \node[dot, inner sep=1.0pt] (1) at (0,0.3) {};
        \node[dot, inner sep=1.0pt] (2) at (0,0.7) {};
        \node (3) at (0,1) {};
        \draw (0.center) to (1.center);
        \draw (2.center) to (3.center);
        \draw[in=180, out=180, looseness=2] (1.center) to (2.center);
        \draw[in=0, out=0, looseness=2] (1.center) to (2.center);
\end{pic}\hspace{-1pt}}}}}}

\newcommand{\tinycup}{\begin{pic}[scale=0.17]
    \draw (0,0) to[out=-90,in=-90,looseness=2] (1.5,0);
\end{pic}}

\newcommand{\tinyground}[1][ground]{\begin{pic}[scale=0.4]
    \node[#1, scale=0.6] (1) at (0,0.4) {};
    \draw [pure] (1.south) to +(0,-.3);
\end{pic}
}

\makeatletter
\pgfdeclareshape{arrow shape}
{
    \savedanchor\centerpoint
    {
        \pgf@x=0pt
        \pgf@y=0pt
    }
    \anchor{center}{\centerpoint}
    \anchor{south}{\centerpoint}
    \anchor{north}{\centerpoint}
    \anchorborder{\centerpoint}
    \backgroundpath
    {
        \begin{pgfonlayer}{foreground}
        \pgfsetstrokecolor{black}
        \pgfsetarrowsstart{>[\strarr]}
        \pgfsetlinewidth{\thickness}
        \pgfpathmoveto{\pgfpoint{0}{-2pt}}
        \pgfpathlineto{\pgfpoint{0}{1pt}}
        \pgfusepath{stroke}
        \end{pgfonlayer}
    }
}
\pgfdeclareshape{double arrow shape}
{
    \savedanchor\centerpoint
    {
        \pgf@x=0pt
        \pgf@y=0pt
    }
    \anchor{center}{\centerpoint}
    \anchor{south}{\centerpoint}
    \anchor{north}{\centerpoint}
    \anchorborder{\centerpoint}
    \backgroundpath
    {
        \begin{pgfonlayer}{foreground}
        \pgfsetstrokecolor{black}
        \pgfsetarrowsstart{>[\strarr]}
        \pgfsetlinewidth{\thickness}
        \pgfpathmoveto{\pgfpoint{0}{-3.5pt}}
        \pgfpathlineto{\pgfpoint{0}{-0.5pt}}
        \pgfusepath{stroke}
        \pgfpathmoveto{\pgfpoint{0}{-0.5pt}}
        \pgfpathlineto{\pgfpoint{0}{2.5pt}}
        \pgfusepath{stroke}
        \end{pgfonlayer}
    }
}
\pgfdeclareshape{reverse arrow shape}
{
    \savedanchor\centerpoint
    {
        \pgf@x=0pt
        \pgf@y=0pt
    }
    \anchor{center}{\centerpoint}
    \anchor{south}{\centerpoint}
    \anchor{north}{\centerpoint}
    \anchorborder{\centerpoint}
    \backgroundpath
    {
        \begin{pgfonlayer}{foreground}
        \pgfsetstrokecolor{black}
        \pgfsetarrowsstart{<[\strarr]}
        \pgfsetlinewidth{\thickness}
        \pgfsetlinewidth{\thickness}
        \pgfpathmoveto{\pgfpoint{0}{-2pt}}
        \pgfpathlineto{\pgfpoint{0}{-1pt}}
        \pgfusepath{stroke}
        \end{pgfonlayer}
    }
}
\pgfdeclareshape{reverse double arrow shape}
{
    \savedanchor\centerpoint
    {
        \pgf@x=0pt
        \pgf@y=0pt
    }
    \anchor{center}{\centerpoint}
    \anchor{south}{\centerpoint}
    \anchor{north}{\centerpoint}
    \anchorborder{\centerpoint}
    \backgroundpath
    {
        \begin{pgfonlayer}{foreground}
        \pgfsetstrokecolor{black}
        \pgfsetarrowsstart{<[\strarr]}
        \pgfsetlinewidth{\thickness}
        \pgfpathmoveto{\pgfpoint{0}{-3.5pt}}
        \pgfpathlineto{\pgfpoint{0}{-0.5pt}}
        \pgfusepath{stroke}
        \pgfpathmoveto{\pgfpoint{0}{-0.5pt}}
        \pgfpathlineto{\pgfpoint{0}{2.5pt}}
        \pgfusepath{stroke}
        \end{pgfonlayer}
    }
}
\makeatother

\tikzset{arrow/.pic={
    \path[pic actions, decoration={ markings,
      mark=at position 1 with {\arrow{>[\strarr]}}
    },
    postaction={decorate}] (0,1pt) -- (0,2pt);
},double arrow/.pic={
    \path[pic actions, decoration={ markings,
      mark=at position \pgfdecoratedpathlength-3pt with {\arrow{>[\strarr]}},
      mark=at position \pgfdecoratedpathlength with {\arrow{>[\strarr]}}
    },
    postaction={decorate}] (0,-6pt) -- (0,4pt);
},
  reverse arrow/.pic={
    \path[pic actions, decoration={ markings,
      mark=at position 1 with {\arrow{<[\strarr]}}
    },
    postaction={decorate}] (0,1pt) -- (0,2pt);
}
}

\tikzset{surface picture/.style={xscale={0.6}, yscale=0.8, line width=\thickness}}
\tikzset{three dimensional picture/.style={}}
\tikzset{morphismtwocell/.style={morphism, width=0.5cm, connect ne length=0cm, connect se length=0cm, connect nw length=0cm, connect sw length=0cm}}
\tikzset{nmorphismtwocell/.style={draw, minimum width=0.8cm, connect ne length=0cm, connect se length=0cm, connect nw length=0cm, connect sw length=0cm}}

\tikzset{dashback/.style={black!40}}

\tikzset{shade 1 local/.style={shade 1}}
\tikzset{shade 2 local/.style={shade 2}}

\usetikzlibrary{matrix}
\newcommand{\multcolor}{black!25}
\newcommand{\copycolor}{white}

\begin{document}

\title[Categorical composable cryptography: Extended Version]{Categorical composable cryptography:\texorpdfstring{\\}{ }extended version\rsuper*}
\titlecomment{{\lsuper*}This is an extended version of a FoSSaCS 2022 conference
paper~\cite{BK22}. This work was supported by the Air Force Office of Scientific
Research under award number FA9550-20-1-0375, Canada’s NFRF and NSERC, an
Ontario ERA, and the University of Ottawa’s Research Chairs program.}

\thanks{We wish to thank Ziad Chaoui, Diana Kessler, Riley Shahar, Fabian Wiesner and anonymous reviewers for useful comments.} 

\author[A.~Broadbent]{Anne Broadbent\lmcsorcid{0000-0003-1911-0093}}
\author[M.~Karvonen]{Martti Karvonen\lmcsorcid{0000-0002-8919-343X}}

\address{Department of Mathematics and Statistics, University of Ottawa, Canada}
\email{abroadbe@uottawa.ca, martti.karvonen@uottawa.ca}  

\begin{abstract}
\noindent We formalize the simulation paradigm of cryptography in terms of category theory and show that protocols secure against abstract attacks form a symmetric monoidal category, thus giving an abstract model of composable security definitions in cryptography. Our model is able to incorporate computational security, set-up assumptions and various attack models such as colluding or independently acting subsets of adversaries in a modular, flexible fashion. We conclude by using string diagrams to rederive the security of the one-time pad, correctness of Diffie-Hellman key exchange and no-go results concerning the limits of bipartite and tripartite cryptography, ruling out e.g., composable commitments and broadcasting. On the way, we exhibit two categorical constructions of resource theories that might be of independent interest: one capturing resources shared among multiple parties and one capturing resource conversions that succeed asymptotically.

This is a corrected version of the paper \url{https://arxiv.org/abs/2208.13232} published originally on December 18, 2023. 
\end{abstract}

\maketitle

\section{Introduction}

Modern cryptographic protocols are complicated algorithmic entities, and their security analyses are often no simpler than the protocols themselves. Given this complexity, it would be highly desirable to be able to design protocols and reason about them compositionally, \ie by breaking them down into smaller constituent parts. In particular, one would hope that combining protocols proven secure results in a secure protocol without need for further security proofs. However, this is not the case for stand-alone security notions that are common in cryptography. To illustrate such failures of composability, let us consider the history of quantum key distribution (QKD), as recounted in~\cite{PR14arxiv}: QKD was originally proposed in the 80s~\cite{BB84}. The first security proofs against unbounded adversaries followed a decade later~\cite{May96,BBB+00,SP00,May01}. However, since composability was originally not a concern, it was later realized that the original security definitions did not provide a good enough level of security~\cite{KRBM07}---they didn't guarantee security if the keys were to be actually used, since even a partial leak of the key would compromise the rest. The story ends on a positive note, as eventually a new security criterion was proposed, together with stronger proofs~\cite{Ren05,BHL+05}.

In this work we initiate a categorical study of composable security definitions in cryptography. In the viewpoint developed here one thinks of cryptography as a resource theory: cryptographic functionalities (e.g.~secure communication channels) are viewed as resources and cryptographic protocols let one transform some starting resources to others. For instance, one can view the one-time-pad as a protocol that transforms an authenticated channel and a shared secret key into a secure channel. For a given protocol, one can then study whether it is secure against some (set of) attack model(s), and protocols secure against a fixed set of models can always be composed sequentially and in parallel.

This is in fact the viewpoint taken in constructive cryptography~\cite{Mau11}, which also develops the one-time-pad example above in more detail. However~\cite{Mau11} does not make a formal connection to resource theories as usually understood, whether as in quantum physics~\cite{horodecki:resource,chitambar:resource}, or more generally as defined in order theoretic~\cite{Fritz2015} or categorical~\cite{CFS16} terms. Instead, constructive cryptography is usually combined with abstract cryptography~\cite{MR11} which is formalized in terms of a novel algebraic theory of systems~\cite{MMP+18}.

Our work can be seen as a particular formalization of the ideas behind constructive cryptography, or alternatively as giving a categorical account of the real-world-ideal-world paradigm (also known as the simulation paradigm~\cite{L17}), which underlies more concrete frameworks for composable security, such as universally composable cryptography~\cite{Can01} and others~\cite{PW00,BPW04,BPW07,MT13,HS15,LHM19,KTR20}. We will discuss these approaches and abstract and constructive cryptography in more detail in Section~\ref{sec:relatedwork}

Our long-term goal is to enable cryptographers to reason about composable security at the same level of formality as stand-alone security, \emph{without having to fix all the details of a machine model nor having to master category theory}. Indeed, our current results already let one define multipartite protocols and security against arbitrary subsets of malicious adversaries \emph{in any symmetric monoidal category $\CC$}. Thus, as long as one's model of interactive computation results in a symmetric monoidal category, or more informally, one is willing to use pictures such as Figure~\ref{fig:attackonprod} to depict connections between computational processes without further specifying the order in which the picture was drawn, one can use the simulation paradigm to reason about multipartite security against malicious participants composably---and specifying finer details of the computational model is only needed to the extent that it affects the validity of one's argument. Moreover, as our attack models and composition theorems are fairly general, we hope that more refined models of adversaries can be incorporated.

We now highlight our contributions to cryptography:
\begin{itemize}
  \item We show how to adapt resource theories as categorically formulated~\cite{CFS16} in order to reason abstractly about \emph{secure} transformations between resources. This is done in Section~\ref{sec:crypto} by formalizing the simulation paradigm in terms of an abstract attack model (Definition~\ref{def:attack}), designed to be general enough to capture standard attack models of interest (and more) while still structured enough to guarantee composability. This section culminates in Corollary~\ref{cor:simultaneoussafety}, which shows that for any fixed set of attack models, the class of protocols secure against each of them results in a symmetric monoidal category. In Theorem~\ref{thm:perfectlifting} we observe that under suitable conditions, images of secure protocols under monoidal functors remain secure, which gives an abstract variant of the lifting theorem~\cite[Theorem 15]{Unr10} that states that perfectly UC-secure protocols are quantum UC-secure.
  \item We adapt this framework to model \emph{computational security} in Section~\ref{sec:extensions} in two ways: either by replacing equations with an equivalence relation, abstracting the idea of computational indistinguishability, or by working with a notion of distance. In the case of a distance, one can then either explicitly bound the distance between desired and actually achieved behavior, or work with sequences of protocols that converge to the target in the limit: the former models working in the finite-key regimen~\cite{TLGR12} and the latter models the kinds of asymptotic security and complexity statements that are common in cryptography. In the former case we show that errors compose additively in Lemma~\ref{lem:epsilonsafe}, and in Theorem~\ref{thm:metriccomposition} and in Corollary~\ref{cor:metricsimultaneoussafety} we show that protocols that are correct in the limit can be composed at will.
  \item We then apply the framework developed to study bipartite and tripartite cryptography. We begin by giving a \emph{purely pictorial} proof of the security of the one-time pad in Section~\ref{sec:otp}, valid for any Hopf algebra in any symmetric monoidal category. We then discuss the Diffie-Hellman key exchange in Section~\ref{sec:dhke}. In Section~\ref{sec:no-go}, we reprove the no-go-theorems of~\cite{PR08,MR11,MMP+18} concerning two-party commitments (resp. three-party broadcasting) in our setting, and reinterpret them as limits on what can be achieved securely in any compact closed category (resp. symmetric monoidal category). The key steps of the proof are done graphically, thus opening the door for cryptographers to use such pictorial representations as rigorous tools rather than merely as illustrations.
  \item We conclude by discussing choice of a model in Section~\ref{sec:models} and further questions in Section~\ref{sec:outlook}.
\end{itemize}
Moreover, we discuss some categorical constructions capturing aspects of resource theories appearing in the physics literature. These contributions may be relevant for further categorical studies on resource theories, independently of their usage here.
  \begin{itemize}
    \item In~\cite{CFS16} it is observed that many resource theories arise from an inclusion $\cat{C}_F\hookrightarrow\cat{C}$ of free transformations into a larger monoidal category, by taking the resource theory of states. We observe that this amounts to applying the monoidal Grothendieck construction~\cite{moeller:monoidalgrothendieck} to the functor $\CF\to\CC\xrightarrow{\hom(I,-)}\cat{Set}$. This suggests applying this construction more generally to the composite of monoidal functors $F\colon\cat{D}\to\cat{C}$ and  $R\colon \cat{C}\to\Set$.
    \item In Example~\ref{ex:n-partite} we note that choosing $F$ to be the $n$-fold monoidal product $\CC^n\to\CC$ captures resources shared by $n$ parties and $n$-partite transformations between them.
    \item In Section~\ref{sec:metric} we model categorically situations where there is a notion of distance between resources, and instead of exact resource conversions one either studies approximate transformations or sequences of transformations that succeed in the limit.
    \item In Section~\ref{sec:set-up} we discuss a variant of a construction on monoidal categories, used in special cases in~\cite{fongetal:backprop} and discussed in more detail in~\cite{cruttwell2022categorical}, that allows one to declare some resources to be free and thus enlarge the set of possible resource conversions.
  \end{itemize}

\subsection{Related work}\label{sec:relatedwork}

We have already mentioned that cryptographers have developed a plethora of frameworks for composable security, such as universally composable cryptography~\cite{Can01}, reactive simulatability~\cite{PW00,BPW04,BPW07} and others~\cite{MT13,HS15,LHM19,KTR20}. Moreover, some of these frameworks have been adapted to the quantum setting~\cite{BM04arxiv,Unr10,MR09}. One might hence be tempted to think that the problem of composability in cryptography has been solved. However, it is fair to say that most mainstream cryptography is not formulated composably and that composable cryptography has yet to realize its full potential. Moreover, this proliferation of frameworks should be taken as evidence of the continued importance of the issue, and is in fact reflected by the existence of a recent Dagstuhl seminar on this matter~\cite{CKL+19}. Indeed, the aforementioned frameworks mostly consist of setting up fairly detailed models of interacting machines, which as an approach suffers from two drawbacks:
\begin{itemize}
  \item In order to be more realistic, the detailed models are often complicated, both to reason in terms of and to define, thus making practicing cryptographers less willing to use them. Perhaps more importantly, it is not always clear whether the results proven in a particular model apply more generally for other kinds of machines, whether those of a competing framework or those in the real world. It is true that the choice of a concrete machine model does affect what can be securely achieved---for instance, quantum cryptography differs from classical cryptography and similarly classical cryptography behaves differently in synchronous and asynchronous settings~\cite{BCG93,KMTZ13}. Nevertheless, one might hope that composable cryptography could be studied at a similar level of formality as complexity theory, where one rarely worries about the number of tapes in a Turing machine or of other low-level details of machine models.
  \item Changing the model slightly (to \eg model different kinds of adversaries or to incorporate a different notion of efficiency) often requires reproving ``composition theorems'' of the framework or at least checking that the existing proof is not broken by the modification.
\end{itemize}
In contrast to frameworks based on detailed machine models, there are two closely related top-down approaches to cryptography: constructive cryptography~\cite{Mau11} and its cousin abstract cryptography~\cite{MR11}. We are indebted to both of these approaches, and indeed our framework could be seen as formalizing the key idea of constructive cryptography---namely, cryptography as a resource theory---and thus occupying a similar space as abstract cryptography. A key difference is that constructive cryptography is usually instantiated in terms of abstract cryptography~\cite{MR11}, which in turn is based on a novel algebraic theory of systems~\cite{MMP+18}. However, our work is not merely a translation from this theory to categorical language, as there are important differences and benefits that stem from formalizing cryptography in terms of a well-established and well-studied algebraic theory of systems---that of (symmetric) monoidal categories:
  \begin{itemize}
    \item  The fact that cryptographers wish to compose their protocols \emph{sequentially and in parallel} strongly suggests using \emph{monoidal categories}, that have these composition operations as primitives. In our framework, protocols secure against a fixed set of attack models results in a symmetric monoidal category. In contrast, the algebraic theory of systems~\cite{MMP+18} on which abstract cryptography is based on takes parallel composition and internal wiring as its primitives. This design choice results in some technical kinks and tangles that are natural with any novel theory but have already been smoothed out in the case of category theory. For instance, in the algebraic theory of systems of~\cite{MMP+18} the parallel composition is a partial operation and in particular the parallel composite of a system with itself is never defined\footnote{While the suggested fix is to assume that one has ``copies'' of the same system with disjoint wire labels, it is unclear how one recognizes or even defines \emph{in terms of the system algebra} that two distinct systems are copies of each other.} and the set of wires coming out of a system is fixed once and for all\footnote{Indeed, while~\cite{portmann:causal} manages to bundle and unbundle ports along isomorphism when convenient, it seems like the chosen technical foundation makes this more of a struggle than it should be.}. In contrast, in a monoidal category parallel composition is a total operation and whether one draws a box with $n$ output wires of types $A_1,\dots A_n$ or single output wire of type $\bigotimes_{i=1}^n A_i$ is a matter of convenience. Technical differences such as these make a direct formal comparison or translation between the frameworks difficult, even if informally and superficially there are similarities.
    \item We do not abstract away from an attacker model, but rather make it an explicit part of the formalism that can be modified without worrying about composability. This makes it possible to consider and combine very easily different security properties, and in particular paves the way to model attackers with limited powers such as honest-but-curious adversaries. In our framework, one can first fix a protocol transforming some resource to another one, and then discuss whether this transformation is secure against different attack models. In contrast, in abstract cryptography a cryptographic resource is a tuple of functionalities, one for each set of dishonest parties, and thus has no prior existence before fixing the attack model. This makes the question ``what attack models is this protocol secure against?'' difficult to formalize.
    \item As category theory is de facto the lingua franca between several subfields of mathematics and computer science, elucidating the categorical structures present in cryptography opens up the door to further connections between cryptography and other fields. For instance, game semantics readily gives models of interactive, asynchronous and probabilistic (or quantum) computation~\cite{winskel:game,clairambault:gamesforquantum,clairambaultetal:gamesforquantum2} in which our theory can be instantiated, and thus further paves the way for programming language theory to inform cryptographic models of concurrency.
    \item Category theory comes with existing theory, results and tools that can readily be applied to questions of cryptographic interest. In particular the graphical calculi of symmetric monoidal and compact closed categories~\cite{Sel10} enables one to rederive impossibility results shown in~\cite{PR08,MR11,MMP+18} purely pictorially. In fact, such pictures were already used as heuristic devices that illuminate the official proofs in~\cite{PR08,MMP+18}, and viewing these pictures categorically lets us promote them from mere illustrations to rigorous yet intuitive proofs. Indeed, in~\cite[Footnote 27]{MR11} the authors suggest moving from a 1-dimensional symbolic presentation to a 2-dimensional one, and this is exactly what the graphical calculus already achieves.
  \end{itemize}

The approaches above result in a framework where security is defined so as to guarantee composability. In contrast, approaches based on various protocol logics~\cite{DMP01,DMP03,DDMP03b,DDMP03a,DDMP05,DDMR07} aim to characterize situations where composition can be done securely, even if one does not use composable security definitions throughout. As these approaches are based on process calculi, they are categorical under the hood~\cite{pavlovic1997,Mifsudetal:controlstructures} even if not overtly so. There is also earlier work explicitly discussing category theory in the context of cryptography~\cite{breiner:graphicaldicrypto,coecke:graphicalqkd,sunwang:graphicalbc,breiner:selftesting,heunen:qkd,hillebrand:superdense,coecke:environment,kissinger2017picture,stay:crypto,dusko:crypto,hines:diagrammaticrypto,pavlovic2012tracing}, but they concern stand-alone security of particular (kinds of) cryptographic protocols, rather than categorical aspects of composable security definitions.

\section{Background on monoidal categories and string diagrams}\label{sec:background}

We assume that the reader is familiar with category theory in general and with monoidal and compact closed categories in particular, so we will recall the main concepts very briefly, mostly to explain the notation and string diagrams used. General references for category theory include~\cite{maclane:categories,awodey:categorytheory,borceux:vol1,borceux:vol2,riehl:categorytheory,leinster:basicCT} and string diagrams are surveyed in~\cite{Sel10}. However, a working cryptographer might find it easier to consult texts which are written with some applications in mind and introduce string diagrams concurrently with categories, such as~\cite{coecke2010categories,FS19,heunenvicary:categories}.

Let $\CC$ be a symmetric monoidal category (SMC). Roughly speaking, this means that we have a class of objects $A,B,C,\dots$, and a class of morphisms $f,g,h,\dots$. We also have functions $\dom$ and $\cod$ that give us the domain and codomain of morphisms, and we write $f\colon A\to B$ to express that $A=\dom(f)$ and $B=\cod(f)$. Morphisms can be composed sequentially, \ie whenever $f\colon A\to B$ and $g\colon B\to C$ there is a morphism $g\circ f=gf\colon A\to C$. In addition, there is a monoidal product $\otimes$ on objects and morphisms, that sends $f\colon A\to B$ and $g\colon C\to D$ to $f\otimes g\colon A\otimes C\to B\otimes D$. For each object there should be an identity morphism $\id[A]\colon A\to A$, and there should be a special object $I$, called the tensor unit. This data is subject to some constraints: composition should be (strictly) associative and unital, and $\circ$ and $\otimes$ should cooperate in that the equations $(g\circ f)\otimes (j\circ h)=(g\otimes j)\circ (f\otimes h)$ and $\id[A\otimes B]=\id[A]\otimes \id[B]$ hold. Moreover, the monoidal product should be associative, commutative and unital \emph{up to coherent isomorphisms}, see~\cite[Section 6.1]{borceux:vol2} for the precise details. We will often assume that the variables $\CC$ and~$\DD$ denote strict SMCs, meaning that associativity and unitality of $\otimes$ holds up to equality. This is mainly for notational convenience---first, any SMC is equivalent to a strict one and second, the theory we put forward could be developed without assuming strictness at the cost of some notational overhead. As an example of a (non-strict) SMC the reader could think \eg of the category $\Set$ of sets and functions between them, with the monoidal structure given by cartesian product, or the category $\cat{Vect}_\mathbb{R}$ of real vector spaces and linear maps between them, with the monoidal structure given by tensor product.

The tersely sketched structure of an SMC is naturally internalized in the \emph{graphical calculus} we use, which provides a sound and complete method for reasoning about them. Thus the reader less familiar with SMCs is invited to trust their visual intuition as it is unlikely to lead them astray. In this graphical calculus, we will denote a morphism $f\colon A\to B$ as \[\begin{pic}
  \node[morphism,font=\tiny] (f) at (0,0) {$f$};
  \draw (f.south) to ++(0,-.5)  node[right] {$A$};
  \draw (f.north) to ++(0,.5)  node[right] {$B$};
\end{pic}\] and composition and monoidal product as
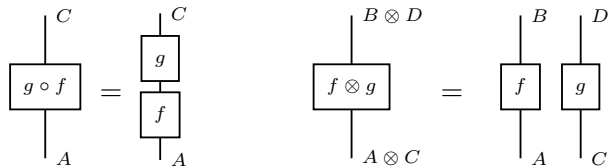
\[
  \begin{pic}
    \node[morphism] (f) {$g \circ f$};
    \draw (f.south) to ++(0,-.65) node[right] {$A$};
    \draw (f.north) to ++(0,.65) node[right] {$C$};
  \end{pic}
  \ =\
  \begin{pic}
    \node[morphism] (g) at (0,.75) {$g\vphantom{f}$};
    \node[morphism] (f) at (0,0) {$f$};
    \draw (f.south) to ++(0,-.3) node[right] {$A$};
    \draw (g.south) to  (f.north);
    \draw (g.north) to ++(0,.3) node[right] {$C$};
  \end{pic}
  \qquad\qquad
  \begin{pic}
    \node[morphism] (f) {$f \otimes g$};
    \draw (f.south) to ++(0,-.65) node[right] {$A \otimes C$};
    \draw (f.north) to ++(0,.65) node[right] {$B \otimes D$};
  \end{pic}
  =\quad
  \begin{pic}
    \node[morphism] (f) at (-.4,0) {$f$};
    \node[morphism] (g) at (.4,0) {$g\vphantom{f}$};
    \draw (f.south) to ++(0,-.65) node[right] {$A$};
    \draw (f.north) to ++(0,.65) node[right] {$B$};
    \draw (g.south) to ++(0,-.65) node[right] {$C$};
    \draw (g.north) to ++(0,.65) node[right] {$D$};
  \end{pic}
\]
Special morphisms get special pictures: identities and symmetries are depicted as
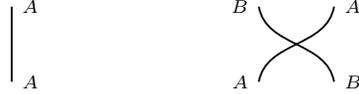
\[\begin{pic}
    \draw (0,0) node[right] {$A$} to (0,1) node[right] {$A$};
  \end{pic}
  \qquad\qquad\qquad
  \begin{pic}
    \draw (0,0) node[left] {$A$} to[out=80,in=-100] (1,1) node[right] {$A$};
    \draw (1,0) node[right] {$B$} to[out=100,in=-80] (0,1) node[left] {$B$};
  \end{pic}
\]
whereas the identity on the tensor unit is denoted by the empty picture. In general, a morphism might have multiple input/output wires, so that $f\colon A_1\otimes \dots\otimes A_n\to B_1\otimes\dots\otimes B_m$ is drawn as
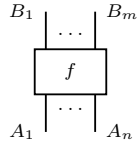
\[
  \begin{pic}
    \setlength\minimummorphismwidth{6mm}
    \node[morphism] (f) at (0,0) {$f$};
    \node (a) at (0,.5) {$\ldots$};
    \node (b) at (0,-.5) {$\ldots$};
    \draw ([xshift=-2.5pt]f.south west) to ++(0,-.5) node[left] {$A_1$};
    \draw ([xshift=-2.5pt]f.north west) to ++(0,.5) node[left] {$B_1$};
    \draw ([xshift=2.5pt]f.south east) to ++(0,-.5) node[right] {$A_n$};
    \draw ([xshift=2.5pt]f.north east) to ++(0,.5) node[right] {$B_m$};
  \end{pic}\]
In particular a morphism $I\to A_1\otimes\dots\otimes A_n$ will have no incoming wires. We will call such morphisms \emph{states} on $A_1\otimes\dots\otimes A_n$ and depict them as triangles instead of boxes:
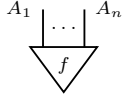
\[
    \begin{pic}
    \node[state] (f) at (0,0) {$f$};
    \node (a) at (0,.25) {$\ldots$};
    \draw ([xshift=-1.5pt]f.A) to ++(0,.5) node[left] {$A_1$};
    \draw ([xshift=1.5pt]f.B) to ++(0,.5) node[right] {$A_n$};
  \end{pic}
  \]
Note that the property $\id[A\otimes B]=\id[A]\otimes \id[B]$ becomes
\[\begin{pic}
    \draw (0,0) node[right] {$A\otimes B$} to (0,1) node[right] {$A\otimes B$};
  \end{pic}
  \qquad=\qquad
  \begin{pic}
    \draw (0,0) node[left] {$A$} to (0,1) node[left] {$A$};
    \draw (1,0) node[right] {$B$} to (1,1) node[right] {$B$};
  \end{pic}
\]
so that whether multiple wires are packaged into one or not is largely a matter of convenience. We will often omit labeling wires with the name of the object unless necessary, and at times the label will only give partial information.

For Theorem~\ref{thm:bipartite} we will assume that our ambient category~$\CC$ is in fact a \emph{compact closed category}. This means that $\CC$ is an SMC, and we are also given for every object $A$ an object $A^*$ and morphisms
  
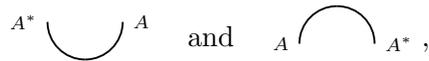
\[
  \begin{pic}[yscale=-1]
   \draw (0,1) node[left] {$A^*$} to[out=90,in=180] (.5,1.5);
   \draw (.5,1.5) to[out=0,in=90] (1,1)  node[right] {$A$};
  \end{pic}
   \quad \text{and}\quad
   \begin{pic}
   \draw (0,1) node[left] {$A$} to[out=90,in=180] (.5,1.5);
   \draw (.5,1.5)  to[out=0,in=90] (1,1) node[right] {$A^*$};
  \end{pic},
 \]
called cups and caps respectively, satisfying
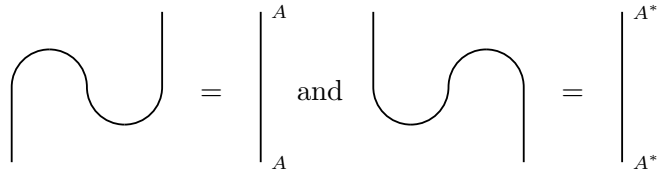
\[
  \begin{pic}
   \draw (0,0) to (0,1) to[out=90,in=180] (.5,1.5);
   \draw (.5,1.5) to[out=0,in=90] (1,1) to[out=-90,in=180] (1.5,.5);
   \draw (1.5,.5) to[out=0,in=-90] (2,1) to (2,2);
  \end{pic}
  \quad = \quad \begin{pic}
    \draw (0,0) node[right] {$A$} to (0,2) node[right] {$A$};
  \end{pic}
  \text{and}
  \quad
  \begin{pic}[xscale=-1]
   \draw (0,0) to (0,1) to[out=90,in=180] (.5,1.5);
   \draw (.5,1.5) to[out=0,in=90] (1,1) to[out=-90,in=180] (1.5,.5);
   \draw (1.5,.5) to[out=0,in=-90] (2,1) to (2,2);
  \end{pic} \quad=\quad
  \begin{pic}
    \draw (0,0) node[right] {$A^*$} to (0,2) node[right] {$A^*$};
  \end{pic}
 \]
Informally, this somewhat blurs the distinction between input and output wires, as one expects to happen if the boxes represent interactive and open computational processes. In particular, morphisms $A\to B$ correspond bijectively to states on $A^*\otimes B$, where the bijection is given by bending and unbending wires, and this correspondence should be seen as the categorical counterpart to the Choi–Jamio\l{}kowski isomorphism from quantum information (see \eg~\cite[Section 3.1.2.]{coeckekissinger:picturing} or~\cite[Section 3.1]{heunenvicary:categories}).

We will briefly conclude this section by discussing functors between SMCs. A lax monoidal functor $\cat{C}\to\cat{D}$ between monoidal categories is a functor $F\colon\cat{C}\to\cat{D}$ equipped with natural maps $F(A)\otimes F(B)\to F(A\otimes B)$ and a morphism $I_\cat{D}\to F(I_\CC)$ subject to certain coherence equations that roughly say that it cooperates with the monoidal structures on $\cat{C},\cat{D}$ in a well-behaved manner. A strong monoidal functor is a lax monoidal one for which the structure maps $F(A)\otimes F(B)\to F(A\otimes B)$ and $I_\cat{D}\to F(I_\CC)$ are isomorphisms. A monoidal functor (in either sense) is symmetric if it additionally cooperates with the symmetries. We will use graphical calculus of strong monoidal functors in the proof of Theorem~\ref{thm:perfectlifting}, but otherwise do not refer to the detailed definitions nor use this graphical language, and hence we do not go into more detail here. Full definitions can be found \eg at~\cite[Section I.1.2]{leinster2004} or at~\cite[Section 6.4]{borceux:vol2}, and a graphical calculus for them is discussed in~\cite{mellies2006functorial}. For us, all functors will be symmetric and either strong or lax monoidal, and we will specify which we mean whenever it makes a difference.

\section{Resource theories}\label{sec:resourcetheories}

We briefly review the categorical viewpoint on resource theories of~\cite{CFS16}. Roughly speaking, a resource theory can be seen as an SMC but the change in terminology corresponds to a change in viewpoint: usually in category theory one studies global properties of a category, such as the existence of (co)limits, relationships to other categories, etc. In contrast, when one views a particular SMC $\CC$ as resource theory, one is interested in local questions. One thinks of objects of $\CC$ as resources, and morphisms as processes that transform a resource to another. From this point of view, one mostly wishes to understand whether $\hom_\CC(X,Y)$ is empty or not for resources $X$ and~$Y$ of interest. Thus from the resource-theoretic point of view, most of the interesting information in $\CC$ is already present in its preorder collapse. As concrete examples of resource-theoretic questions, one might wonder if
 \begin{itemize}
  \item some noisy channels can simulate a (almost) noiseless channel~\cite[Example 3.13.]{CFS16},
  \item there is a protocol that uses only local quantum operations and classical communication and transforms a particular quantum state to another one~\cite{Chitambaretal:LOCC},
  \item some non-classical statistical behavior can simulate some other such behavior~\cite{abramskyetal:comonadicview}.
 \end{itemize}
In~\cite{CFS16} the authors show how many familiar resource theories arise in a uniform fashion: starting from an SMC $\CC$ of processes equipped with a wide sub-SMC $\CF$, the morphisms of which correspond to ``free'' processes, they build several resource theories (=SMCs). Here, the ``free'' processes should be understood as giving the morphisms that do not increase the resource at hand: for instance, deterministic processes form a natural choice of free operations in order to form a resource theory of randomness. Formally, the resource theory is specified by giving any $\CF\hookrightarrow \CC$ and then applying one of the constructions building a resource theory out of it.
Perhaps the most important of these constructions is the resource theory of states: given $\CF\hookrightarrow\CC$, the corresponding resource theory of states can be explicitly constructed by taking the objects of this resource theory to be states of $\CC$, \ie maps $r\colon I\to A$ for some $A$, and maps $r\to s$ are maps $f\colon A\to B$ in $\CF$ that transform $r$ to $s$ as in Figure~\ref{fig:state_transform}.

We now turn our attention towards cryptography. As contemporary cryptography is both broad and complex in scope, any faithful model of it is likely to be complicated as well. A benefit of the categorical idiom is that we can build up to more complicated models in stages, which is what we will do in the sequel. We phrase our constructions in terms of an arbitrary SMC $\CC$, but in order to model actual cryptographic protocols, the morphisms of $\CC$ should represent interactive computational machines with open ``ports'', with composition then amounting to connecting such machines together. Different choices of $\CC$ set the background for different kinds of cryptography, so that quantum cryptographers want $\CC$ to include quantum systems whereas in classical cryptography it is sufficient that these computational machines are probabilistic. Constructing such categories $\CC$ in detail is not trivial but is outside our scope---we will discuss this in more detail in Section~\ref{sec:outlook}.

Our first observation is that there is no reason to restrict to inclusions $\CF\hookrightarrow\CC$ in order to construct a resource theory of states. Indeed, while it is straightforward to verify explicitly that the resource theory of states is a symmetric monoidal category, it is instructive to understand more abstractly why this is so: in effect, the constructed category is something known as the category of elements of the composite functor $\CF\to\CC\xrightarrow{\hom(I,-)}\cat{Set}$. Moreover, this composite is lax symmetric monoidal, and it has been proven abstractly that the category of elements of any lax symmetric monoidal functor is symmetric monoidal~\cite{moeller:monoidalgrothendieck}. Thus this construction goes through for any symmetric (lax) monoidal functor into $\Set$. In this work, such functors will arise as composites $\cat{D}\xrightarrow{F}\cat{C}\xrightarrow{R}\Set$ with $\CC$ equipped with further structure in order to discuss security. Here it may be helpful to think of $F$ as interpreting free processes into an ambient category of all processes, and $R\colon\CC\to\cat{Set}$ as an operation that gives for each object $A$ of $\CC$ the set $R(A)$ of resources of type~$A$.

Explicitly, given symmetric monoidal functors $\cat{D}\xrightarrow{F}\cat{C}\xrightarrow{R}\Set$, the category of elements $\res RF$ has as its objects pairs $(A,r)$ where $A$ is an object of $\cat{D}$ and $r\in RF(A)$, the intuition being that $r$ is a resource of type $F(A)$. A morphism $(A,r)\to (B,s)$ is given by a morphism $f\colon A\to B$ in $\cat{D}$ that takes $r$ to $s$, \ie satisfies $RF(f)(r)=s$. The symmetric monoidal structure comes from the symmetric monoidal structures of $\cat{D}, \Set$ and $RF$. Somewhat more explicitly, $(A,r)\otimes (B,s)$ is defined by $(A\otimes B,r\otimes s)$ where $r\otimes s$ is the image of $(r,s)$ under the function $RF(A)\times RF(B)\to RF(A\otimes B)$ that is part of the monoidal structure on $RF$, and on morphisms of $\res RF$ the monoidal product is defined from that of $\cat{D}$.

From now on we will assume that $F$ is strong monoidal, and while $R=\hom(I,-)$ captures our main examples of interest, we will phrase our results for an arbitrary lax monoidal $R$. Allowing $F$ to be an arbitrary strong monoidal functor instead of an inclusion lets us capture the $n$-partite structure often used when studying cryptography, as shown next.
\begin{exa}\label{ex:n-partite}
Consider the resource theory induced by $\CC^n\xrightarrow{\otimes}\CC\xrightarrow{\hom(I,-)}\Set$, where we write $\bigotimes$ for the $n$-fold monoidal product\footnote{As $\CC$ is symmetric, the functor $\bigotimes$ is strong monoidal.}. The resulting resource theory has a natural interpretation in terms of $n$ agents trying to transform resources to others: an object of this resource theory corresponds to a pair $((A_i)_{i=1}^n,r\colon I\to \bigotimes A_i)$, and can be thought of as an $n$-partite state, depicted in Figure~\ref{fig:n-partite_state}, where the $i$th agent has access to a port of type $A_i$. A morphism $\f=(f_1,\dots f_n)\colon ((A_i)_{i=1}^n,r)\to ((B_i)^n_{i=1},s)$ between such resources then amounts to a protocol that prescribes, for each agent $i$ a process $f_i$ that they should perform so that~$r$ gets transformed to $s$ as in Figure~\ref{fig:n-partite_state_transformation}.
\begin{figure}
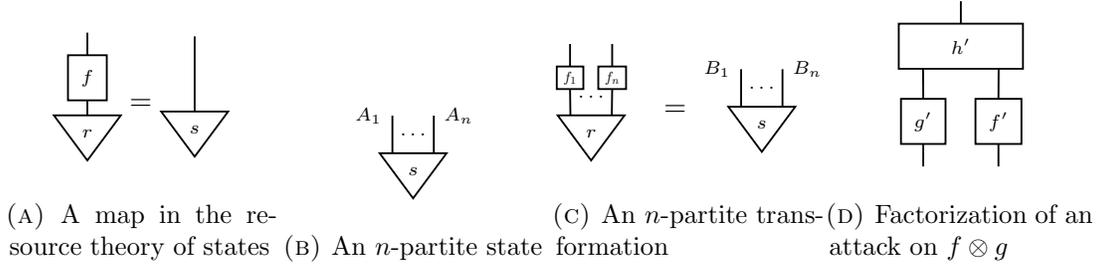

    \centering
\begin{subfigure}[b]{0.23\textwidth}
\[\begin{pic}
    \node[morphism] (f) at (0,.5) {$f$};
    \node[state] (x) at (0,0) {$r$};
    \draw (f.south) to (x.north);
    \draw (f.north) to ++(0,.3);
  \end{pic}=\begin{pic}
  \node[state] (y) at (0,0) {$s$};
  \draw (y.north) to ++(0,1);
\end{pic}\]
  \caption{A map in the resource theory of states}
  \label{fig:state_transform}
  \end{subfigure}
  ~
    \begin{subfigure}[b]{0.23\textwidth}
\[
    \begin{pic}
    \node[state] (f) at (0,0) {$r$};
    \node (a) at (0,.25) {$\ldots$};
    \draw ([xshift=-1.5pt]f.A) to ++(0,.5) node[left] {$A_1$};
    \draw ([xshift=1.5pt]f.B) to ++(0,.5) node[right] {$A_n$};
  \end{pic}
  \]
        \caption{An $n$-partite state}
        \label{fig:n-partite_state}
    \end{subfigure}
    ~
    \begin{subfigure}[b]{0.23\textwidth}
\[
  \begin{pic}
    \node[state] (r) at (0,0) {$r$};
    \node (a) at (0,.25) {$\ldots$};
    \node[morphism,scale=.5,font=\normalsize] (f) at (-.28,.5) {$f_1$};
    \node[morphism,scale=.5,font=\normalsize] (g) at (.28,.5) {$f_n$};
    \draw ([xshift=-1.5pt]r.A) to (f.south);
    \draw ([xshift=1.5pt]r.B) to (g.south);
    \draw (f.north) to ++(0,.3);
    \draw (g.north) to ++(0,.3);
  \end{pic}\quad=
    \begin{pic}
    \node[state] (f) at (0,0) {$s$};
    \node (a) at (0,.25) {$\ldots$};
    \draw ([xshift=-1.5pt]f.A) to ++(0,.5) node[left] {$B_1$};
    \draw ([xshift=1.5pt]f.B) to ++(0,.5) node[right] {$B_n$};
  \end{pic}
  \]
        \caption{An $n$-partite transformation}
        \label{fig:n-partite_state_transformation}
    \end{subfigure}
   ~\begin{subfigure}[b]{0.23\textwidth}
   \[\begin{pic}
  \node[morphism] (g) at (0,0) {$g'$};
  \node[morphism] (f) at (1,0) {$f'$};
  \setlength\minimummorphismwidth{13mm}
  \node[morphism] (h) at (0.5,1) {$h'$};
  \draw (g.north) to ([xshift=0.5pt]h.south west);
  \draw (f.north) to ([xshift=-0.5pt]h.south east);
  \draw (h.north) to ++(0,.3);
  \draw (g.south) to ++(0,.-.3);
  \draw (f.south) to ++(0,.-.3);
\end{pic}
\]
        \caption{Factorization of an attack on $f\otimes g$}
        \label{fig:attackonprod}
   \end{subfigure}
   \caption{Some resource transformations}
\end{figure}
\end{exa}
In this resource theory, all of the agents are equally powerful and can perform all processes allowed by $\CC$, and this might be unrealistic: first of all, $\CC$ might include computational processes that are too powerful/expensive for us to use in our cryptographic protocols. Moreover, having agents with different computational powers is important to model \eg blind quantum computing~\cite{BFK09} where a client with access only to limited, if any, quantum computation tries to securely delegate computations to a server with a powerful quantum computer. This limitation is easily remedied: we could take the $i$th agent to be able to implement computations in some sub-SMC $\CC_i$ of $\CC$, and then consider $\prod_{i=1}^n \CC_i\to\CC$.

A more serious limitation is that such transformations have no security guarantees---they only work if each agent performs~$f_i$ as prescribed by the protocol. We will fix this in the next section.

First we address another limitation: for a given $\CC$, one might want to view~\emph{morphisms of $\CC$} as resources, instead of just working with the resource theory of states. This can be achieved by ``the resource theory of universally-combinable processes'' of~\cite[Section 3.4]{CFS16}, which we will now show to arise as the resource theory of states of another category.

 \begin{defi}\label{def:ncomb} Given an SMC $\CC$, the category $\ncomb(\CC)$ is defined as follows: objects of $\ncomb(\CC)$ are finite lists $(A_i,B_i)_{i=1}^m$ of pairs of objects of $\CC$. Morphisms are defined in two stages: A morphism $(A_i,B_i)_{i=1}^m\to (C,D)$ is given by permutation $\sigma\colon \{1,\dots , m\}\to\{1,\dots , m\}$ and an $m$-comb
 \[\begin{pic}
   \node[dashedmorphism] (f) at (-0.4,0) {};
   \node[dashedmorphism] (e) at (-0.4,3.75) {};
   \setlength\minimummorphismwidth{10mm}
   \node (a) at (-0.4,2) {$\vdots$};
   \node (b) at (.4,2) {$\vdots$};
   \node[morphism] (g) at (0,1) {$g_1$};
   \node[morphism] (h) at (0,-1) {$g_0$};
   \node[morphism] (i) at (0,2.75) {$g_{m-1}$};
   \node[morphism] (j) at (0,4.75) {$g_{m}$};
   \draw (g.south west) to node[left] {$B_{\sigma(1)}$} (f.north);
   \draw (h.north west) to node[left] {$A_{\sigma(1)}$} (f.south);
   \draw (g.south east) to (h.north east);
   \draw (g.north west) to ++(0,.3) node[left]{$A_{\sigma(2)}$};
   \draw (g.north east) to ++(0,.3);
   \draw (h.south) to ++(0,-.35) node[right] {$C$};
   \draw (i.south west) to node[left]{$B_{\sigma(m-1)}$} +(0,.-.3);
   \draw (i.south east) to ++(0,.-.3);
   \draw (i.north west) to node[left]{$A_{\sigma(m)}$}  (e.south);
   \draw (i.north east) to (j.south east);
   \draw (e.north) to node[left] {$B_{\sigma(m)}$} (j.south west);
   \draw (j.north) to ++(0,.35) node[right] {$D$};
 \end{pic}\]
 in $\CC$. Formally, an $m$-comb is an equivalence class of tuples $(g_0,\dots g_m)$ of maps in $\CC$,  where $g_0\colon C\to A_{\sigma(1)}\otimes Y_1$, $g_i\colon B_{\sigma(i)}\otimes Y_i\to A_{\sigma(i+1)}\otimes Y_{i+1}$ for $i=1\dots m-1$ and $g_m\colon B_{\sigma(m)}\otimes Y_m\to D$ for some objects $Y_i$. Two such tuples are identified if, whenever one ``plugs the holes'' with maps of the form $Z_i\otimes A_{\sigma(i)}\to Z_i\otimes B_{\sigma(i)}$, the resulting maps in $\CC$ are equal.

 A morphism $(A_i,B_i)_{i=1}^m\to (C_j,D_j)_{j=1}^k$ is given by a function $f\colon \{1,\dots , m\}\to\{1,\dots , k\}$ and a morphism $(A_i,B_i)_{i\in f^{-1}(j)}\to (C_j,D_j)$ for each $j$. Composition is defined by nesting circuits into circuits, and the monoidal product is given by concatenation of lists.
 \end{defi}

 We will elaborate on the sequential composition in $\ncomb(\CC)$ in the following paradigmatic situation:  consider a morphism $(\sigma,g)\colon (A_i,B_i)_{i=1}^m\to (C,D)$ as depicted above, and morphisms $(\tau_i, h^i)\colon (X_i,Y_j)_{j=1}^{n_i}\to (A_{\sigma(i)},B_{\sigma(i)})$ for $i=1,\dots ,m$. The morphisms $(\tau_i, h^i)$ can be composed in parallel to result in a map to $(A_i,B_i)_{i=1}^m$. The circuit representing the composite of this map with our starting map $(A_i,B_i)_{i=1}^m\to (C,D)$ is depicted in Figure~\ref{fig:nested_combs}.

 \begin{figure}
 \[\begin{pic}
  \node[dashedmorphism,scale=.75] (f1) at (-0.2,1.55) {};
  \node[dashedmorphism,scale=.75] (e1) at (-0.2,3.7) {};
    \node[dashedmorphism,scale=.75] (f2) at (-0.2,-4.45) {};
  \node[dashedmorphism,scale=.75] (e2) at (-0.2,-2.3) {};
  \setlength\minimummorphismwidth{10mm}
  \node (a1) at (-0.2,2.7) {$\vdots$};
  \node (b1) at (.2,2.75) {$\vdots$};
  \node[morphism,scale=.5] (g1) at (0,2.15) {\normalsize$h^m_1$};
  \node[morphism,scale=.5] (h1) at (0,1) {\normalsize$h^m_0$};
  \node[morphism,scale=.5] (i1) at (0,3.15) {\normalsize$h^m_{n_m-1}$};
  \node[morphism,scale=.5] (j1) at (0,4.25) {\normalsize$h^m_{n_m}$};
  \draw (g1.south west) to  (f1.north);
  \draw (h1.north west) to  (f1.south);
  \draw (g1.south east) to (h1.north east);
  \draw (g1.north west) to ++(0,.15);
  \draw (g1.north east) to ++(0,.15);
  \draw (i1.south west) to  ++(0,.-.15);
  \draw (i1.south east) to ++(0,.-.15);
  \draw (i1.north west) to  (e1.south);
  \draw (i1.north east) to (j1.south east);
  \draw (e1.north) to  (j1.south west);
  \node (a2) at (-0.2,-3.3) {$\vdots$};
  \node (b2) at (.2,-3.25) {$\vdots$};
  \node[morphism,scale=.5] (g2) at (0,-3.85) {\normalsize$h^1_1$};
  \node[morphism,scale=.5] (h2) at (0,-5) {\normalsize$h^1_0$};
  \node[morphism,scale=.5] (i2) at (0,-2.85) {\normalsize$h^1_{n_1-1}$};
  \node[morphism,scale=.5] (j2) at (0,-1.75) {\normalsize$h^1_{n_1}$};
  \draw (g2.south west) to  (f2.north);
  \draw (h2.north west) to  (f2.south);
  \draw (g2.south east) to (h2.north east);
  \draw (g2.north west) to ++(0,.15);
  \draw (g2.north east) to ++(0,.15);
  \draw (i2.south west) to ++(0,.-.15);
  \draw (i2.south east) to ++(0,.-.15);
  \draw (i2.north west) to   (e2.south);
  \draw (i2.north east) to (j2.south east);
  \draw (e2.north) to (j2.south west);
  \node (a3) at (0,-.25) {$\vdots$};
  \node (b3) at (.8,-.25) {$\vdots$};
  \node[morphism] (g3) at (0.4,-1.1) {$g_1$};
  \node[morphism] (h3) at (0.4,-5.75) {$g_0$};
  \node[morphism] (i3) at (0.4,0.35) {$g_{m-1}$};
  \node[morphism] (j3) at (0.4,5) {$g_{m}$};
  \draw (j3.south west) to (j1.north);
  \draw (j3.south east) to (i3.north east);
  \draw (j3.north) to ++(0,.35) node[right] {$D$};
  \draw (h3.south) to ++(0,-.35) node[right] {$C$};
  \draw (g3.north west) to ++(0,.15);
  \draw (g3.north east) to ++(0,.15);
  \draw (g3.south east) to (h3.north east);
  \draw (g3.south west) to (j2.north);
  \draw (i3.south west) to ++(0,.-.15);
  \draw (i3.south east) to ++(0,.-.15);
  \draw (i3.north west) to (h1.south);
  \draw (h3.north west) to (h2.south);
\end{pic}\]
\caption{Nested combs arising when composing morphisms in $\ncomb(\CC)$}
\label{fig:nested_combs}
\end{figure}

\begin{rem}
As remarked in~\cite[Section 3.4]{CFS16}, these $n$-combs form naturally a symmetric multicategory, a kind of category-like structure where the domain of a morphism is a list of objects but the codomain is always a single object (see \eg\cite[Section 2.2.21]{leinster2004} for the precise definition). The SMC $\ncomb(\CC)$ arises by a standard method of turning a symmetric multicategory into an SMC.
\end{rem}

The tensor unit of $\ncomb(\CC)$ is given by the empty sequence, and a state of type $(A,B)$ is just a morphism $A\to B$ in $\CC$. Thus a state of type $(A_i,B_i)_{i=1}^m$ consists of a list of morphisms $(f_i\colon A_i\to B_i)_{i=1}^m$ in $\CC$. Given a subcategory $\CF$ of $\CC$, the resource theory induced by $\ncomb(\CF)\hookrightarrow\ncomb(\CC)\xrightarrow{\hom(I,-)}\Set$ is the ``resource theory of universally-combinable processes'' of~\cite[Section 3.4]{CFS16}: its objects are now lists of morphisms in $\CC$, and a resource conversion between two such lists consists of a map in $\ncomb(\CF)$ that converts the first list to the second one. In particular, the maps $g_i$ appearing in such $n$-combs have to be maps of $\CF$.

We can add a multipartite structure to this resource theory just like we did for the resource theory of states in Example~\ref{ex:n-partite}. The $k$-fold tensor product $\CC^k\to \CC$ induces a monoidal functor $\ncomb(\CC^k)\to \ncomb(\CC)$ (and again, one could have a different subcategory of $\CC$ for each of the parties). In the resource theory induced by  $\ncomb(\CC^k)\to\ncomb(\CC)\xrightarrow{\hom(I,-)}\Set$, resources are now lists of morphisms, each equipped with the structure of a $k$-fold tensor product on its domain and codomain. Intuitively, such a morphism corresponds to a box shared by the $k$ parties, each with a specified (possibly trivial) input and output port. When $k=3$ and the three parties are labeled Eve, Alice and Bob, any map of the form
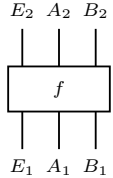
\[
  \begin{pic}
    \setlength\minimummorphismwidth{10mm}
    \node[morphism] (f) at (0,0) {$f$};
    \draw ([xshift=-2.5pt]f.south west) to ++(0,-.5) node[below] {$E_1$};
    \draw ([xshift=-2.5pt]f.north west) to ++(0,.5) node[above] {$E_2$};
    \draw (f.south) to ++(0,-.5) node[below] {$A_1$};
    \draw (f.north) to ++(0,.5) node[above] {$A_2$};
    \draw ([xshift=2.5pt]f.south east) to ++(0,-.5) node[below] {$B_1$};
    \draw ([xshift=2.5pt]f.north east) to ++(0,.5) node[above] {$B_2$};
  \end{pic}\]
is an example of such a resource, and general resources in this theory are then tuples of such maps. As another example,
a channel letting Alice broadcast any message over a fixed message space to Eve and Bob as in Figure~\ref{fig:insecure_channel} is an example of such a resource, where now some of the input and output ports of the parties are trivial (\ie equal to the tensor unit) and hence not drawn.

In this resource theory, all resources are ``single-use'' by design. This is both true for the combs, which must use each of the resources going into them exactly once, and for the general maps, which partition the starting resources to be used for each target resource. One can modify the construction of $\ncomb(\CC)$ so as to result in a cartesian monoidal category   $\ncombcart(\CC)$ (\ie a monoidal category where the monoidal product is given by the categorical product). In the resulting resource theory, all the resources are then usable arbitrarily many times. We will give the definition below, but won't use this in the sequel. The reason for this is that the protocols we study use their starting resources exactly once and hence fit already in $\ncomb(\CC)$. However, there is a natural embedding $\ncomb(\CC)\to\ncombcart(\CC)$, and using this one could view our protocols as transformations in the resource theory of reusable processes, and our security claims carry over through this inclusion.

 \begin{defi} Given an SMC $\CC$, the category $\ncombcart(\CC)$ is defined as follows: objects of $\ncombcart(\CC)$ are finite lists $(A_i,B_i)_{i=1}^n$ of pairs of objects of $\CC$. Morphisms are defined in two stages: A morphism $(A_i,B_i)_{i=1}^n\to (C,D)$ is given by a \emph{function}  $\sigma\colon \{1,\dots , m\}\to\{1,\dots , n\}$ for some $m\in \N$ and an $m$-comb
 \[\begin{pic}
   \node[dashedmorphism] (f) at (-0.4,0) {};
   \node[dashedmorphism] (e) at (-0.4,3.75) {};
   \setlength\minimummorphismwidth{10mm}
   \node (a) at (-0.4,2) {$\vdots$};
   \node (b) at (.4,2) {$\vdots$};
   \node[morphism] (g) at (0,1) {$g_1$};
   \node[morphism] (h) at (0,-1) {$g_0$};
   \node[morphism] (i) at (0,2.75) {$g_{m-1}$};
   \node[morphism] (j) at (0,4.75) {$g_{m}$};
   \draw (g.south west) to node[left] {$B_{\sigma(1)}$} (f.north);
   \draw (h.north west) to node[left] {$A_{\sigma(1)}$} (f.south);
   \draw (g.south east) to (h.north east);
   \draw (g.north west) to ++(0,.3) node[left]{$A_{\sigma(2)}$};
   \draw (g.north east) to ++(0,.3);
   \draw (h.south) to ++(0,-.35) node[right] {$C$};
   \draw (i.south west) to node[left]{$B_{\sigma(m-1)}$} +(0,.-.3);
   \draw (i.south east) to ++(0,.-.3);
   \draw (i.north west) to node[left]{$A_{\sigma(m)}$}  (e.south);
   \draw (i.north east) to (j.south east);
   \draw (e.north) to node[left] {$B_{\sigma(m)}$} (j.south west);
   \draw (j.north) to ++(0,.35) node[right] {$D$};
 \end{pic}\]
 in $\CC$, where two such $m$-combs are identified as in Definition~\ref{def:ncomb}.

 A morphism $(A_i,B_i)_{i=1}^m\to (C_j,D_j)_{j=1}^k$ is given by a \emph{relation} $R\colon \{1,\dots , m\}\to\{1,\dots , k\}$ and a morphism $(A_i,B_i)_{i\in R^{-1}(j)}\to (C_j,D_j)$ for each $j$. Composition is defined by nesting circuits into circuits, and the monoidal product is given by concatenation of lists.
 \end{defi}

\section{Cryptography as a resource theory}\label{sec:crypto}

\subsection{Attack models and security}

\begin{figure}
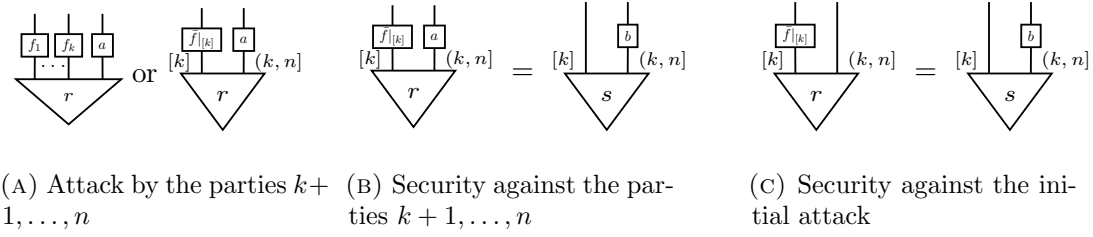

    \centering
\begin{subfigure}[b]{0.28\textwidth}
\[\begin{pic}[scale=0.9]
  	\setlength\minimumstatewidth{14mm}
    \node[state] (r) at (0,0) {$r$};
    \node (a) at (-.23,.2) {$\ldots$};
    \node[morphism,scale=.5,font=\normalsize] (f) at (-.5,.5) {$f_1$};
    \node[morphism,scale=.5,font=\normalsize] (h) at (0,.5) {$f_{k}$};
    \node[morphism,scale=.5,font=\normalsize] (g) at (.5,.5) {$a$};
    \draw ([xshift=-3.2pt]r.A) to (f.south);
    \draw ([xshift=3.2pt]r.B) to (g.south);
    \draw (f.north) to ++(0,.3);
    \draw (g.north) to ++(0,.3);
     \draw (h.north) to ++(0,.3);
    \draw (r.center) to (h.south);
  \end{pic}\text{ or}
\begin{pic}[scale=0.9]
    \node[morphism,scale=.5,font=\normalsize] (f) at (-.31,.5) {$\f|_{[k]}$};
     \node[morphism,scale=.5,font=\normalsize] (g) at (.31,.5) {$a$};
    \node[state,scale=1.25] (x) at (0,0) {$r$};
    \draw (f.south) to node[left] {$[k]$} (x.A);
    \draw (f.north) to ++(0,.3);
    \draw (g.south) to node[right] {$(k,n]$} (x.B);
    \draw (g.north) to ++(0,.3);
  \end{pic}\]
	\caption{Attack by the parties $k+1,\dots ,n$}
	\label{fig:attack}
	\end{subfigure}
	\enspace
    \begin{subfigure}[b]{0.28\textwidth}
\[\begin{pic}[scale=0.9]
    \node[morphism,scale=.5,font=\normalsize] (f) at (-.31,.5) {$\f|_{[k]}$};
     \node[morphism,scale=.5,font=\normalsize] (g) at (.31,.5) {$a$};
    \node[state,scale=1.25] (x) at (0,0) {$r$};
    \draw (f.south) to node[left] {$[k]$} (x.A);
    \draw (f.north) to ++(0,.3);
    \draw (g.south) to node[right] {$(k,n]$} (x.B);
    \draw (g.north) to ++(0,.3);
  \end{pic}=\begin{pic}
     \node[morphism,scale=.5,font=\normalsize] (g) at (.28,.5) {$b$};
    \node[state,scale=1.25] (x) at (0,0) {$s$};
    \draw (x.A) to node[left] {$[k]$} ++(0,.35) to ++(0,.6);
    \draw (g.south) to node[right] {$(k,n]$} (x.B);
    \draw (g.north) to ++(0,.3);
  \end{pic}\]
        \caption{Security against the parties $k+1,\dots ,n$}
        \label{fig:security}
    \end{subfigure}
    \qquad\enspace
    \begin{subfigure}[b]{0.28\textwidth}
\[\begin{pic}
    \node[morphism,scale=.5,font=\normalsize] (f) at (-.28,.5) {$\f|_{[k]}$};
    \node[state,scale=1.25] (x) at (0,0) {$r$};
    \draw (f.south) to node[left] {$[k]$} (x.A);
    \draw (f.north) to ++(0,.3);
    \draw (x.B) to node[right]  {$(k,n]$} ++(0,.35)  to ++(0,.6);
  \end{pic}=\begin{pic}
     \node[morphism,scale=.5,font=\normalsize] (g) at (.28,.5) {$b$};
    \node[state,scale=1.25] (x) at (0,0) {$s$};
    \draw (x.A) to node[left] {$[k]$} ++(0,.35) to ++(0,.6);
    \draw (g.south) to node[right] {$(k,n]$} (x.B);
    \draw (g.north) to ++(0,.3);
  \end{pic}\]
        \caption{Security against the initial attack}
        \label{fig:initial_attack}
    \end{subfigure}
\caption{Attacks and security constraints\vspace{-1em}}
\end{figure}
In order for a protocol $\f=(f_1,\dots ,f_n)\colon ((A_i)_{i=1}^n,r)\to ((B_i)^n_{i=1},s)$ to be secure, we should have some guarantees about what happens if, as a result of \emph{an attack} on the protocol, something else than $(f_1,\dots ,f_n)$ happens. For instance, some subset of the parties might deviate from the protocol and do something else instead. In the simulation paradigm~\cite{L17}, security is then defined by saying that, anything that could happen when running the real protocol, \ie $\f$ with $r$, could also happen in the ideal world, \ie with $s$. A given protocol might be secure against some kinds of attacks and insecure against others, so we define security against an abstract attack model. This abstract notion of an attack model is one of the main definitions of our paper. It isolates conditions needed for the composition theorem (Theorem ~\ref{thm:composition}). It also captures our key examples that we use to illustrate the definition after giving it.
\begin{defi}\label{def:attack} An \emph{attack model} $\A$ on an SMC $\CC$ consists of giving for each morphism $f$ of $\CC$ a class $\A(f)$ of morphisms of $\CC$ such that
  \begin{enumerate}[(i)]
    \item $f\in\A(f)$ for every $f$.
    \item For any $f\colon A\to B$ and $g\colon B\to C$ and composable $g'\in \A(g),f'\in \A(f)$ we have $g'\circ f'\in \A(g\circ f)$.
     Moreover, any $h\in \A(g\circ f)$ factorizes as $g'\circ f'$ with $g'\in \A(g)$ and $f'\in \A(f)$.
    \item For any $f\colon A\to B$, $g\colon C\to D$ in $\cat{C}$ and $f'\in \A(f), g'\in \A(g)$ we have $f'\otimes g'\in \A(f\otimes g)$. Moreover, any $h\in \A(f\otimes g)$ factorizes as $h'\circ (f'\otimes g')$ with $f'\in \A(f)$, $g'\in \A(g)$ and $h'\in \A(\id[B\otimes D])$. In addition, if $\dom(h)=A\otimes C$, we can choose the factorization  $h=h'\circ (f'\otimes g')$ so that $\dom(f')=A$ and $\dom(g')=C$. 
  \end{enumerate}
 Let $x\colon (A,r)\to (B,s)$ define a morphism in the resource theory $\res RF$ induced by $F\colon\cat{D}\to\cat{C}$ and $R\colon \cat{C}\to\cat{Set}$. We say that $x$ is \emph{secure} against an attack model $\A$ on $\cat{C}$ (or $\A$-secure) if for any $a\in \A(F(x))$ with $\dom(a)=F(A)$ there is $b\in\A(\id[F(B)])$ with $\dom(b)=F(B)$ such that $R(a)r=R(b)s$, \ie for any such $a$ there is a $b$ making the following square commute:
  \[\begin{tikzpicture}
    \matrix (m) [matrix of math nodes,row sep=2em,column sep=4em,minimum width=2em]
    {
     I_\Set=\{*\} & RF(B) \\
     RF(A) & R(X) \\};
    \path[->]
    (m-1-1) edge node [left] {$r$} (m-2-1)
           edge node [above] {$s$} (m-1-2)
    (m-1-2) edge node [right] {$R(b)$} (m-2-2)
    (m-2-1) edge node [below] {$R(a)$}  (m-2-2);
  \end{tikzpicture}\]
\end{defi}
The above definition of security asks for perfect equality and corresponds to information-theoretic security in cryptography. This is often too much to hope for, and we will relax this requirement in Section~\ref{sec:extensions}.

The intuition is that $\A$ gives, for each process in $\CC$, the set of behaviors that the attackers could force to happen instead of honest behavior. In particular, $\A(\id[B])$ gives the behaviors that are available to attackers given access to a system of type $B$. Then property (i) amounts to the assumption that the adversaries could behave honestly. The first halves of properties (ii) and (iii) say that, given an attack on $g$ and one on $f$, both attacks could happen when composing $g$ and~$f$ sequentially or in parallel. The second parts of these say that attacks on composite processes can be understood as composites of attacks. However, note that (iii) does not say that an attack on a product has to be a product of attacks: the factorization says that any $h\in\A(g\otimes f)$ factorizes as in Figure~\ref{fig:attackonprod} with $g'\in \A(g)$, $f'\in \A(f)$ and $h'\in \A(\id[B\otimes D])$. The intuition is that an attacker does not have to attack two parallel protocols independently of each other, but might play the protocols against each other in complicated ways. This intuition also explains why we do not require that all morphisms in $\A(f)$ have $F(A)$ as their domain, despite the definition of $\A$-security quantifying only against those: when factoring $h\in \A(g\circ f)$ as $g'\circ f'$ with  $g'\in \A(g)$ and $f'\in \A(f)$, we can no longer guarantee that $F(B)$ is the domain of $g'$---perhaps the attackers take us elsewhere when they perform~$f'$. Finally, the security definition is directly abstracted from the real-world-ideal-world paradigm, and can be seen as saying that a protocol realizes some target functionality securely if, for any attack on the protocol, there is an attack on the target functionality with identical end results. In other words, the attackers can not achieve anything during the protocol that they could not achieve with the target functionality, so that the protocol is at least as secure as the target functionality.

If one thinks of $F\colon\cat{D}\to\cat{C}$ as representing the inclusion of free processes into general processes, one also gets an explanation why we do not insist that free processes and attacks live in the same category, \ie that $F=\id[\CC]$. This is simply because we might wish to prove that some protocols are secure against attackers that can use more resources than we wish or can use in the protocols.

\begin{rem}
One can rewrite condition (ii) as $\A(g)\circ \A(f)=\A(g\circ f)$, if one defines $\A(g)\circ \A(f)$ as $\{g'\circ f'\colon g'\in\A(g),f'\in \A(f), g'\circ f'\text{ is defined}\}$. Similarly, $\A(f\otimes g)=\A(\id[B\otimes D])\circ (\A(f)\otimes \A(g))$ captures most of condition (iii). These equations look suspiciously close to stating that $\A$ is some kind of a functor, except that \begin{itemize}
  \item We do not require that identities are preserved. Intuitively, this corresponds to the fact that even if a protocol tells everyone to do nothing, dishonest parties might deviate arbitrarily.
  \item Given $f\colon A\to B$, the class $\A(f)$ is not required to be a subset of $\CC(A,B)$, making it difficult to see what the codomain of $\A$ is supposed to be. We do not require $\A(f)\subseteq \CC(A,B)$ as we cannot expect attackers to respect our type system: for instance, if a given party is supposed to map their system to the tensor unit (representing them discarding information), they might not do so if they're not honest.
\end{itemize}
We leave understanding attack models in more familiar categorical terms for future work.
\end{rem}

\begin{exa} For any SMC $\CC$ there are two trivial attack models: the minimal one defined by $\A(f)=\{f\}$ and the maximal one sending $f$ to the class of all morphisms of $\CC$. We interpret the minimal attack model as representing honest behavior, and the maximal one as representing arbitrary malicious behavior.
\end{exa}
\begin{prop} If $\A_1,\dots ,\A_n$ are attack models on SMCs $\CC_1,\dots ,\CC_n$ respectively, then there is a product $\prod_{i=1}^n\A_i$ attack model on $\prod_{i=1}^n\CC_i$ defined by $(\prod_{i=1}^n \A_i)(f_1,\dots, f_n)=\prod_{i=1}^n \A_i(f_i)$.
\end{prop}
\begin{proof} The required properties of $\prod_{i=1}^n\A_i$ follow from those of each $\A_i$ and the fact that operations in $\prod_{i=1}^n\CC_i$ are defined pointwise.
\end{proof}
This proposition, together with the minimal and maximal attack models, is already expressive enough to model multi-party computation where some subset of the parties might do arbitrary malicious behavior. Indeed, consider the $n$-partite resource theory induced by $\CC^n\xrightarrow{\otimes}\CC\xrightarrow{\hom(I,-)}\Set$. Let us first model a situation where the first $n-1$ participants are honest and the last participant is dishonest. In this case we can set $\A=\prod_{i=1}^n \A_i$ where each of $\A_1,\dots ,\A_{n-1}$ is the minimal attack model on $\CC$ and $\A_n$ is the maximal attack model. Then, an attack on $\f=(f_1,\dots f_n)\colon ((A_i)_{i=1}^n,r)\to ((B_i)^n_{i=1},s)$ can be represented by the first $n-1$ parties obeying the protocol and the $n$th party doing an arbitrary computation~$a$, as depicted in the two pictures of Figure~\ref{fig:attack},
where $[n]:=\{1,\dots,n\}$, $(k,n]:=\{k+1,\dots n\}$, $\f|_{[k]}:=\bigotimes_{i=1}^k f_i$, and here $k=n-1$. The latter representation will be used when we do not need to emphasize pictorially the fact that the honest parties are each performing their own individual computations.

If instead of just one attacker, there are several \emph{independently} acting adversaries, we can take $\A=\prod_{i=1}^n \A_i$ where $\A_i$ is the minimal or maximal attack structure depending on whether the $i$th participant is honest or not. If the set of dishonest parties can collude and communicate arbitrarily during the process, we need the flexibility given in Definition~\ref{def:attack} and have the attack structure live in a different category than where our protocols live. For simplicity of notation, assume that the first~$k$ agents are honest but the remaining parties are malicious and might do arbitrary (joint) processes in $\CC$. In particular, the action done by the dishonest parties $k+1,\dots , n$ need not be describable as a product $\bigotimes_{i=k+1}^n (a_i)$ of individual actions. In that case we define $\A$ as follows: we first consider our resource theory as arising from $\CC^n\xrightarrow{\id^k\times\otimes}\CC^k\times \CC\xrightarrow{\otimes}\CC\xrightarrow{\hom(I,-)}\Set$, and define $\A$ on $\CC^k\times \CC$ as the product of the minimal attack model on $\CC^k$ and the maximal one on $\CC$. Concretely, this means that the first $k$ agents always obey the protocol, but the remaining agents can choose to perform arbitrary joint behaviors in $\CC$. Then a generic attack on a protocol $\f$ can be represented exactly as before in Figure~\ref{fig:attack}, except we no longer insist that $k=n-1$. Now a protocol $\f$ is $\A$-secure if for any $a$ with $\dom(a)=(A_i)_{i=k+1}^n$ there is a $b$ with $\dom(b)=(B_i)_{i=k+1}^n$   satisfying the equation of Figure~\ref{fig:security}.

If one is willing to draw more wire crossings, one can easily depict and define security against an arbitrary subset of the parties behaving maliciously, and henceforward this is the attack model we have in mind when we say that some $n$-partite protocol is secure against some subset of the parties. Moreover, for any subset $J$ of dishonest agents, one could consider more limited kinds of attacks: for instance, the agents might have limited computational power or limited abilities to perform joint computations---as long as the attack model satisfies the conditions of Definition~\ref{def:attack} one automatically gets a composable notion of secure protocols by Theorem~\ref{thm:composition} in Section~\ref{sec:composition}.

We've seen that one party acting maliciously defines an attack model on $\CC^k$. We now show that this also defines an attack model $\A$ on $\ncomb(\CC^k)$.  Informally, $\A$ lets the $i$th party change their part of any $m$-comb arbitrarily, leaving everything else about morphisms of  $\ncomb(\CC^k)$ fixed.

\begin{defi}\label{def:maliciousparty_ncomb}
For a fixed $i\in\{1,\dots k)$, we define an attack model $\A$ on $\ncomb(\CC^k)$ corresponding to a malicious $i$th party as follows.
Consider first a basic morphism
\[(\sigma,(g_0,\dots g_m))\colon (A_i,B_i)_{i=1}^m\to (C,D)\] in $\ncomb(\CC^k)$
given by a permutation $\sigma\colon \{1,\dots , m\}\to\{1,\dots , m\}$ and an $m$-comb given by a tuple $(g_0,\dots g_m)$ of morphisms in $\CC^k$. Note that each $g_j$ is a morphism in $\CC^k$, and hence a tuple of morphisms of $\CC$.
Let us write $\pi_j\colon \CC^k \to \CC$ for the $j$th projection. We can now define
\[\A((\sigma,(g_0,\dots g_m))):=\{(\sigma,(h_0,\dots h_m)) \mid \pi_j(h_\ell)=\pi_j(g_\ell)\text{ for all }\ell\text { and all }j\neq i\}.\]
We note that this is well-defined, as equality of morphisms in $\CC^k$ is defined pointwise, so that modifying the $i$th coordinate of an $m$-comb in a particular way respects the equivalence relation on $m$-combs.

A morphism $(A_i,B_i)_{i=1}^m\to (C_j,D_j)_{j=1}^n$ is given by a function $f\colon \{1,\dots , m\}\to\{1,\dots , n\}$ and morphisms $x_j \colon (A_i,B_i)_{i\in f^{-1}(j)}\to (C_j,D_j)$ for each $j$. For such a morphism $(f,(x_i)_{i=1}^n)$ we define
\[\A(f,(x_i)_{i=1}^n):=\{(f,(y_i)_{i=1}^n)\mid y_i\in\A(x_i)\text{ for }i=1,\dots,n\}.\]
 \end{defi}

To see that this is an attack model in the sense of Definition~\ref{def:attack}, note first that $f\in\A(f)$ for any morphism $f$, satisfying condition (i).  We also note that if one composes sequentially a composable pair of attacks, the result is also an attack, so the first half of (ii) holds. For the other direction it suffices to consider an attack on a nested comb as in Figure~\ref{fig:nested_combs}. Such an attack corresponds to the $i$th party replacing their resulting $m$-comb with an arbitrary $m$-comb: however, we may think of this as resulting from the $i$th party first replacing each nested comb with something else, and then replacing the outer comb appropriately. This gives the other direction of (ii). Condition (iii) is clear, as by construction attacks on a parallel composite are parallel composites of attacks.

\subsection{Composition theorems}\label{sec:composition}
\begin{thm}\label{thm:composition} Given symmetric monoidal functors $F\colon\cat{D}\to\cat{C}$, $R\colon \cat{C}\to\Set$ with $F$ strong monoidal and $R$ lax monoidal, and an attack model $\A$ on $\cat{C}$, the class of $\A$-secure maps forms a wide sub-SMC of the resource theory $\res RF$ induced by $RF$.
\end{thm}

\begin{proof}
 We first prove the claim when $F=\id[\CC]$. As the class of $\A$-secure maps is a subclass of maps inside an SMC, it suffices to show it contains all coherence isomorphisms (and thus all identities) and is closed under $\circ$ and $\otimes$.

For coherence isomorphisms we prove a stronger claim and show that all isomorphisms are $\A$-secure. Let $f\colon (A,r)\to (B,s)$ be an isomorphism so that $f$ is an isomorphism $A\to B$ in $\cat{C}$, and consider $a\in \A(f)$ with $\dom(a)=A$. Then $R(a)r=R(a)R(f^{-1})R(f)r=R(a)R(f^{-1})s$, so it suffices to show that $af^{-1}\in\A(\id[B])$. Property (i) of $\A$ implies that $(f^{-1})\in\A(f^{-1})$ so that property (ii) gives us $af^{-1}\in\A(ff^{-1})=\A(\id[B])$, as desired.

Assume now that $f\colon (A,r)\to (B,s)$ and $g\colon (B,s)\to (C,t)$ are $\A$-secure. Given $h\in\A(g\circ f)$ with domain $A$, factorize it as $g'\circ f'$ as guaranteed by (ii). As $f$ is $\A$-secure, there is some $b\in \A(\id[B])$ with $R(f')r=R(b)s$ and by (ii) $g'b\in\A(g)$, so that security of $g$ implies the existence of $c\in\A(\id[B])$ such that $R(g'b)(s)=R(c)t$. Thus $R(g'f')r=R(g')R(b)s=R(c)t$ showing that $g\circ f$ is $\A$-secure.

To show that secure maps are closed under $\otimes$, let $f\colon (A,r)\to (B,s)$ and $g\colon (C,t)\to (D,u)$ be $\A$-secure. Given $h\in \A(f\otimes g)$ with domain $A\otimes C$, factorize it as $h'\circ (f'\otimes g')$ as guaranteed by (iii). Then security of $f$ and $g$ gives us $b\in\A(\id[B])$ and $d\in\A(\id[D])$ so that $R(f')r=R(b)s$ and $R(g')t=R(d)u$.
We now claim that the diagram
  \[\begin{tikzpicture}
    \matrix (m) [matrix of math nodes,row sep=2em,column sep=4em,minimum width=2em]
    {
     I_\Set=\{*\} & R(A)\times R(C) & R(A\otimes C) \\
     R(B)\times R(D) & R(Y)\times R(Z) & R(Y\otimes Z) & R(X) \\
     & R(B\otimes D) & \\};
    \path[->]
    (m-1-1) edge node [left] {$(s,u)$} (m-2-1)
           edge node [above] {$(r,t)$} (m-1-2)
    (m-1-2) edge node [right] {$R(f')\times R(g')$} (m-2-2)
            edge node [above] {} (m-1-3)
    (m-1-3) edge node [right] {$R(f'\otimes g')$} (m-2-3)
            edge[out=0,in=135] node [above,sloped] {$R(h)$} (m-2-4)
    (m-2-1) edge node [above] {$R(b)\times R(d)$}  (m-2-2)
            edge node [left] {$$} (m-3-2)
    (m-2-2) edge node [below] {$$}  (m-2-3)
    (m-2-3) edge node [below] {$R(h')$}  (m-2-4)
    (m-3-2) edge node [below,sloped] {$R(b\otimes d)$} (m-2-3);
  \end{tikzpicture}\]
commutes, where the unlabelled arrows come from the natural transformation $R(-)\times R(-)\to R(-\otimes - )$ that is part of the lax monoidal structure on $R$. The top left square commutes by security of $f$ and $g$, and the rightmost shape commutes by the factorization $h=h'\circ (f'\otimes g')$. The remaining  two sub-diagrams are naturality squares and hence commute. Hence the whole diagram commutes. As the top path equals $R(h)(r\otimes t)$ and the bottom path equals $R(h'\circ (b\otimes d))(s\otimes u)$, we see that $h'\circ (b\otimes d)\in \A(\id[B]\otimes\id[D])$ witnesses that $f\otimes g$ is $\A$-secure.

To prove the claim for an arbitrary strong monoidal~$F$, observe first that $f\colon (A,r)\to (B,s)$ is $\A$-secure in $\res RF$ if, and only if $F(f)\colon (F(A),r)\to (F(B),s)$ is $\A$-secure in $\res R$. The claim can now be deduced from the existence and description of pullbacks in the category of SMCs, but we give an explicit proof: the class of $\A$-secure maps in $\res RF$ contains all isomorphisms and is closed under composition because it is so in $\res R$. As~$F$ is strong monoidal, the square
\[\begin{tikzpicture}
    \matrix (m) [matrix of math nodes,row sep=2em,column sep=4em,minimum width=2em]
    {
     F(A\otimes C) & F(B\otimes D) \\
     F(A)\otimes F(C) & F(B)\otimes F(D) \\};
    \path[->]
    (m-1-1) edge node [left] {$\cong$} (m-2-1)
           edge node [above] {$F(f\otimes g)$} (m-1-2)
    (m-2-2) edge node [right] {$\cong$} (m-1-2)
    (m-2-1) edge node [below] {$F(f)\otimes F(g)$}  (m-2-2);
  \end{tikzpicture}\]
commutes in $\CC$. If $f\colon (A,r)\to (B,s)$ and $g\colon (C,t)\to (D,u)$ are $\A$-secure in $\res RF$, then $F(f)$ and $F(g)$ are $\A$-secure in $\res R$. The case $F=\id[\CC]$ implies that $F(f)\otimes F(g)$ is $\A$-secure so that $F(f\otimes g)$ is $\A$-secure as a composite of secure maps, which means that $f\otimes g$ is $\A$-secure in $\res RF$ as desired.
\end{proof}

\noindent So far we have discussed security only against a single, fixed subset of dishonest parties, while in multi-party computation it is common to consider security against any subset containing \eg at most $n/3$ or $n/2$ of the parties. However, as monoidal subcategories are closed under intersection, we immediately obtain composability against multiple attack models.
\begin{cor}\label{cor:simultaneoussafety}
Given a non-empty family of functors $(\cat{D}\xrightarrow{F_i}\cat{C}_i\xrightarrow{R_i}\Set)_{i\in I}$ with $R_iF_i=R_jF_j=:R$ for all $i,j\in I$ and attack models $\A_i$ on $\CC_i$ for each $i$, the class of maps in $\res R$ that is secure against every $\A_i$ is a sub-SMC of $\res R$.
\end{cor}
Using Corollary~\ref{cor:simultaneoussafety} one readily obtains composability of protocols that are simultaneously secure against different attack models $\A_i$. Thus one could, in principle, consider composable cryptography in an $n$-party setting where some subsets are honest-but-curious, some might be outright malicious but have limited computational power, and some subsets might be outright malicious but not willing or able to coordinate with each other, without reproving any composition theorems.

While the security definition of $f$ quantifies over $\A(f)$, which may be infinite, under suitable conditions it is sufficient to check security only on a subset of $\A(f)$, so that whether $f$ is $\A$-secure often reduces to finitely many equations.

\begin{defi}\label{def:initialattack}
Given $f\colon A\to B$, a subset $X$ of $\A(f)$ is said to be \emph{initial} if any $f'\in\A(f)$ with $\dom(f')=A$ can be factorized as $b\circ a$ with $a\in X$ and $b\in\A(\id[B])$.
\end{defi}
\begin{thm}\label{thm:initialattacks}
 Let $f\colon (A,r)\to (B,s)$ define a morphism in the resource theory induced by $F\colon\cat{D}\to\cat{C}$ and $R\colon \cat{C}\to\cat{Set}$ and let $\A$ be an attack model on $\CC$. If $X\subset\A(F(f))$ is initial, then $f$ is $\A$-secure if, and only if the security condition holds against attacks in $X$, \ie if for any $f'\in X$ with $\dom(f')=F(A)$ there is $b\in\A(\id[F(B)])$ such that $R(f')r=R(b)s$.
\end{thm}

\begin{proof}
Let $f'\in\A(F(f))$ be an attack satisfying $\dom(f')=F(A)$. As $X$ is initial, we can factorize $f'$ as $b\circ a$ with $a\in X$ and $b\in\A(\id[F(B)])$. As $f$ is secure against attacks in $X$ and $a\in X$, there is a $c\in\A(\id[F(B)])$  such that $R(a)r=R(c)s$. Then $R(f')r=R(b)R(a)r=R(b)R(c)s=R(b\circ c)s$ so that $b\circ a\in \A(\id[F(B)])$ witnesses security against $f'$.
\end{proof}

Let us return to the example of $\CC^n\to\CC$ with the first $k$ agents being honest and the final $n-k$ dishonest and collaborating. Then we can take a singleton as our initial subset of attacks on $\f$, and this is given by $\f|_{[k]}\otimes (\bigotimes_{i=k+1}^n\id)$. Intuitively, this represents a situation where the dishonest parties $k+1,\dots ,n$ merely stand by and forward messages between the environment and the functionality, so that initiality can be seen as explaining ``completeness of the dummy adversary''~\cite[Claim 11]{Can01} in UC-security. In this case the security condition can be equivalently phrased by saying that there exists $b\in\A([\id[b]])$ satisfying the equation of Figure~\ref{fig:initial_attack},
which reproduces the pictures of~\cite{MT13}. Similarly, for classical honest-but-curious adversaries one usually only considers the initial such adversary, who follows the protocol otherwise except that they keep track of the protocol transcript.
\begin{thm}\label{thm:perfectlifting}
In the resource theory of $n$-partite states, if $(f_1,\dots f_n)$ is secure against some subset $J$ of $[n]$ and $F$ is a strong monoidal, then $(Ff_1,\dots, Ff_n)$ is secure against~$J$ as well.
\end{thm}

\begin{proof}
Let us first spell out explicitly how the domain and codomain of $(Ff_1,\dots,Ff_n)$ depend on those of $\f$: if $\bar{f}\colon ((A_i)_{i=1}^n,r)\to ((B_i)^n_{i=1},s)$, then $Fr\colon F(I_\CC)\to F(\bigotimes_{i=1}^n A_i)$ induces a state on $\bigotimes_{i=1}^n F(A_i)$ by precomposing with the isomorphism $I_\DD\to F(I_\CC)$ and postcomposing with the isomorphism $F(\bigotimes_{i=1}^n A_i)\cong \bigotimes_{i=1}^n F(A_i)$ stemming from the strong monoidal structure of $F$. This is the state that $(Ff_1,\dots,Ff_n)$ transforms to the one induced by $F(s)$. Let us now show that this transformation is secure provided that $\f$ is.

The heart of the argument is already apparent in the case of $n=2$, so let us first show that if $(f_A,f_B)$ is secure against a malicious Bob, so is $(Ff_A,Ff_B)$. For this attack model, there is an initial attack, and the corresponding security constraint is depicted in Figure~\ref{fig:initial_attack}. Then security of $(Ff_A,Ff_B)$ can be shown graphically using the functorial boxes of~\cite{mellies2006functorial} by considering the equations
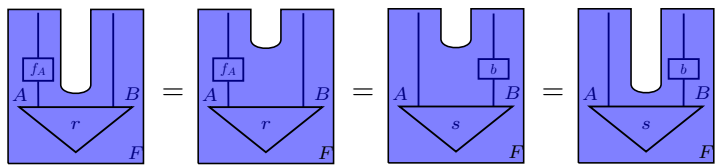
\[
\begin{pic}
\node at (.8,-.6) {$F$};
  \setlength\minimumstatewidth{15mm}
    \node[nmorphismtwocell,scale=.5,font=\normalsize] (f) at (-.5,.5) {$f_{A}$};
    \node[state] (x) at (0,0) {$r$};
    \draw (f.south) to node[left] {$A$} ([xshift=-3.5pt]x.A);
    \draw (f.north) to +(0,.6) node[right] {};
    \draw ([xshift=3.5pt]x.B) to node[right]  {$B$} +(0,.35)  to +(0,1.25);
\draw[black, fill = blue, fill opacity = 0.5, semithick] (-.9,-.75) to (-.9,1.3)
to (-0.2,1.3) to ++(0,-1) to[in=-90,out=-90] ++(0.4,0) to ++(0,1)
to (.9,1.3) to (.9,-.75) to (-.9,-.75);
  \end{pic}=
\begin{pic}
  \setlength\minimumstatewidth{15mm}
    \node[nmorphismtwocell,scale=.5,font=\normalsize] (f) at (-.5,.5) {$f_{A}$};
    \node[state] (x) at (0,0) {$r$};
    \draw (f.south) to node[left] {$A$} ([xshift=-3.5pt]x.A);
    \draw (f.north) to +(0,.6) node[right] {};
    \draw ([xshift=3.5pt]x.B) to node[right]  {$B$} +(0,.35)  to +(0,1.25);
\draw[black, fill = blue, fill opacity = 0.5, semithick] (-.9,-.75) to (-.9,1.3)
to (-0.2,1.3) to ++(0,-.4) to[in=-90,out=-90] ++(0.4,0) to ++(0,0.4)
to (.9,1.3) to (.9,-.75) to (-.9,-.75);
\node at (.8,-.6) {$F$};
  \end{pic}=\begin{pic}
    \setlength\minimumstatewidth{15mm}
     \node[nmorphismtwocell,scale=.5,font=\normalsize] (g) at (.5,.5) {$b$};
    \node[state] (x) at (0,0) {$s$};
    \draw ([xshift=-3.5pt]x.A) to node[left] {$A$} +(0,.35) to +(0,1.25);
    \draw (g.south) to node[right] {$B$} ([xshift=3.5pt]x.B);
    \draw (g.north) to +(0,.6) node[right] {};
   \draw[black, fill = blue, fill opacity = 0.5, semithick] (-.9,-.75) to (-.9,1.3)
to (-0.2,1.3) to ++(0,-.4) to[in=-90,out=-90] ++(0.4,0) to ++(0,0.4)
to (.9,1.3) to (.9,-.75) to (-.9,-.75);
\node at (.8,-.6) {$F$};
  \end{pic}
  =
  \begin{pic}
    \setlength\minimumstatewidth{15mm}
     \node[nmorphismtwocell,scale=.5,font=\normalsize] (g) at (.5,.5) {$b$};
    \node[state] (x) at (0,0) {$s$};
    \draw ([xshift=-3.5pt]x.A) to node[left] {$A$} +(0,.35) to +(0,1.25);
    \draw (g.south) to node[right] {$B$} ([xshift=3.5pt]x.B);
    \draw (g.north) to ++(0,.6) node[right] {};
   \draw[black, fill = blue, fill opacity = 0.5, semithick] (-.9,-.75) to (-.9,1.3)
to (-0.2,1.3) to ++(0,-1) to[in=-90,out=-90] ++(0.4,0) to ++(0,1)
to (.9,1.3) to (.9,-.75) to (-.9,-.75);
\node at (.8,-.6) {$F$};
  \end{pic}
  \]
where the second equation is security of the original protocol and the other two equations rely on $F$ being strong monoidal. The case of an arbitrary $n$ can be shown similarly by drawing a similar picture with $n-1$ dips in the box.
\end{proof}

\noindent For instance, if the inclusion of classical interactive computations into quantum ones is strong monoidal, \ie respects sequential and parallel composition (up to isomorphism), then unconditionally secure classical protocols are also secure in the quantum setting, as shown in the context of UC-security in~\cite[Theorem 15]{Unr10}. More generally, this result implies that the construction of the category of $n$-partite transformations secure against any fixed subset of $[n]$ is functorial in $\cat{C}$, and this is in fact also true for any family of subsets of $[n]$ by Corollary~\ref{cor:simultaneoussafety}.

\section{Further extensions of the framework}\label{sec:extensions}

The discussion above has been focused on perfect security, so that the equations defining security hold exactly. This is often too high a standard for security to hope for, and consequently cryptographers routinely work with computational or approximate security. We model this in two ways. The first approach replaces equations with an equivalence relation abstracting from the idea that the end results are ``computationally indistinguishable'' rather than strictly equal. The latter approach amounts to working in terms of a (pseudo)metric, that quantifies how close we are to the ideal resource, so that one can discuss approximately correct transformations or sequences of transformations that succeed in the limit. The first approach is mathematically straightforward and we discuss it next, while the second approach is discussed in Sections~\ref{sec:metric} and~\ref{sec:metricsecurity}. The second approach, while mathematically more involved, is needed to model protocols that are ``close enough'' to being computationally indistinguishable from the ideal, and thus to model statements in finite-key cryptography~\cite{TLGR12}.

\subsection{Security up to indistinguishability}\label{sec:equivalence}

We let $\Equ$ stand for the category of equivalence relations: its objects are sets $(X,\approx_X)$ equipped with an equivalence relation, and morphisms $(X,\approx_X)\to(Y,\approx_Y)$ are given by functions $f\colon X\to Y$  that respect the equivalence relation in that $x\approx_X x'\Rightarrow f(x)\approx_Y f(x')$. From now on, we will omit the subscript and write $\approx$ for the equivalence relation associated with any object of $\Equ$. The cartesian product of sets with the equivalence relations defined pointwise induces a symmetric monoidal structure on this category, and the operation $(X,\approx)\mapsto X/\approx_X$ of forming quotients gives a (strong) symmetric monoidal functor $E\colon \Equ\to\Set$. We now relax our notions of correctness and security when we have a symmetric monoidal functor $R\colon \CC\to\Equ$.

\begin{defi}\label{def:secuptoindistinguishability} Let $R\colon \CC\to\Equ$ be a symmetric monoidal functor.  Given $r\in R(A)$ and $s\in R(B)$, a morphism $f\colon A\to B$ is an \emph{$\approx$-correct transformation $(A,r)\to (B,s)$} if $R(f)r\approx s$. We let $\reseq R$ denote the resource theory of $\approx$-correct transformations.

Consider now symmetric monoidal functors $F\colon\cat{D}\to\cat{C}$, $R\colon \cat{C}\to\Equ$ with $F$ strong monoidal and $R$ lax monoidal and an attack model $\A$ on $\cat{C}$.  Let $f\colon (A,r)\to (B,s)$ define a morphism in the resource theory $\reseq   RF$. We say that $f$ is \emph{secure up to $\approx$} against an attack model $\A$ on $\cat{C}$ (or $\A$-secure up to $\approx$) if for any $a\in \A(F(f))$ with $\dom(a)=F(A)$ there is $b\in\A(\id[F(B)])$ with $\dom(b)=F(B)$ such that $R(a)r\approx R(b)s$, \ie the square from Definition~\ref{def:attack} commutes up to $\approx$.
\end{defi}

\begin{cor} Let $R,F,\A$ be as in Definition~\ref{def:secuptoindistinguishability}. The category $\reseq RF$ is a symmetric monoidal category, and the category of $\approx$-secure transformations against $\A$ is a wide sub-SMC of it.
\end{cor}

\begin{proof} Consider the composite $ERF$, where $E\colon \Equ\to\Set$ sends a set with an equivalence relation to its set of equivalence classes. It is straightforward to check that $f\colon A\to B$ defines a map $(A,r)\to (B,s)$ in $\reseq RF$ iff it defines a map $(A,[r])\to (B,[s])$ in $\res ERF$, implying the first claim. Similarly, $f$ is secure up to $\approx$ in $\reseq RF$ iff it is secure in $\res ERF$, so the second claim follows from Theorem~\ref{thm:composition}.
\end{proof}

It is also straightforward to adapt and prove Corollary~\ref{cor:simultaneoussafety} in order to state that maps that are secure up to $\approx$ against a set of attack models form an SMC.

\begin{exa} Let $\CC$ be a category equipped with a \emph{monoidal congruence}: a family of equivalence relations $\approx$ on each hom-set of $\CC$ that respects $\otimes$ and $\circ$ in that $f\approx f'$ and $g\approx g'$ imply $gf\approx g'f'$ (whenever defined) and $g\otimes f\approx g'\otimes f'$. In other words, let $\CC$ be an SMC enriched in $\Equ$. Then $\hom(I,-)$ gives a symmetric monoidal functor $\CC\to \Equ$, so one can work with $\approx$-correct resource theory of states and with $\approx$-security against a malicious subset of parties. Moreover, in this situation we also obtain a symmetric monoidal functor \[\hom(I,-)\colon \ncomb(\CC)\to\Equ,\] enabling us to work with $\approx$-correct and $\approx$-secure transformations  of morphisms of $\CC$.
\end{exa}

In Section~\ref{sec:dhke} we will give an example of this example, where $\approx$ will denote computational indistinguishability \ie the inability of efficient computational systems to tell two systems apart (except for a negligible advantage), with Diffie-Hellman key exchange then giving rise to a transformation that is secure up to $\approx$.

\subsection{Approximately correct transformations}\label{sec:metric}

We now move to the metric case. If for each $A$ the set of resources $R(A)$ associated to it is not just a set but has the structure of a metric space, using this additional structure enables one to construct other resource theories where instead of transforming $r\in R(A)$ to $s\in R(B)$ exactly we are happy to be able to get (arbitrarily) close. While such approximate (or asymptotic) conversions are readily studied in the physics literature (see~\eg~\cite[V.A and V.B]{chitambar:resource}), as far as we are aware this has not been formalized in the categorical context, so we first describe the situation without security constraints. As many interesting measures of distance in cryptography are in fact pseudometrics (non-equal functionalities might have distance~$0$), we work in a more general setting.

\begin{defi}
An \emph{extended pseudometric space} is a pair $(X,d)$ where $X$ is a set and $d\colon X\times X\to [0,\infty]$ is a function satisfying (i) $d(x,x)=0$, (ii) $d(x,y)=d(y,x)$ and (iii) $d(x,z)\leq d(x,y)+d(y,z)$ for all $x,y,z\in X$. A \emph{short map} $(X,d)\to (Y,e)$ is a function $f\colon X\to Y$ satisfying $d(x,y)\geq e(f(x),f(y))$. We will denote the category of extended pseudometric spaces and short maps simply by $\Met$. We equip $\Met$ with a monoidal structure where $(X,d)\otimes (Y,e)$ is given by equipping $X\times Y$ with $\ell^1$-distance, \ie the distance between $(x,y)$ and $(x',y')$ is given by $d(x,x')+e(y,y')$.

Let $R\colon \cat{C}\to\Met$ be a symmetric monoidal functor. Given $r\in R(A)$, $s\in R(B)$ and $\epsilon>0$, a morphism $f\colon A\to B$ is an \emph{$\epsilon$-correct transformation $(A,r)\to (B,s)$} if $d(R(f)r,s)<\epsilon$. The resource theory $\resmet R$ of \emph{asymptotically correct conversions} is defined as follows: an object is given by a pair $(A,r)$ where $A$ is an object of $\CC$ and $r\in R(A)$. A morphism $(A,r)\to (B,s)$ is given by a sequence $(f_n)_{n\in\N}$ of maps $A\to B$ in $\CC$ that is eventually $\epsilon$-correct for any $\epsilon>0$, \ie for which $R(f_n)r\to s$ as $n\to\infty$.
\end{defi}

In some resource theories, the relevant asymptotic transformations are allowed to use more and more copies of the resource, so that instead of a sequence of maps $A\to B$ we have a sequence $(f_n)_{n\in\N}$ of maps $A^{\otimes n}\to B$ taking $r^{\otimes n}$ to~$s$ in the limit. The theory developed here adapts easily to this variant as well, with essentially the same proofs.

\begin{lem}\label{lem:epsiloncorrect}
Let $R\colon \cat{C}\to\Met$ be symmetric monoidal. The composite or tensor product of an $\epsilon$-correct map with an $\epsilon'$-correct map is $\epsilon+\epsilon'$-correct.
\end{lem}

\begin{proof}
Assume that $f$ is an $\epsilon$-correct transformation $(A,r)\to (B,s)$ and that $g$ is an $\epsilon'$-correct transformation $(B,s)\to (C,t)$. As $R(g)$ is a short map, this gives $d(R(gf)r,s)\leq d(R(gf)r,R(g)s)+d(R(g)s,t)<\epsilon+\epsilon'$.

Assume now that $f:(A,r)\to (B,s)$ is $\epsilon$-correct and that $g\colon(C,t)\to (D,u)$ is $\epsilon'$-correct. Then $d(R(f\otimes g)r\otimes t,s\otimes u)\leq d((R(f)r,R(g)t),(s,u))=d(R(f)r,s)+d(R(g)t,u)<\epsilon+\epsilon'$.\looseness=-1
\end{proof}

\begin{thm}
The resource theory $\resmet R$ of asymptotically correct conversions induced by $R\colon \cat{C}\to\Met$ is a symmetric monoidal category.
\end{thm}

\begin{proof}
The coherence isomorphisms are given by constant sequences of coherence isomorphisms of the resource theory induced by $\cat{C}\xrightarrow{R}\Met\to\Set$, and this implies that they satisfy the required equations of an SMC. Moreover, as they are exact resource conversions, they are also asymptotically correct. Thus it suffices to check that asymptotically correct conversions are closed under $\circ$ and $\otimes$. But this follows from Lemma~\ref{lem:epsiloncorrect}: given two asymptotically correct transformations and $\epsilon>0$, the two transformations are eventually $\epsilon/2$-correct after which their composite (whether $\circ$ or $\otimes$) is $\epsilon$-correct.
\end{proof}

In particular, if $\CC$ is $\Met$-enriched, the functor $\hom(I,-)$ lands in $\Met$ so that one can discuss asymptotic transformations between states.

While in resource theories one first tries to understand whether a given transformation is possible at all, once some resource conversion has been shown to be possible one might ask for more. In particular, in the asymptotic setting one might want the sequence $(f_n)_{n\in \N}$ to be efficient (and in particular computable) in $n$, and to converge to the target fast in terms of some measure of cost of implementing $f_n$. One might even want to be able to give an explicit bound on the distance between $R(f_n)r$ and $s$, as is done for instance in finite-key cryptography~\cite{TLGR12}.
However, such considerations are best addressed when working inside a specific resource theory rather than being hardwired into the definitions at the abstract level. Conversely, if one can show that a given asymptotic transformation is impossible even for such a permissive notion of transformation, the resulting no-go theorem is stronger than if one worked with ``efficient'' sequences.

\subsection{Computational security}\label{sec:metricsecurity}

We now show that one can reason composably about computational security in such a metric setting. The proofs follow rather straightforwardly from the definitions we have by using the structure at hand: most importantly, from the triangle inequality of any metric space and the fact that our maps between metric spaces are contractive. For concrete models of cryptography, one might need to do nontrivial work to show that one has all this structure, after which our theorems apply.

\begin{defi}
Consider $F\colon\cat{D}\to\cat{C}$ and $R\colon \cat{C}\to\Met$ and an attack model $\A$ on $\cat{C}$. For an $\epsilon>0$ and an $\epsilon$-correct map $f\colon (A,r)\to (B,s)$, we say that $f$ is an \emph{$\epsilon$-secure transformation} $(A,r)\to (B,s)$ against $\A$ if for any $a\in \A(F(f))$ with $\dom(a)=F(A)$ there is $b\in\A(\id[F(B)])$ such that $d(R(a)r,R(b)s)<\epsilon$.

Let $(f_n)_{n\in\N}\colon (A,r)\to (B,s)$ now define an asymptotically correct conversion in $\resmet RF$. We say that $(f)_{n\in\N}$ is \emph{asymptotically secure} against $\A$ (or asymptotically $\A$-secure) if it is eventually $\epsilon$-secure for any $\epsilon>0$. Explicitly, $(f_n)_{n\in\N}\colon (A,r)\to (B,s)$ is asymptotically secure if for any $\epsilon>0$ there is a threshold $k\in\N$ such that for any $n>k$ and any $a\in \A(F(f_n))$ with $\dom(a)=F(A)$ there is $b\in\A(\id[F(B)])$ such that $d(R(a)r,R(b)s)<\epsilon$.
\end{defi}

Note that while we changed our notion of security both in the presence of $\approx$ and in the
(pseudo)metric setting, we kept our notion of an attack model the same. We expect that the theory would carry over even if one defined ``attack models up to $\approx$''  (and ``metric attack models'') by relaxing Definition~\ref{def:attack} to state that the required factorizations only exist up to $\approx$ (up to every $\epsilon>0$). However, we do not pursue these generalizations here as they are not needed for the present work.  

We now show that bounds on security compose additively.

\begin{lem}\label{lem:epsilonsafe} Let $R\colon \cat{C}\to\Met$ be lax monoidal and $\A$ an attack model on $\CC$. The composite or tensor product of an $\epsilon$-secure map with an $\epsilon'$-secure map is $\epsilon+\epsilon'$-secure.
\end{lem}

\begin{proof}
We have already seen that $\epsilon$-correctness behaves as desired in Lemma~\ref{lem:epsiloncorrect}. Assume that $f$ is an $\epsilon$-secure transformation $(A,r)\to (B,s)$ and that $g$ is an $\epsilon'$-secure transformation $(B,s)\to (C,t)$ against $\A$. Given $h\in\A(g\circ f)$ with domain~$A$, factorize it as $g'\circ f'$ as guaranteed by (ii).
 As $f$ is $\A$-secure there is some $b\in \A(\id[B])$ with \mbox{$d(R(f')r,R(b)s)<\epsilon$}. Now $g'b\in\A(g)$ by (ii) so that security of $g$ implies the existence of $c\in\A(\id[B])$ such that $d(R(g'b)(s),R(c)t)<\epsilon'$. Thus \[d(R(h)r,R(c)t)=d(R(g'f')r,R(c)t)\leq d(R(g'f')r,R(g')R(b)s)+ d(R(g')R(b)s,R(c)t)<\epsilon+\epsilon'\] as desired.

 Assume now that $f$ is $\epsilon$-secure transformation $(A,r)\to (B,s)$ against $\A$ and that $g$ is $\epsilon'$-secure transformation $(C,t)\to (D,u)$ against $\A$. Given $h\in \A(f\otimes g)$ with domain $A\otimes C$ factorize it as $h'\circ (f'\otimes g')$ as guaranteed by (iii). Then $\epsilon$-security of $f$ ($\epsilon'$-security of $g$) gives us $b\in\A(\id[B])$ so that $d(R(f')r,R(b)s)<\epsilon$ ($c\in\A(\id[D])$ so that $d(R(g')t,R(c)u)<\epsilon'$). Now
 \begin{align*} d(R(h)(r\otimes t),R(h'\circ (b\otimes c))(s\otimes u)) &=d(R(h')\circ R(f'\otimes g')(r\otimes t),R(h')\circ R(b\otimes c)(s\otimes u))\\
                &\leq d(R(f'\otimes g')(r\otimes t),R(b\otimes c)(s\otimes u))\\
                &=d(R(f')r,R(b)s)+d(R(g')t,R(c)u)<\epsilon+\epsilon'\end{align*} as desired.
\end{proof}

We now give a composition theorem for asymptotically secure protocols.

\begin{thm}\label{thm:metriccomposition}
Given symmetric monoidal functors $F\colon\cat{D}\to\cat{C}$, $R\colon \cat{C}\to\Met$ with $F$ strong monoidal and $R$ lax monoidal, and an attack model $\A$ on $\cat{C}$, the class of asymptotically $\A$-secure maps forms a wide sub-SMC of the asymptotic resource theory $\resmet RF$ induced by $F$ and $R$.
\end{thm}

\begin{proof}
As with Theorem~\ref{thm:composition}, it suffices to show that asymptotically secure maps contain all coherence isomorphisms and are closed under $\circ$ and $\otimes$. Moreover, the reduction from the general case to $F=\id$ is the same, so we assume that $F=\id$.
It is easy to see that whenever $f$ is $\A$-secure in the resource theory induced by $\cat{C}\xrightarrow{R}\Met\to\Set$, the constant sequence $(f)_{n\in \N}$ is asymptotically $\A$-secure. Thus security of coherence isomorphisms implies their asymptotic security.

Assume now that $(f_n)_{n\in \N}\colon (A,r)\to (B,s)$ and $(g_n)_{n\in \N}\colon (B,s)\to (C,t)$ are asymptotically $\A$-secure. Given $\epsilon>0$, for sufficiently large $n$ both $f_n$ and $g_n$ are $\epsilon/2$-secure so that their composite is $\epsilon$-secure by Lemma~\ref{lem:epsilonsafe}. The case for~$\otimes$ follows similarly from Lemma~\ref{lem:epsilonsafe}.
\end{proof}

\begin{cor}\label{cor:metricsimultaneoussafety}
Given a non-empty family of functors $(\cat{D}\xrightarrow{F_i}\cat{C}_i\xrightarrow{R_i}\Met)_{i\in I}$ with $R:=R_iF_i=R_jF_j$ for all $i,j\in I$ and attack models $\A_i$ on $\CC_i$ for each $i$, the class of maps in $\resmet R$ that is asymptotically secure against each $\A_i$ is a sub-SMC of $\resmet R$.
\end{cor}

To make these abstract results closer to cryptographic practice, one would work within some explicit $\CC$ and with (pseudo)metrics relevant for cryptographers. A paradigmatic case is given by metrics induced by distinguisher advantage, where one defines the distance between two behaviors by first taking the supremum over all (efficient) distinguishers $d$ of the probability of $d$ distinguishing the two behaviors and then normalizing this value from $[1/2,1]$ to $[0,1]$. If our starting category $\CC$ contains processes that are not (efficiently) computable, such distinguisher metrics might not be contractive as composing two distinct behaviors with a very powerful behavior might help a distinguisher trying to tell them apart. However, as long as
one restricts~$\CC$ (and consequently the behaviors available as resources, protocols and attacks) to behaviors that the relevant class of distinguishers can freely implement, this readily results in a $\Met$-enrichment, as composing two morphisms with a fixed morphism available to the distinguishers cannot increase distinguisher advantage. For instance, if the metric is induced by distinguisher advantage of polynomial-time distinguishers, one should get a $\Met$-enrichment on the subcategory of $\CC$ corresponding to polynomial-time behaviors. Once one has specified a concrete~$\CC$ and a $\Met$-enrichment on it, for any asymptotically secure protocol one can then discuss its speed of convergence, and in principle discuss which actual value of the security parameter is sufficiently secure for the task at hand.

We now wish to prove a variant of Theorem~\ref{thm:perfectlifting} in the approximate setting, abstracting from~\cite[Theorem 18]{Unr10}. Again, we specialize to the $n$-partite resource theory of states, where our attack models consist of some subset $J\subset\{1,\dots,n\}$ behaving maliciously. In this case, we assume our base categories to be $\Met$-enriched, so that $\hom(I,-)$ lands in $\Met$. In such a setting, a protocol is a sequence $(\f_i)_{i\in\N}$ where each $\bar{f_i}:=(f_{i,1},\dots f_{i,n})$ is an $n$-tuple of morphisms.

\begin{thm}\label{thm:computationallifting}
Let $\cat{C}$ and $\cat{D}$ be $\Met$-enriched SMCs, and let $F\colon \cat{C}\to\cat{D}$ be a strong monoidal $\Met$-enriched functor. If $(\f_i)_{i\in\N}$ is an asymptotic transformation between two states of $\cat{C}$ that is asymptotically secure against $J\subset\{1,\dots, n\}$, so is $(F\f_i)_{i\in \N}$.
\end{thm}

\begin{proof}
Again, it suffices to prove security against initial attacks. Now, the proof of Theorem~\ref{thm:perfectlifting} implies that if the desired equation in $\CC$ holds up to $\epsilon>0$, so does the equation in~$\DD$, so the claim follows.
\end{proof}

As discussed in~\cite{Unr10}, the computational version above is not as strong as the result in the case of perfect security, as the assumptions of Theorem~\ref{thm:computationallifting} are rather strong. For instance, if a protocol is secure against polynomial-time classical adversaries, it does not follow that it is secure against polynomial-time quantum adversaries. Correspondingly, if we use the metric induced by ``polynomial-time distinguishers'', the inclusion of classical computations into quantum computations is not $\Met$-enriched, as the distances might increase. However, if on the quantum side we use polynomial-time distinguishers, but on the classical side we use distinguishers that are able to simulate quantum polynomial-time machines, then protocols that are classically secure remain secure when thought of as quantum computations.

\subsection{Setup assumptions and freely usable resources}\label{sec:set-up}

Cryptographers often prove results saying that a given functionality is impossible to realize in the \emph{plain} model but is possible with some \emph{setup}. For instance, in~\cite{CF01} they show that bit commitment (BC) is impossible in the plain UC-framework but it is possible assuming a common reference string (CRS)---a functionality that gives all parties the same string drawn from some fixed distribution. In our viewpoint, claims such as these can be interpreted in the categories we have already built: for instance, impossibility of commitments amounts to non-existence of a secure map $I\to BC$ that builds bit commitments out of a trivial resource $I$, and possibility of bit commitments given a common reference string amounts to the existence of a secure protocol $CRS\to BC$.

A related, but distinct matter is that sometimes cryptographers wish to make some (possibly shared) functionalities freely available to all parties without having to explicitly mention them being used as a resource. For instance, so far in our framework all communication between the honest parties has been mediated by the functionality $r$ that they start from. However, one might want to model situations where \eg pairwise communication between parties is freely available (as is standard in multi-party computation) and does not need to be provided explicitly by the functionality one starts from. Put more abstractly, one might wish to declare some set $\mathcal{X}$ of functionalities ``free'' and think of secure protocols that build $s$ from $r$ and some functionalities from $\mathcal{X}$ just as maps $r\to s$, without having to explicitly keep track of how many copies of which $x\in \mathcal{X}$ was used. This is in fact something that happens quite often in resource theories even before any security conditions arise, as it could happen that the free processes $\CF$ are not quite expressive enough for the resource theory at hand. While one could try to define a larger category of free processes directly, it might be technically more convenient to obtain a larger class of free processes by allowing resource transformations to consume a resource from some class that is considered free. This can be achieved via a general construction on SMCs, a special case used in~\cite{fongetal:backprop} when constructing the category of learners. A special case also appears in the resource theory of contextuality as defined in~\cite{abramskyetal:comonadicview}, where one first defines deterministic free processes, and probabilistic (but classical) transformations $d\to e$ are then defined as transformations $d\otimes c\to e$ where $c$ is a non-contextual (and thus free) resource. This construction is discussed more generally in~\cite{cruttwell2022categorical}, but we modify it slightly by allowing one to choose a class of objects as ``parameters'' instead of taking that class to consist of all objects: this modification is important for resource theories as it lets one control which resources are made freely available.
\begin{prop}\label{prop:paraconstruction}
Let $\CC$ be an SMC and $\mathcal{X}$
 a class of objects that contains $I$ and is closed under $\otimes$. Then there is an SMC whose objects are those of $\CC$, and whose morphisms $A\to B$ are given by equivalence classes of morphisms $A\otimes X\to B$ in $\CC$ with $X\in\mathcal{X}$, where $f\colon A\otimes X\to B$,$f'\colon A\otimes X'\to B$ are equivalent if there is an isomorphism $g\colon X\to X'$ such that $f=f'\circ(\id[A]\otimes g)$
\end{prop}

\begin{proof}[Proof sketch] The composites $g\circ f$ and $g\otimes f$ are depicted by
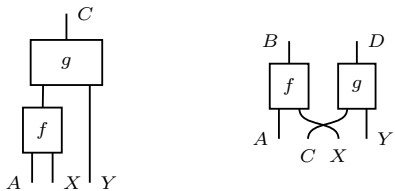
\[
  \begin{pic}
  \node[morphism] (f) at (0,0) {$f$};
  \setlength\minimummorphismwidth{6mm}
  \node[morphism] (g) at (0.32,.9) {$g$};
  \draw ([xshift=-2pt]f.south west) to +(0,-.4) node[left] {$A$};
    \draw ([xshift=2pt]f.south east) to +(0,-.4) node[right] {$X$};
    \draw (f) to ([xshift=-2pt]g.south west);
    \draw ([xshift=2pt]g.south east) to +(0,-1.3) node[right] {$Y$};
    \draw (g.north) to +(0,.35) node[right] {$C$};
  \end{pic}
  \qquad\qquad
  \begin{pic}
  \node[morphism] (f) at (0,0) {$f$};
  \node[morphism] (g) at (.9,0) {$g$};
  \draw ([xshift=-2pt]f.south west) to +(0,-.4) node[left] {$A$};
    \draw ([xshift=2pt]f.south east) to[out=-90,in=90] +(0.52,-.4) node[below] {$X$};
    \draw (f.north) to +(0,.3) node[left] {$B$};
    \draw (g.north) to +(0,.3) node[right] {$D$};
    \draw ([xshift=-2pt]g.south west) to[out=-90,in=90] +(-0.52,-.4) node[below] {$C$};
    \draw ([xshift=2pt]g.south east) to +(0,-.4) node[right] {$Y$};
  \end{pic}
\]
It is easy to show graphically that these are well-defined and that this results in an SMC. \end{proof}
Using Proposition~\ref{prop:paraconstruction} we can easily model protocols that have free access to some cryptographic functionalities: one just declares a class $\mathcal{X}$ of functionalities (\eg pairwise communication channels) that is closed under $\otimes$ to be free. In that case a protocol acting on $(A_{i=1}^n,r)$ can be depicted by
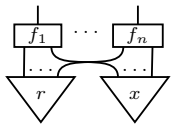
\[
  \begin{pic}
    \setlength\minimumstatewidth{9mm}
    \setlength\morphismheight{3mm}
    \node[state] (r) at (0,0) {$r$};
    \node[state] (x) at (1.25,0) {$x$};
    \node (a) at (0.6,.6) {$\ldots$};
    \node (b) at (0,.1) {$\ldots$};
    \node (c) at (1.25,.1) {$\ldots$};
    \node[morphism,font=\tiny] (f) at (-0.04,.55) {$f_1$};
    \node[morphism,font=\tiny] (g) at (1.29,.55) {$f_n$};
    \draw ([xshift=-0pt]r.A) to ([xshift=-2pt]f.south west);
     \draw ([xshift=-0pt]x.A) to [in=-90,out=90] ([xshift=2pt]f.south east);
    \draw ([xshift=0pt]r.B) to [in=-90,out=90]  ([xshift=-2pt]g.south west);
    \draw ([xshift=0pt]x.B) to ([xshift=2pt]g.south east);
    \draw (f.north) to +(0,.25);
    \draw (g.north) to +(0,.25);
  \end{pic}
\]
where $x\in\mathcal{X}$ is a free resource. For example, if we are in a multi-party setting and want to treat secure pairwise communication as a solved problem as is common in multi-party computation, we could let $\mathcal{X}$ be generated by all pairwise channels.

\section{The one-time pad}\label{sec:otp}

We will now explore how the one-time pad (OTP) fits into our framework, paralleling the discussion of the OTP in~\cite{Mau11}. We will start from the category $\cat{FinStoch}$ of finite sets and stochastic maps between them, with $\otimes$ given by cartesian product of sets. This is sufficient for the OTP over a fixed message space, even if more complicated and interactive cryptographic protocols will need a different starting category. However, the actual category~$\CC$ we work in is built from $\cat{FinStoch}$ as we want our resources to be morphisms of $\cat{FinStoch}$ with a tripartite structure. As we've seen at the end of Section~\ref{sec:resourcetheories}, this can be achieved by working with the resource theory $\CC$ induced by
\[\ncomb(\cat{FinStoch}^3)\xrightarrow{\ncomb(\otimes)}\ncomb(\cat{FinStoch})\xrightarrow{\hom(I,-)}\Set,\] with the three parties named Eve, Alice and Bob. We will equip $\ncomb(\cat{FinStoch}^3)$ with the attack model from Definition~\ref{def:maliciousparty_ncomb}, with Eve being the malicious party.

To recap, a basic resource in this resource theory is given by finite sets $E_i$,$A_i,B_i$ for $i=1,2$, and  a map $f\colon E_1\otimes A_1\otimes B_1\to E_2\otimes A_2\otimes B_2$ in $\cat{FinStoch}$, depicted as 
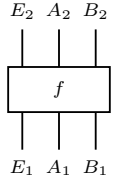
\[
  \begin{pic}
    \setlength\minimummorphismwidth{10mm}
    \node[morphism] (f) at (0,0) {$f$};
    \draw ([xshift=-2.5pt]f.south west) to ++(0,-.5) node[below] {$E_1$};
    \draw ([xshift=-2.5pt]f.north west) to ++(0,.5) node[above] {$E_2$};
    \draw (f.south) to ++(0,-.5) node[below] {$A_1$};
    \draw (f.north) to ++(0,.5) node[above] {$A_2$};
    \draw ([xshift=2.5pt]f.south east) to ++(0,-.5) node[below] {$B_1$};
    \draw ([xshift=2.5pt]f.north east) to ++(0,.5) node[above] {$B_2$};
  \end{pic}\]
The intuition is that $\tuple{(A_i,B_i,E_i)_{i\in\{1,2\}},f}$ represents a box shared by Alice, Bob and Eve, with Alice's inputs and outputs ranging over $A_1$ and $A_2$ respectively, and similarly for Bob and Eve. A general object of $\CC$ then consists of a list of such basic objects, representing a list of such resources shared between Alice, Bob and Eve. We will often label the ports just by the party who accesses it, and omit labeling trivial ports. For example, if Figure~\ref{fig:copy_map}
\begin{figure}
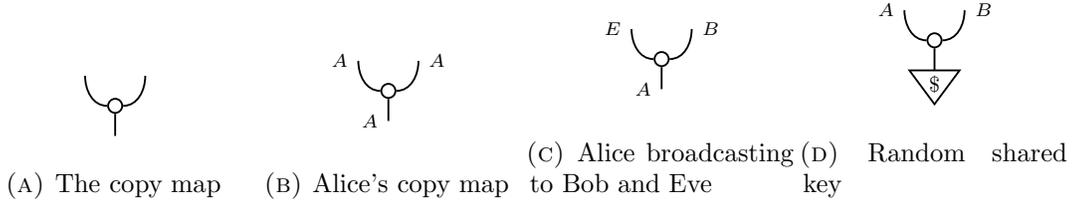

    \centering
\begin{subfigure}[b]{0.23\textwidth}
\[
  \begin{pic}[scale=.4,yscale=-1]
    \node[dot,fill=\copycolor] (d) {};
    \draw (d) to +(0,1);
    \draw (d) to[out=0,in=90] +(1,-1);
    \draw (d) to[out=180,in=90] +(-1,-1);
  \end{pic}
  \]
  \caption{The copy map}
  \label{fig:copy_map}
  \end{subfigure}
  ~
    \begin{subfigure}[b]{0.23\textwidth}
\[
  \begin{pic}[scale=.4,yscale=-1]
    \node[dot,fill=\copycolor] (d) {};
    \draw (d) to +(0,1) node [left] {$A$};
    \draw (d) to[out=0,in=90] +(1,-1) node[right] {$A$};
    \draw (d) to[out=180,in=90] +(-1,-1) node[left] {$A$};
  \end{pic}
  \]
        \caption{Alice's copy map}
        \label{fig:private_copy}
    \end{subfigure}
    ~
    \begin{subfigure}[b]{0.23\textwidth}
\[
  \begin{pic}[scale=.4,yscale=-1]
    \node[dot,fill=\copycolor] (d) {};
    \draw (d) to +(0,1) node [left] {$A$};
    \draw (d) to[out=0,in=90] +(1,-1) node[right] {$B$};
    \draw (d) to[out=180,in=90] +(-1,-1) node[left] {$E$};
  \end{pic}
  \]
        \caption{Alice broadcasting to Bob and Eve}
        \label{fig:insecure_channel}
    \end{subfigure}
   ~\begin{subfigure}[b]{0.23\textwidth}
\[
  \begin{pic}[scale=.4]
    \node[dot,fill=\copycolor] (d) {};
    \draw (d) to +(0,-1) node[state,scale=0.75] {\normalsize$\rand$};
    \draw (d) to[out=180,in=-90] +(-1,1) node[left] {$A$};
    \draw (d) to[out=0,in=-90] +(1,1) node[right] {$B$};
  \end{pic}
  \]
        \caption{Random shared key}
        \label{fig:shared_key}
   \end{subfigure}
   \caption{Variants of the copy map}
\end{figure}
depicts the copy map $X\to X\otimes X$ for some set $X$ in $\cat{FinStoch}$, then Figure~\ref{fig:private_copy}
denotes an object of $\CC$ representing Alice copying data in $X$ privately, whereas Figure~\ref{fig:insecure_channel}
denotes an object $\CC$ that sends Alice's input unchanged to Bob and to Eve---which we view as an insecure (but authenticated) channel from Alice to Bob.

In the version of the one-time pad we discuss, our starting resources consist of an insecure but authenticated channel\footnote{If the insecure channel allows Eve to tamper with the message, the analysis changes, as the resulting channel still lets Eve flip bits in the sent message, even if she cannot read the message. Consequently, one can study the OTP with a different domain and codomain than chosen here.} from Alice to Bob as in Figure~\ref{fig:insecure_channel}
and of a random key over the same message space, shared by Alice and Bob, depicted in Figure~\ref{fig:shared_key}. Here Alice and Bob having access to both of these resources is modelled by them sharing the monoidal product of these resources.
The goal is to build a secure channel
  $\begin{pic}[scale=.4] \draw (0,0) node[left] {$A$} to (0,3) node[right] {$B$}; \end{pic}$
from Alice to Bob from these.

We will next describe the local ingredients needed for OTP. First of all, Alice and Bob must agree on a group structure on the message space: this consists of a multiplication $\tinymult$ with unit $\tinyunit$ that is associative and unital, \ie satisfies the equations
  \begin{equation}\label{eq:monoid}
    \begin{pic}[scale=.4]
      \node[dot,fill=\multcolor] (t) at (0,1) {};
      \node[dot,fill=\multcolor] (b) at (1,0) {};
      \draw (t) to +(0,1);
      \draw (t) to[out=0,in=90] (b);
      \draw (t) to[out=180,in=90] (-1,0) to (-1,-1);
      \draw (b) to[out=180,in=90] (0,-1);
      \draw (b) to[out=0,in=90] (2,-1);
    \end{pic}
    =
    \begin{pic}[yscale=.4,xscale=-.4]
      \node[dot,fill=\multcolor] (t) at (0,1) {};
      \node[dot,fill=\multcolor] (b) at (1,0) {};
      \draw (t) to +(0,1);
      \draw (t) to[out=0,in=90] (b);
      \draw (t) to[out=180,in=90] (-1,0) to (-1,-1);
      \draw (b) to[out=180,in=90] (0,-1);
      \draw (b) to[out=0,in=90] (2,-1);
    \end{pic}
  \qquad\qquad
  \begin{pic}[scale=.4]
    \node[dot,fill=\multcolor] (d) {};
    \draw (d) to +(0,1);
    \draw (d) to[out=0,in=90] +(1,-1);
    \draw (d) to[out=180,in=90] +(-1,-1) node[dot,fill=\multcolor] {};
  \end{pic}
  =
  \begin{pic}[scale=.4]
    \draw (0,0) to (0,3);
  \end{pic}
  =
  \begin{pic}[yscale=.4,xscale=-.4]
    \node[dot,fill=\multcolor] (d) {};
    \draw (d) to +(0,1);
    \draw (d) to[out=0,in=90] +(1,-1);
    \draw (d) to[out=180,in=90] +(-1,-1) node[dot,fill=\multcolor] {};
  \end{pic}
  \end{equation}
Note that copying and deleting, denoted by a different color for convenience, satisfy similar equations
  \begin{equation}\label{eq:comonoid}
    \begin{pic}[xscale=.4,yscale=-.4]
      \node[dot,fill=\copycolor] (t) at (0,1) {};
      \node[dot,fill=\copycolor] (b) at (1,0) {};
      \draw (t) to +(0,1);
      \draw (t) to[out=0,in=90] (b);
      \draw (t) to[out=180,in=90] (-1,0) to (-1,-1);
      \draw (b) to[out=180,in=90] (0,-1);
      \draw (b) to[out=0,in=90] (2,-1);
    \end{pic}
    =
    \begin{pic}[yscale=-.4,xscale=-.4]
      \node[dot,fill=\copycolor] (t) at (0,1) {};
      \node[dot,fill=\copycolor] (b) at (1,0) {};
      \draw (t) to +(0,1);
      \draw (t) to[out=0,in=90] (b);
      \draw (t) to[out=180,in=90] (-1,0) to (-1,-1);
      \draw (b) to[out=180,in=90] (0,-1);
      \draw (b) to[out=0,in=90] (2,-1);
    \end{pic}
  \qquad\qquad
  \begin{pic}[scale=.4,yscale=-1]
    \node[dot,fill=\copycolor] (d) {};
    \draw (d) to +(0,1);
    \draw (d) to[out=0,in=90] +(1,-1);
    \draw (d) to[out=180,in=90] +(-1,-1) node[dot,fill=\copycolor] {};
  \end{pic}
  =
  \begin{pic}[scale=.4]
    \draw (0,0) to (0,3);
  \end{pic}
  =
  \begin{pic}[yscale=-.4,xscale=-.4]
    \node[dot,fill=\copycolor] (d) {};
    \draw (d) to +(0,1);
    \draw (d) to[out=0,in=90] +(1,-1);
    \draw (d) to[out=180,in=90] +(-1,-1) node[dot,fill=\copycolor] {};
  \end{pic}
  \end{equation}
 In addition, multiplication and copying interact:
\begin{equation}\label{eq:bialg}
  \begin{pic}[scale=.4]
    \node[dot,fill=\multcolor] (d) at (-0.75,1) {};
    \node[dot,fill=\copycolor] (e) at (-0.75,-1) {};
    \draw (d) to +(0,1);
    \draw (d) to[out=210,in=90] ++(-.75,-1) to[out=-90,in=150] (e);
    \draw (e) to +(0,-1);
    \node[dot,fill=\multcolor] (d1) at (0.75,1) {};
    \node[dot,fill=\copycolor] (e1) at (0.75,-1) {};
    \draw (d1) to[out=-30,in=90] ++(.75,-1) to[out=-90,in=30] (e1);
    \draw (e1) to +(0,-1);
    \draw (d1) to +(0,1);
    \draw (d) to (e1);
    \draw (e) to (d1);
  \end{pic}
  =
  \begin{pic}[scale=.4]
    \node[dot,fill=\copycolor] (r) at (0,.75) {};
    \node[dot,fill=\multcolor] (d) at (0,-.75) {};
    \draw (d) to (r);
    \draw (d) to[out=0,in=90] +(1,-1);
    \draw (d) to[out=180,in=90] +(-1,-1);
    \draw (r) to[out=0,in=-90] +(1,1);
    \draw (r) to[out=180,in=-90] +(-1,1);
  \end{pic}
    \end{equation}
 and the map $\begin{pic}
  \node[morphism] (f) at (0,0) {$i$};
  \draw (f.south) to +(0,-.1);
  \draw (f.north) to +(0,.1);
\end{pic}$ giving inverses\footnote{Usually for OTP one works over a power of $\mathbb{Z}_2$ so that $i$ is given by the identity map.} satisfies
\begin{equation}\label{eq:antipode}
  \begin{pic}[scale=.4,xscale=-1]
    \node[dot,fill=\multcolor] (d) at (0,1) {};
    \node[dot,fill=\copycolor] (e) at (0,-1) {};
    \draw (d) to +(0,1);
    \draw (d) to[out=0,in=90] ++(1,-1)  node[morphism,scale=.5] {\large $i$} to[out=-90,in=0] (e);
    \draw (d) to[out=180,in=90] ++(-1,-1) to[out=-90,in=180] (e);
    \draw (e) to +(0,-1);
  \end{pic}
  =
  \begin{pic}[scale=.4]
    \node[dot,fill=\multcolor] (r) at (0,1) {};
    \node[dot,fill=\copycolor] (d) at (0,-1) {};
    \draw (d) to +(0,-1);
    \draw (r) to +(0,1);
  \end{pic}
  =
\begin{pic}[scale=.4]
    \node[dot,fill=\multcolor] (d) at (0,1) {};
    \node[dot,fill=\copycolor] (e) at (0,-1) {};
    \draw (d) to +(0,1);
    \draw (d) to[out=0,in=90] ++(1,-1)  node[morphism,scale=.5] {\large $i$} to[out=-90,in=0] (e);
    \draw (d) to[out=180,in=90] ++(-1,-1) to[out=-90,in=180] (e);
    \draw (e) to +(0,-1);
  \end{pic}
  \end{equation}
That the key is uniformly random is captured by
  
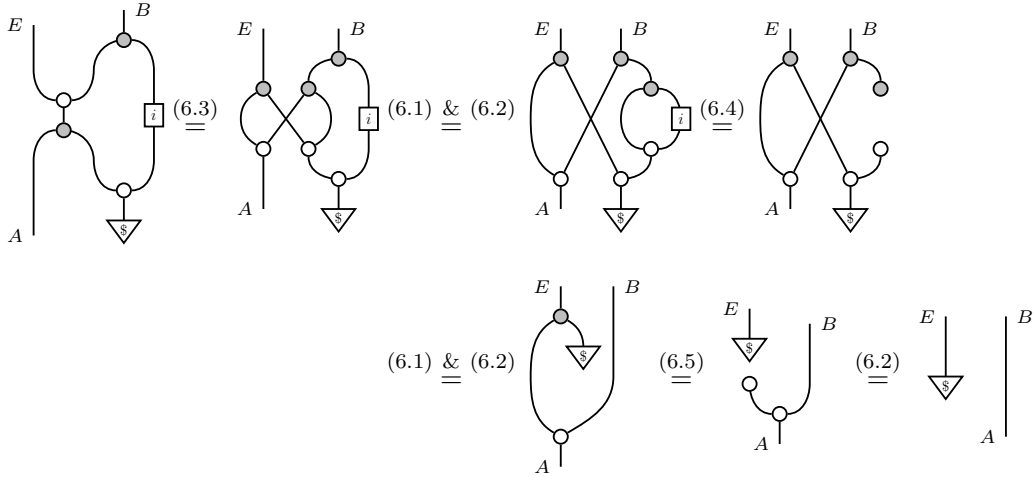
\begin{equation}\label{eq:uniformly_random}
  \begin{pic}[scale=.4]
    \node[dot,fill=\multcolor] (d) {};
    \draw (d) to +(0,1);
    \draw (d) to[out=0,in=90] +(1,-1);
    \draw (d) to[out=180,in=90] +(-1,-1) node[state,scale=0.5] {\normalsize$\rand$};
  \end{pic}
  =
  \begin{pic}[scale=.4]
    \node[state,scale=0.5] (r) at (0,1) {\normalsize$\rand$};
    \node[dot,fill=\copycolor] (d) at (0,-1) {};
    \draw (d) to +(0,-1);
    \draw (r) to +(0,1);
  \end{pic}
  =
  \begin{pic}[yscale=.4,xscale=-.4]
    \node[dot,fill=\multcolor] (d) {};
    \draw (d) to +(0,1);
    \draw (d) to[out=0,in=90] +(1,-1);
    \draw (d) to[out=180,in=90] +(-1,-1) node[state,scale=0.5] {\normalsize$\rand$};
  \end{pic},\end{equation}
 which amounts to saying that ``adding uniform noise to a channel results in uniform noise''. Moreover, producing a random key and deleting it amounts to doing nothing:
 \begin{equation}\label{eq:deletion}
 \begin{pic}[scale=.4]
 \node[dot,fill=\copycolor] (f) at (0,0) {};
  \draw (f) to ++(0,-1.25) node[state,scale=0.75] {\normalsize$\rand$};
 \end{pic}=\quad \end{equation}

 Taken together, the equations describe a structure known as a Hopf algebra with an integral~\cite{sweedler:integrals} in a symmetric monoidal category. In $\cat{FinStoch}$, if one keeps the meaning of the copy maps fixed, such structures correspond exactly to finite groups equipped with the uniform distribution, and any such group can be used to implement the one-time pad with the message space given by the set underlying the group. Concretely, this means that Alice and Bob must agree on a group structure on the message space, and the fact that this multiplication forms a group and that the key is random can be captured by these equations.

The OTP protocol is then depicted as follows:
\[
  \begin{pic}[scale=.4]{}
    \node[dot,fill=\copycolor] (d) {};
    \draw (d) to +(0,-1) node[state,scale=0.5] {\normalsize$\rand$};
    \draw (d) to[out=0,in=-90] ++(1,1) node[right] {$B$} to ++(0,1.5) node[morphism,scale=.5] (i) {\large $i$};
    \draw (d) to[out=180,in=-90] ++(-1,1) node[left] {$A$} to[out=90,in=0] ++(-1,1) node[dot,fill=\multcolor] (e) {};
    \draw (e) to[out=180,in=90] ++(-1,-1) to ++(0,-2.5) node[left] {$A$};
    \draw (e) to ++(0,1) node[dot,fill=\copycolor] (f) {};
    \draw (f) to[out=180,in=-90] ++(-1,1) node[left] {$E$} to ++(0,.5) node[dot,fill=\copycolor] {};
    \draw (f) to[out=0,in=-90] ++(1,1) node[right] {$B$} to[out=90,in=180] ++(1,1) node[dot,fill=\multcolor] (g) {};
    \draw (g) to +(0,1) node[right] {$B$};
    \draw (g) to[out=0,in=90] ++(1,-1) to (i);
  \end{pic}
  \]
\ie Alice adds the key to her message, broadcasts it to Eve and Bob. Eve deletes her part and Bob adds the inverse of the key to the ciphertext to recover the message.

That the protocol is correct (\ie works when everyone follows it) can be proven as follows:
\[ \begin{pic}[scale=.4]
    \node[dot,fill=\copycolor] (d) {};
    \draw (d) to +(0,-1) node[state,scale=0.5] {\normalsize$\rand$};
    \draw (d) to[out=0,in=-90] ++(1,1)  to ++(0,1.5) node[morphism,scale=.5] (i) {\large $i$};
    \draw (d) to[out=180,in=-90] ++(-1,1) to[out=90,in=0] ++(-1,1) node[dot,fill=\multcolor] (e) {};
    \draw (e) to[out=180,in=90] ++(-1,-1) to  ++(0,-2.5) node[left] {$A$};
    \draw (e) to ++(0,1) node[dot,fill=\copycolor] (f) {};
    \draw (f) to[out=180,in=-90] ++(-1,1)  to ++(0,.5) node[dot,fill=\copycolor] {};
    \draw (f) to[out=0,in=-90] ++(1,1) to[out=90,in=180] ++(1,1) node[dot,fill=\multcolor] (g) {};
    \draw (g) to +(0,1) node[right] {$B$};
    \draw (g) to[out=0,in=90] ++(1,-1) to (i);
  \end{pic}\ \overset{\eqref{eq:comonoid}}{=}
    \begin{pic}[scale=.4]
    \node[dot,fill=\copycolor] (d) {};
    \draw (d) to +(0,-1) node[state,scale=0.5] {\normalsize$\rand$};
    \draw (d) to[out=0,in=-90] ++(1,1) to ++(0,1.5) node[morphism,scale=.5] (i) {\large $i$};
    \draw (d) to[out=180,in=-90] ++(-1,1) to[out=90,in=0] ++(-1,1) node[dot,fill=\multcolor] (e) {};
    \draw (e) to[out=180,in=90] ++(-1,-1) to  ++(0,-3) node[left] {$A$};
    \draw (e) to[out=90,in=180] ++(2,3) node[dot,fill=\multcolor] (g) {};
    \draw (g) to +(0,1) node[right] {$B$};
    \draw (g) to[out=0,in=90] ++(1,-1) to (i);
  \end{pic}
    \ \overset{\eqref{eq:monoid}}{=}
  \begin{pic}[scale=.4]
    \node[dot,fill=\multcolor] (d) at (0,1) {};
    \node[dot,fill=\copycolor] (e) at (0,-1) {};
    \draw (d) to[out=90,in=0] ++(-1,1) node[dot,fill=\multcolor] (f) {};
    \draw (d) to[out=0,in=90] ++(1,-1)  node[morphism,scale=.5] {\large $i$} to[out=-90,in=0] (e);
    \draw (d) to[out=180,in=90] ++(-1,-1) to[out=-90,in=180] (e);
    \draw (e) to +(0,-1) node[state,scale=.5] {\normalsize$\rand$};
    \draw (f) to +(0,1) node[right] {$B$};
    \draw (f) to[out=180,in=90] ++(-1,-1) to ++(0,-3)  to++(0,-.5) node[left] {$A$};
  \end{pic}
  \  \overset{\eqref{eq:antipode}}{=}
  \begin{pic}[scale=.4]
    \node[dot,fill=\multcolor] (d) at (0,1) {};
    \node[dot,fill=\copycolor] (e) at (0,-1) {};
    \draw (d) to[out=90,in=0] ++(-1,1) node[dot,fill=\multcolor] (f) {};
    \draw (e) to +(0,-1) node[state,scale=.5] {\normalsize$\rand$};
    \draw (f) to +(0,1) node[right] {$B$};
    \draw (f) to[out=180,in=90] ++(-1,-1) to ++(0,-3)  to++(0,-.5) node[left] {$A$};
  \end{pic}
\overset{\eqref{eq:monoid}\ \&\ \eqref{eq:deletion}}{=}
  \begin{pic}[scale=.4]
    \draw (0,0) node[left] {$A$} to (0,3) node[right] {$B$};
  \end{pic}
  \]

However, we also need to show that the protocol is secure. In this case, Eve has an initial attack given by just reading the ciphertext. We can now prove security pictorially as follows:
\begin{align*}
  \begin{pic}[scale=.4]
    \node[dot,fill=\copycolor] (d) {};
    \draw (d) to +(0,-1) node[state,scale=0.5] {\normalsize$\rand$};
    \draw (d) to[out=0,in=-90] ++(1,1) to ++(0,1.5) node[morphism,scale=.5] (i) {\large $i$};
    \draw (d) to[out=180,in=-90] ++(-1,1) to[out=90,in=0] ++(-1,1) node[dot,fill=\multcolor] (e) {};
    \draw (e) to[out=180,in=90] ++(-1,-1) to  ++(0,-2.5) node[left] {$A$};
    \draw (e) to ++(0,1) node[dot,fill=\copycolor] (f) {};
    \draw (f) to[out=180,in=-90] ++(-1,1)  to ++(0,1.5) node[left] {$E$};
    \draw (f) to[out=0,in=-90] ++(1,1) to[out=90,in=180] ++(1,1) node[dot,fill=\multcolor] (g) {};
    \draw (g) to +(0,1) node[right] {$B$};
    \draw (g) to[out=0,in=90] ++(1,-1) to (i);
  \end{pic}
   \overset{\eqref{eq:bialg}}{=}
    \begin{pic}[scale=.4]
    \node[dot,fill=\multcolor] (d) at (-0.75,1) {};
    \node[dot,fill=\copycolor] (e) at (-0.75,-1) {};
    \draw (d) to ++(0,1) to ++(0,1) node[left] {$E$} ;
    \draw (d) to[out=210,in=90] ++(-.75,-1) to[out=-90,in=150] (e);
    \draw (e) to ++(0,-1)  to ++(0,-1) node[left] {$A$};
    \node[dot,fill=\multcolor] (d1) at (0.75,1) {};
    \node[dot,fill=\copycolor] (e1) at (0.75,-1) {};
    \draw (d) to (e1);
    \draw (e) to (d1);
    \draw (d1) to[out=-30,in=90] ++(.75,-1) to[out=-90,in=30] (e1);
    \draw (e1) to[out=-90,in=180] ++(1,-1) node[dot,fill=\copycolor] (f) {};
    \draw (d1) to[out=90,in=180] ++(1,1) node[dot,fill=\multcolor] (g) {};
    \draw (f) to ++(0,-1) node[state,scale=.5] {\normalsize$\rand$};
    \draw (f) to[out=0,in=-90] ++(1,1) to ++(0,1) node[morphism,scale=.5] (i) {\large $i$};
    \draw (g) to[out=0,in=90] ++(1,-1) to (i);
    \draw (g) to ++(0,1) node[right] {$B$};
  \end{pic}
  &
   \overset{\eqref{eq:monoid}\ \&\ \eqref{eq:comonoid}}{=}
    \begin{pic}[scale=.4]
    \node[dot,fill=\multcolor] (d) at (1,1) {};
    \node[dot,fill=\copycolor] (e) at (1,-1) {};
    \node[dot,fill=\multcolor] (a) at (-2,2) {};
    \node[dot,fill=\copycolor] (b) at (-2,-2) {};
    \draw (d) to[out=90,in=0] ++(-1,1) node[dot,fill=\multcolor] (f) {};
    \draw (d) to[out=0,in=90] ++(1,-1)  node[morphism,scale=.5] {\large $i$} to[out=-90,in=0] (e);
    \draw (d) to[out=180,in=90] ++(-1,-1) to[out=-90,in=180] (e);
    \draw (e) to[out=-90,in=0] ++(-1,-1) node[dot,fill=\copycolor] (g) {};
    \draw (f) to +(0,1) node[right] {$B$};
    \draw (f) to (b);
    \draw (g) to +(0,-1) node[state,scale=.5] {\normalsize$\rand$};
    \draw (g) to (a);
    \draw (a) to +(0,1) node[left] {$E$};
    \draw (b) to +(0,-1) node[left] {$A$};
    \draw (a) to[in=90,out=210] ++(-1,-2) to[in=150,out=-90] (b);
  \end{pic}
  \overset{\eqref{eq:antipode}}{=}
      \begin{pic}[scale=.4]
    \node[dot,fill=\multcolor] (d) at (1,1) {};
    \node[dot,fill=\copycolor] (e) at (1,-1) {};
    \node[dot,fill=\multcolor] (a) at (-2,2) {};
    \node[dot,fill=\copycolor] (b) at (-2,-2) {};
    \draw (d) to[out=90,in=0] ++(-1,1) node[dot,fill=\multcolor] (f) {};
    \draw (e) to[out=-90,in=0] ++(-1,-1) node[dot,fill=\copycolor] (g) {};
    \draw (f) to +(0,1) node[right] {$B$};
    \draw (f) to (b);
    \draw (g) to +(0,-1) node[state,scale=.5] {\normalsize$\rand$};
    \draw (g) to (a);
    \draw (a) to +(0,1) node[left] {$E$};
    \draw (b) to +(0,-1) node[left] {$A$};
    \draw (a) to[in=90,out=210] ++(-1,-2) to[in=150,out=-90] (b);
  \end{pic}
  \\
  &
  \overset{\eqref{eq:monoid}\ \&\ \eqref{eq:comonoid}}{=}
    \begin{pic}[scale=.4]
    \node[dot,fill=\multcolor] (d) at (0,2) {};
    \node[dot,fill=\copycolor] (e) at (0,-2) {};
    \draw (d) to +(0,1) node[left] {$E$};
    \draw (d) to[out=-30,in=90] ++(.75,-1)  node[state,scale=.5] {\normalsize$\rand$};
    \draw (d) to[in=90,out=210] ++(-1,-2) to[in=150,out=-90] (e);
    \draw (e) to +(0,-1) node[left] {$A$};
    \draw (e) to[out=30,in=-90] ++(1.75,2) to ++(0,3) node[right] {$B$};
  \end{pic}
 \overset{\eqref{eq:uniformly_random}}{=}
   \begin{pic}[scale=.4]
  \node[state,scale=.5] (d) at (0,1.5) {\normalsize$\rand$};
    \draw (d) to +(0,1) node[left] {$E$};
    \node[dot,fill=\copycolor] (a) at (1,-1) {};
    \draw (a) to[out=180,in=-90] ++(-1,1) node[dot,fill=\copycolor] {};
    \draw (a) to ++(0,-1) node[left] {$A$};
    \draw (a) to[out=0,in=-90] ++(1,1) to ++(0,2) node[right] {$B$};
  \end{pic}
 \overset{\eqref{eq:comonoid}}{=}
 \begin{pic}[scale=.4]
    \draw (0,0) node[left] {$A$} to (0,4) node[right] {$B$};
    \node[state,scale=.5] (d) at (-2,2) {\normalsize$\rand$};
    \draw (d) to +(0,2) node[left] {$E$};
  \end{pic}
\end{align*}
Taken together, these show that Eve's initial attack is equal to her just producing a random message herself when Alice and Bob share the target resource.

Thus OTP represents a map \[\text{authenticated channel}\otimes \text{shared key}\to\text{secure channel}\] that is secure against Eve. In pictures, we might say that OTP is a map
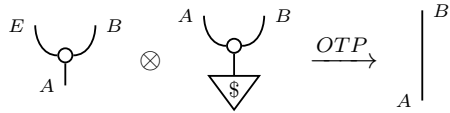
\[ \begin{pic}[scale=.4,yscale=-1]
    \node[dot,fill=\copycolor] (d) {};
    \draw (d) to +(0,1) node [left] {$A$};
    \draw (d) to[out=0,in=90] +(1,-1) node[right] {$B$};
    \draw (d) to[out=180,in=90] +(-1,-1) node[left] {$E$};
  \end{pic}
  \otimes
    \begin{pic}[scale=.4]
    \node[dot,fill=\copycolor] (d) {};
    \draw (d) to +(0,-1) node[state,scale=0.75] {\normalsize$\rand$};
    \draw (d) to[out=180,in=-90] +(-1,1) node[left] {$A$};
    \draw (d) to[out=0,in=-90] +(1,1) node[right] {$B$};
  \end{pic}
  \xrightarrow{OTP}
    \begin{pic}[scale=.4] \draw (0,0) node[left] {$A$} to (0,3) node[right] {$B$}; \end{pic}
  \]

\section{Diffie-Hellman key exchange}\label{sec:dhke}

We now turn our attention to the problem of obtaining keys in the first place, by modelling Diffie-Hellman key exchange (DHKE). Let us first recall DHKE as usually described. First of all, Alice and Bob agree on a finite cyclic group $G$ with generator $g$ (in fact, the group might depend on the security parameter, so one could think of them as agreeing on a sequence of groups). Alice and Bob uniformly sample $a$ and $b$ from $\Z_n$ where $n=\left\vert G\right\vert$. Alice and Bob then broadcast $g^a$ and $g^b$ over a public channel, after which Alice computes $k_a:=(g^b)^a$ and Bob computes $k_b:=(g^a)^b$. As $k_a=k_b=g^{ab}$, they've now agreed on a key. If the group (or sequence of them) is chosen well, the distribution $(g^a,g^b,g^{ab})$ for uniformly sampled $a,b$  ``looks like'' the distribution $(g^a,g^b,g^c)$ for uniformly sampled $a,b,c$, and this is in fact often taken as the definition of security~\cite[Section 7.3.2]{KL15}.

\subsection{The category of efficient probabilistic computations}

We will now build the category we will work in. While $\cat{FinStoch}$ was a workable starting point for the OTP, it is no longer good enough for DHKE as the security constraints are computational and hence require asymptotic notions. Informally, we will work with the category whose maps are ``efficient sequences of stochastic maps'', where we identify efficiency with computability in polynomial-time as usual. As a first approximation, one might think of the larger category $[\N,\cat{FinStoch}]$ of all sequences of stochastic maps (here $\N$ is viewed as a discrete category) and then restrict to the efficient sequences therein. However, one usually defines polynomial-time computability for functions on some standard set (\eg $\N$ or $\{0,1\}^*$) and then lifts this to other sets by encoding. We capture this in the following definition\footnote{Our construction is a variant of a special case from~\cite{dusko:crypto}---the most important difference is that we define indistinguishability in terms of negligible distinguisher advantage (as is standard), rather than in terms of negligible $\ell^\infty$-distance.}.

\begin{defi}
An \emph{efficient sequence of finite sets} consists of a sequence of injections $(\sem{-}_n\colon A_n\to \{0,1\}^*$) where each $A_n$ is a finite set and the characteristic functions of the sets $\sem{A_n}_n\subset \{0,1\}^*$ can be computed in polynomial time. When there is no risk of confusion, we will often drop the subscript and write $\sem{-}$ instead of $\sem{-}_n$, and similarly we won't disambiguate notationally between the encoding functions $\sem{-}$ of different sets unless we really have to. When we think of the index as security parameter, we will often index with $\lambda$ rather than with $n$.

An efficient stochastic map $(A_n,\sem{-})_{n\in\N} \to (B_n,\sem{-})_{n\in \N}$ consists of a sequence of stochastic maps $(f_n\colon A_n\to B_n)_{n\in \N}$ such that sequence of maps $(\sem{f_n}\colon \sem{A_n}_n\to \sem{B_n}_n)_{n\in\N}$ defined by $\sem{f_n}(\sem{a})=\sem{f_n(a)}$ can be computed in probabilistic polynomial time.

We denote the category of efficient sequences of finite sets and efficient stochastic maps between them by $\Eff$. Fixing some efficient bijective pairing function $\inprod{-}{-}\colon\{0,1\}^*\times \{0,1\}^*\to \{0,1\}^*$, we will equip $\Eff$ with the structure of a symmetric monoidal category, where the monoidal structure on objects is induced by the cartesian product of sets, and the encoding of $(a,b)\in A_n\times B_n$ is given by $\inprod{\sem{a}}{\sem{b}}$.
\end{defi}

We will now formalize the notion of ``computational indistinguishability'' as a congruence on $\Eff$.  Recall that a function $f\colon \N\to\R^+$ is \emph{negligible} if for any exponent $c\in\N$ the function $f$ is eventually smaller than $1/x^c$.

\begin{defi}\label{def:indistuingishability}
Given two objects $A$ and $B$ of $\Eff$, a distinguisher of type $A\to B$  is informally speaking a program that tries to tell apart programs of type $A\to B$. More precisely, a distinguisher of type $A\to B$ is a probabilistic polynomial-time algorithm that queries an oracle (taking inputs in $A$ and giving outputs in $B$), and outputs a single bit. In particular, it can query the oracle multiple repeatedly, albeit only polynomially many times. Given a distinguisher $D$ of type $A\to B$ and a map $f\colon A\to B$ in $\Eff$, we let $D(f)$ denote the function $\N\to \R^+$ that given $n$, outputs the probability that $D$ outputs $1$ given access $f_n$ as its oracle. Two parallel maps $f,g\colon A\rightrightarrows B$ in $\Eff$ are \emph{computationally indistinguishable}, written $f\approx g$, if for any distinguisher $D$ of type $A\to B$, the function $\left\vert D(f)-D(g)\right\vert$ is negligible.
\end{defi}

As morphisms of $\Eff$ and our chosen distinguishers are by definition polynomial-time, for any morphism  $f$ of $\Eff$ one can always construct distinguishers that use $f$ as a subroutine. This is at the heart of the following proposition.
\begin{prop} Computational indistinguishability is a monoidal congruence on $\Eff$.
\end{prop}

\begin{proof}
It is obvious that $\approx$ is reflexive and symmetric. To see that it is transitive, assume that $f,g,h$ are maps $A\to B$ with $f\approx g\approx h$. Consider an arbitrary distinguisher $D$ of type $A\to B$. Then

\[\left\vert D(f)-D(h)\right\vert=\left\vert D(f)-D(g)+D(g)-D(h)\right\vert\leq \left\vert D(f)-D(g)\right\vert+\left\vert D(g)-D(h)\right\vert\]
As the sum of two negligible functions is negligible, we have $f\approx h$.

We now check that $\approx$ respects sequential composition. Consider morphisms $f,f'\colon A\rightrightarrows B$ and $g,g'\colon B\rightrightarrows C$ with $f\approx f'$ and $g\approx g'$. Now, as $f\approx f'$ we must have $gf\approx gf'$, as any distinguisher $D$ of  type $A\to C$ can be used to construct a new distinguisher $D_{g\circ -}$ of type $A\to B$, that behaves otherwise as $D$ but whenever it receives an answer in $B$ from the oracle, uses $g$ to transform it to an answer in $C$. Similarly, $gf\approx g' f'$. Taken together, we have $gf\approx gf'\approx g'f'$, so that by transitivity we have $gf\approx g'f'$.  The argument for parallel composition is similar.
\end{proof}

\subsection{Resources needed for DHKE}

As with the one-time pad, we want our resources to be morphisms of this base category (equipped with a chosen labeling of input and output ports telling which parties control/access which). Consequently, we will work with the resource theory given by $\ncomb(\Eff^3)\xrightarrow{\ncomb(\otimes)}\ncomb(\Eff)\xrightarrow{\hom(I,-)}\Equ$, with the attack model on $\ncomb(\Eff^3)$ as in Definition~\ref{def:maliciousparty_ncomb}, with the malicious party called Eve. While our protocol will actually be correct up to $=$, DHKE is only secure up to $\approx$, so we are working with security up to computational indistinguishability as in Section~\ref{sec:equivalence}.

We will now discuss the ingredients of the DHKE protocol. First of all, we can think of the sequence of groups $G=(G_\lambda)_{\lambda\in \N}$ as an object in $\Eff$, with the group structure giving rise to a similar structure as for the OTP, and the element $g\in G$ (\ie $g_\lambda\in G_\lambda$ for $\lambda\in\N$) is then modeled as a (deterministic) map $I\to G \in \Eff$.

However, Alice and Bob don't really need to use the whole group structure of $G$: rather, it's sufficient that they can compute the map $(a,h)\mapsto h^a$, where $h$ is any element of the group and $a$ is an integer (modulo $\left\vert G\right\vert$). We will model this by defining the object $\Z_n:=(\Z_{n(\lambda)})$ where $n(\lambda)=\left\vert G_\lambda\right\vert$ equipped with some encoding, so that the map $(a,h)\mapsto h^a$ becomes a map $act\colon\Z_n\otimes G\to G$ as in Figure~\ref{fig:action}.  In addition to these, Alice and Bob both need the availability to sample uniformly from $\Z_n$, which then becomes a state $r\colon I\to \Z_n$ as in Figure~\ref{fig:uniformint}, and to compute the generator $g$ as in Figure~\ref{fig:generator}.

\begin{figure}
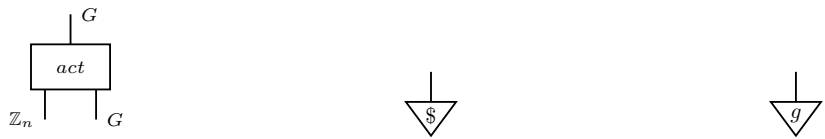

    \centering
    \begin{subfigure}[b]{0.31\textwidth}
\[
  \begin{pic}[scale=.4]
      \setlength\minimummorphismwidth{7mm}
      \node[morphism] (a) at (0,0) {$act$};
      \draw (a.north) to ++(0,1) node[right] {$G$};
      \draw ([xshift=-4.25pt]a.south west) to ++(0,-1) node[left] {$\Z_n$};
      \draw ([xshift=4.25pt]a.south east) to ++(0,-1) node[right] {$G$};
  \end{pic}
\]
        \caption{The action $\Z_n\otimes G\to G$.} 
        \label{fig:action}
    \end{subfigure}
    ~
\begin{subfigure}[b]{0.31\textwidth}
  \[
 \begin{pic}[scale=.4]
    \node[state,scale=.75] (d) at (0,0) {\normalsize$\rand$};
    \draw (d) to ++(0,1);
  \end{pic}
  \]
  \caption{Uniform distribution on $\Z_n$}
  \label{fig:uniformint}
  \end{subfigure}
  ~
    \begin{subfigure}[b]{0.31\textwidth}
\[
 \begin{pic}[scale=.4]
    \node[state,scale=.75] (d) at (0,0) {\normalsize$g$};
    \draw (d) to ++(0,1);
  \end{pic}
  \]
  \caption{The generator $g$ of $G$}
        \label{fig:generator}
    \end{subfigure}

   \caption{Private computations used in Diffie-Hellman key exchange}
\end{figure}

The shared resources needed for the protocol amount to just the broadcasting maps of type $G$, with one for Alice and one for Bob. Alice's broadcasting map is depicted in Figure~\ref{fig:insecure_channel}, and a depiction of Bob's broadcasting map is obtained by switching the roles of $A$ and $B$ in the pictures.
Similarly, the target resource is a uniformly random key $k$ sampled from  $G$, depicted in Fig~\ref{fig:shared_key}.

We now discuss some properties enjoyed by these structures. First of all, both $\Z_n$ and $G$ are groups on their own right, so that together with their copy maps they satisfy the equations~\eqref{eq:monoid}-\eqref{eq:antipode} from the previous section. Similarly, the fact that the private random distribution used by Alice and Bob is uniform is captured by the equation~\eqref{eq:uniformly_random}. Now, the map  $act\colon\Z_n\otimes G\to G$ makes $G$ into a module over $\Z_n$, which pictorially corresponds to the following two equations
  \begin{equation}\label{eq:module}
  \begin{pic}[scale=.4]
      \setlength\minimummorphismwidth{7mm}
      \node[morphism] (a) at (.85,2) {$act$};
      \node[dot,fill=\multcolor] (b) at (0,0) {};
      \draw (a.north) to ++(0,1) node[right] {$G$};
      \draw ([xshift=-4.25pt]a.south west) to (b);
      \draw ([xshift=4.25pt]a.south east) to ++(0,-3.15) node[right] {$G$};
      \draw (b) to[out=180,in=90] ++(-1,-1) to ++(0,-.85) node[left] {$\Z_n$};
      \draw (b) to[out=0,in=90] ++(1,-1) to ++(0,-.85) node[left] {$\Z_n$};
  \end{pic}
  =
  \begin{pic}[scale=.4]
      \setlength\minimummorphismwidth{7mm}
      \node[morphism] (a) at (0,2.15) {$act$};
      \node[morphism] (b) at (.85,0) {$act$};
      \draw (a.north) to ++(0,1) node[right] {$G$};
      \draw ([xshift=-4.25pt]a.south west)  to ++(0,-3.15) node[left] {$\Z_n$};
       \draw ([xshift=4.25pt]a.south east) to (b.north);
      \draw ([xshift=-4.25pt]b.south west)  to ++(0,-1) node[right] {$\Z_n$};
       \draw ([xshift=4.25pt]b.south east) to ++(0,-1) node[right] {$G$};
  \end{pic}
  \quad\text{and}\quad
    \begin{pic}[scale=.4]
      \setlength\minimummorphismwidth{7mm}
      \node[morphism] (a) at (.85,2) {$act$};
      \node[dot,fill=\multcolor] (b) at (0,0) {};
      \draw (a.north) to ++(0,1) node[right] {$G$};
      \draw ([xshift=-4.25pt]a.south west) to (b);
      \draw ([xshift=4.25pt]a.south east) to ++(0,-2) node[right] {$G$};
  \end{pic}
    =\quad
  \begin{pic}
    \draw (0,0) node[right] {$G$} to (0,1.8);
  \end{pic}
    \end{equation}

Moreover, we need the following general fact: copying the result of a deterministic function (whether $g$ or the action) is equal to copying the inputs and applying the function twice\footnote{In fact, this is often taken as the definition of deterministic morphisms in categorical approaches probability~\cite[Definition 10.1]{fritz:markovcats}.}. In pictures, this amounts to saying that if $f\colon A\to B$ is deterministic, then

\begin{equation}\label{eq:determinism}
  \begin{pic}[scale=.4]
    \node[dot,fill=\copycolor] (d) at (0,0) {};
    \draw (d) to ++(0,-1) node[left] {$A$};
    \draw (d) to[in=-90,out=180] ++(-1,1) to ++(0,.5) node[morphism,scale=.75] (f1)  {$f$};
    \draw (d) to[in=-90,out=0] ++(1,1)  to ++(0,.5)  node[morphism,scale=.75]  (f2) {$f$};
    \draw (f1.north) to ++(0,1) node[left] {$B$};
    \draw (f2.north) to ++(0,1) node[right] {$B$};
  \end{pic}
  =
     \begin{pic}[scale=.4]
    \node[morphism,scale=.75] (c) at (0,0) {\normalsize$f$}; 
    \draw (c.north) to ++(0,1) node[dot,fill=\copycolor] (d){};
    \draw (c.south) to ++(0,-1) node[left] {$A$};
    \draw (d) to[in=-90,out=180] ++(-1,1) node[left]  {$B$};
    \draw (d) to[in=-90,out=0] ++(1,1) node[right] {$B$};
  \end{pic}
\end{equation}

The final properties that we need are that multiplication on $\Z_n$ is commutative and that copying is ``cocommutative''. In pictures, these are captured by
\begin{equation}\label{eq:commutativity}
  \begin{pic}[scale=.4]
    \node[dot,fill=\multcolor] (d) {};
    \draw (d) to ++(0,1);
    \draw (d) to[out=0,in=90] ++(1,-1) to[out=-90,in=90] ++(-2,-2);
    \draw (d) to[out=180,in=90] ++(-1,-1) to[out=-90,in=90] ++(2,-2);
  \end{pic}
  =
  \begin{pic}[scale=.4]
    \node[dot,fill=\multcolor] (d) {};
    \draw (d) to ++(0,1.5);
    \draw (d) to[out=0,in=90] ++(1,-1) to ++(0,-.5);
    \draw (d) to[out=180,in=90] ++(-1,-1) to ++(0,-.5);
  \end{pic}
\end{equation}
and
\begin{equation}\label{eq:cocommmutativity}
    \begin{pic}[scale=.4]
    \node[dot,fill=\copycolor] (d) {};
    \draw (d) to +(0,-1);
    \draw (d) to[out=180,in=-90] ++(-1,1) to[out=90,in=-90] ++(2,2);
    \draw (d) to[out=0,in=-90] ++(1,1)  to[out=90,in=-90] ++(-2,2);
  \end{pic}
    =
    \begin{pic}[scale=.4]
    \node[dot,fill=\copycolor] (d) {};
    \draw (d) to ++(0,-1.5);
    \draw (d) to[out=180,in=-90] ++(-1,1) to ++(0,.5);
    \draw (d) to[out=0,in=-90] ++(1,1) to ++(0,.5);
  \end{pic}
\end{equation}

\subsection{Correctness and security of DHKE}

We have now discussed all of the building blocks of DHKE, and can now depict the protocol pictorially as follows.
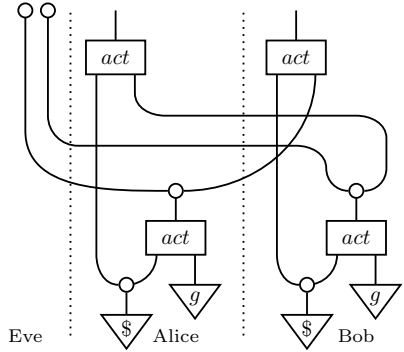
\[
  \begin{pic}[scale=.4]
    \setlength\minimummorphismwidth{7mm}
    \node[morphism,scale=.75] (a1) at (0,0) {\normalsize$act$};
    \draw ([xshift=-3.25pt]a1.south west) to[in=0,out=-90] ++(-1,-1) node[dot,fill=\copycolor] (d1) {};
    \draw (d1) to +(0,-1) node[state,scale=0.75] {\normalsize$\rand$};
    \node[morphism,scale=.75] (b1) at (-2,6) {\normalsize$act$}; 
    \draw (b1.north) to ++(0,1);
    \draw (d1) to[out=180,in=-90] ++(-1,1) to ([xshift=-3.25pt]b1.south west);
    \draw ([xshift=3.25pt]a1.south east) to ++(0,-1) node[state,scale=0.75] {\normalsize$g$};
    \draw (a1.north) to ++(0,1) node[dot,fill=\copycolor] (e1) {};
    \node[morphism,scale=.75] (a2) at (6,0) {\normalsize$act$}; 
    \draw ([xshift=-3.25pt]a2.south west) to[in=0,out=-90] ++(-1,-1) node[dot,fill=\copycolor] (d2) {};
    \draw (d2) to +(0,-1) node[state,scale=0.75] {\normalsize$\rand$};
    \node[morphism,scale=.75] (b2) at (4,6) {\normalsize$act$};  
    \draw (b2.north) to ++(0,1);
    \draw (d2) to[out=180,in=-90] ++(-1,1) to ([xshift=-3.25pt]b2.south west);
    \draw ([xshift=3.25pt]a2.south east) to ++(0,-1) node[state,scale=0.75] {\normalsize$g$};
    \draw (a2.north) to ++(0,1) node[dot,fill=\copycolor] (e2) {};
    \draw (e1) to[out=0,in=-90] ([xshift=3.25pt]b2.south east);
    \draw (e1) to[out=180,in=-90] ++(-5,2) to ++(0,4) node[dot,fill=\copycolor] {}; 
     \draw (e2) to[out=180,in=-90] ++(-1,.75) to[out=90,in=0] ++(-1,.75) to ++(-7.25,0) to[out=180,in=-90] ++(-1,1) to ++(0,3.5) node[dot,fill=\copycolor] {}; 
    \draw (e2) to[out=0,in=-90] ++(1,.75) to ++(0,1) to[out=90,in=0] ++(-1,.75) to ++(-6.36,0) to[in=-90,out=180] ++(-1,.75) to ([xshift=3.25pt]b1.south east);
    \draw[dotted] (-3.5,-3.25) to (-3.5,7.5);
    \draw[dotted] (2.25,-3.25) to (2.25,7.5);
    \node at (-5,-3.25) {Eve};
    \node at (0,-3.25) {Alice};
    \node at (6,-3.25) {Bob};
  \end{pic}
  \]
This picture corresponds to our earlier description of the protocol, where Alice and Bob first uniformly sample $a$ and $b$ from $\Z_n$ where $n=\left\vert G\right\vert$. Alice and Bob then broadcast $g^a$ and $g^b$ over a public channel, after which Alice computes $k_a:=(g^b)^a$ and Bob computes $k_b:=(g^a)^b$, with Eve supposed to delete the public messages.
We can now prove that the protocol is correct using the properties we have laid out.
 
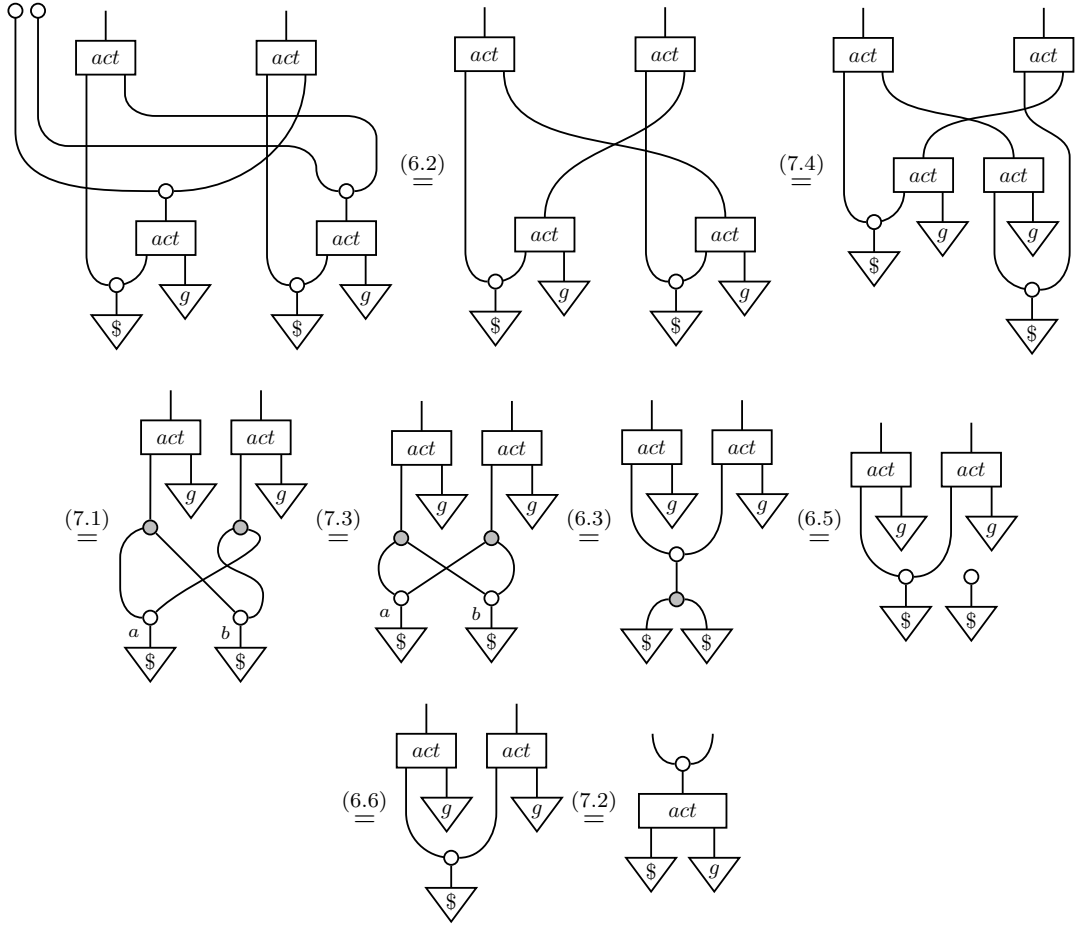
\[
  \begin{pic}[scale=.4]
    \setlength\minimummorphismwidth{7mm}
    \node[morphism,scale=.75] (a1) at (0,0) {\normalsize$act$};
    \draw ([xshift=-3.25pt]a1.south west) to[in=0,out=-90] ++(-1,-1) node[dot,fill=\copycolor] (d1) {};
    \draw (d1) to +(0,-1) node[state,scale=0.75] {\normalsize$\rand$};
    \node[morphism,scale=.75] (b1) at (-2,6) {\normalsize$act$};
    \draw (b1.north) to ++(0,1);
    \draw (d1) to[out=180,in=-90] ++(-1,1) to ([xshift=-3.25pt]b1.south west);
    \draw ([xshift=3.25pt]a1.south east) to ++(0,-1) node[state,scale=0.75] {\normalsize$g$};
    \draw (a1.north) to ++(0,1) node[dot,fill=\copycolor] (e1) {};
    \node[morphism,scale=.75] (a2) at (6,0) {\normalsize$act$};
    \draw ([xshift=-3.25pt]a2.south west) to[in=0,out=-90] ++(-1,-1) node[dot,fill=\copycolor] (d2) {};
    \draw (d2) to +(0,-1) node[state,scale=0.75] {\normalsize$\rand$};
    \node[morphism,scale=.75] (b2) at (4,6) {\normalsize$act$};  
    \draw (b2.north) to ++(0,1);
    \draw (d2) to[out=180,in=-90] ++(-1,1) to ([xshift=-3.25pt]b2.south west);
    \draw ([xshift=3.25pt]a2.south east) to ++(0,-1) node[state,scale=0.75] {\normalsize$g$};
    \draw (a2.north) to ++(0,1) node[dot,fill=\copycolor] (e2) {};
    \draw (e1) to[out=0,in=-90] ([xshift=3.25pt]b2.south east);
    \draw (e1) to[out=180,in=-90] ++(-5,2) to ++(0,4) node[dot,fill=\copycolor] {}; 
     \draw (e2) to[out=180,in=-90] ++(-1,.75) to[out=90,in=0] ++(-1,.75) to ++(-7.25,0) to[out=180,in=-90] ++(-1,1) to ++(0,3.5) node[dot,fill=\copycolor] {};
    \draw (e2) to[out=0,in=-90] ++(1,.75) to ++(0,1) to[out=90,in=0] ++(-1,.75) to ++(-6.36,0) to[in=-90,out=180] ++(-1,.75) to ([xshift=3.25pt]b1.south east);
  \end{pic}
   \overset{\eqref{eq:comonoid}}{=}
     \begin{pic}[scale=.4]
    \setlength\minimummorphismwidth{7mm}
    \node[morphism,scale=.75] (a1) at (0,0) {\normalsize$act$}; 
    \node[morphism,scale=.75] (a2) at (6,0) {\normalsize$act$}; 
    \node[morphism,scale=.75] (b1) at (-2,6) {\normalsize$act$}; 
    \node[morphism,scale=.75] (b2) at (4,6) {\normalsize$act$}; 
    \draw ([xshift=-3.25pt]a1.south west) to[in=0,out=-90] ++(-1,-1) node[dot,fill=\copycolor] (d1) {};
    \draw (d1) to +(0,-1) node[state,scale=0.75] {\normalsize$\rand$};
    \draw (b1.north) to ++(0,1);
    \draw (d1) to[out=180,in=-90] ++(-1,1) to ([xshift=-3.25pt]b1.south west);
    \draw ([xshift=3.25pt]a1.south east) to ++(0,-1) node[state,scale=0.75] {\normalsize$g$};
    \draw (a1.north) to[out=90,in=-90] ([xshift=3.25pt]b2.south east);
    \draw ([xshift=-3.25pt]a2.south west) to[in=0,out=-90] ++(-1,-1) node[dot,fill=\copycolor] (d2) {};
    \draw (d2) to +(0,-1) node[state,scale=0.75] {\normalsize$\rand$};
    \draw (b2.north) to ++(0,1);
    \draw (d2) to[out=180,in=-90] ++(-1,1) to ([xshift=-3.25pt]b2.south west);
    \draw ([xshift=3.25pt]a2.south east) to ++(0,-1) node[state,scale=0.75] {\normalsize$g$};
    \draw (a2.north) to[out=90,in=-90] ([xshift=3.25pt]b1.south east);
  \end{pic}
   \overset{\eqref{eq:cocommmutativity}}{=}
     \begin{pic}[scale=.4]
    \setlength\minimummorphismwidth{7mm}
    \node[morphism,scale=.75] (a1) at (0,0) {\normalsize$act$}; 
    \node[morphism,scale=.75] (a2) at (3,0) {\normalsize$act$}; 
    \node[morphism,scale=.75] (b1) at (-2,4) {\normalsize$act$};
    \node[morphism,scale=.75] (b2) at (4,4) {\normalsize$act$}; 
    \draw ([xshift=-3.25pt]a1.south west) to[in=0,out=-90] ++(-1,-1) node[dot,fill=\copycolor] (d1) {};
    \draw (d1) to +(0,-1) node[state,scale=0.75] {\normalsize$\rand$};
    \draw (b1.north) to ++(0,1);
    \draw (d1) to[out=180,in=-90] ++(-1,1) to ([xshift=-3.25pt]b1.south west);
    \draw ([xshift=3.25pt]a1.south east) to ++(0,-1) node[state,scale=0.75] {\normalsize$g$};
    \draw (a1.north) to[out=90,in=-90] ([xshift=3.25pt]b2.south east);
    \draw ([xshift=-3.25pt]a2.south west) to ++(0,-2) to[in=180,out=-90] ++(1.25,-1.25) node[dot,fill=\copycolor] (d2) {};
    \draw (d2) to ++(0,-1) node[state,scale=0.75] {\normalsize$\rand$};
    \draw (b2.north) to ++(0,1) node[right] {};
    \draw (d2)  to[out=0,in=-90] ++(1.25,1.25) to ++(0,3.155) to[out=90,in=-90] ++(-1.5,1.75) to ([xshift=-3.25pt]b2.south west);
    \draw ([xshift=3.25pt]a2.south east) to ++(0,-1) node[state,scale=0.75] {\normalsize$g$};
    \draw (a2.north) to[out=90,in=-90] ([xshift=3.25pt]b1.south east);
  \end{pic}
  \]
  \[
       \overset{\eqref{eq:module}}{=}
     \begin{pic}[scale=.4]
    \setlength\minimummorphismwidth{7mm}
    \node[morphism,scale=.75] (a) at (0.68,6) {\normalsize$act$};
    \node[morphism,scale=.75] (c) at (3.68,6) {\normalsize$act$};
    \draw (a.north) to ++(0,1) node[left] {};
    \draw ([xshift=4.25pt]a.south east) to ++(0,-1) node[state,scale=0.75]  {\normalsize$g$};
    \node[dot,fill=\copycolor] (e) at (3,0) {};
    \draw (e) to ++(0,-.5) node[left] {$b$} to ++(0,-.5) node[state,scale=0.75] {\normalsize$\rand$};
    \draw (e) to[out=0,in=-90] ++(0.75,.75) to[out=90,in=-90] ++(-1.5,1.5) to[out=90,in=180] ++(.75,.75) node[dot,fill=\multcolor] (f) {};
    \node[dot,fill=\copycolor] (d) at (0,0) {}; 
    \draw (d) to ++(0,-.5) node[left] {$a$} to ++(0,-.5) node[state,scale=0.75] {\normalsize$\rand$};
    \draw (d) to[in=-90,out=180] ++(-1,1) to ++(0,1) to[out=90,in=180] ++(1,1) node[dot,fill=\multcolor] (g) {};
    \draw (d) to[in=-15] (f);
    \draw (e) to (g);
     \draw ([xshift=-4.25pt]a.south west) to (g);
    \draw ([xshift=-4.25pt]c.south west) to (f);
    \draw ([xshift=4.25pt]c.south east) to ++(0,-1) node[state,scale=0.75]  {\normalsize$g$};
    \draw (c.north) to ++(0,1); 
  \end{pic}
   \overset{\eqref{eq:commutativity}}{=}
     \begin{pic}[scale=.4]
    \setlength\minimummorphismwidth{7mm}
    \node[morphism,scale=.75] (a) at (0.68,5) {\normalsize$act$};
    \node[morphism,scale=.75] (c) at (3.68,5) {\normalsize$act$};
    \draw (a.north) to ++(0,1);
    \draw ([xshift=4.25pt]a.south east) to ++(0,-1) node[state,scale=0.75]  {\normalsize$g$};
    \node[dot,fill=\copycolor] (e) at (3,0) {};
    \draw (e) to ++(0,-.5) node[left] {$b$} to ++(0,-.5) node[state,scale=0.75] {\normalsize$\rand$};
    \draw (e) to[out=30,in=-90] ++(0.75,1) to[out=90,in=-30] ++(-0.75,1) node[dot,fill=\multcolor] (f) {};
    \node[dot,fill=\copycolor] (d) at (0,0) {}; 
    \draw (d) to ++(0,-.5) node[left] {$a$} to ++(0,-.5) node[state,scale=0.75] {\normalsize$\rand$};
    \draw (d) to[in=-90,out=150] ++(-0.75,1) to[out=90,in=210] ++(0.75,1) node[dot,fill=\multcolor] (g) {};
    \draw (d) to (f);
    \draw (e) to (g);
     \draw ([xshift=-4.25pt]a.south west) to (g);
    \draw ([xshift=-4.25pt]c.south west) to (f);
    \draw ([xshift=4.25pt]c.south east) to ++(0,-1) node[state,scale=0.75]  {\normalsize$g$};
    \draw (c.north) to ++(0,1); 
  \end{pic}
   \overset{\eqref{eq:bialg}}{=}
        \begin{pic}[scale=.4]
    \setlength\minimummorphismwidth{7mm}
    \node[morphism,scale=.75] (a) at (0,6) {\normalsize$act$};
    \node[morphism,scale=.75] (c) at (3,6) {\normalsize$act$};
    \draw (a.north) to ++(0,1);
    \draw ([xshift=4.25pt]a.south east) to ++(0,-1) node[state,scale=0.75]  {\normalsize$g$};
    \draw (c.north) to ++(0,1); 
    \draw ([xshift=4.25pt]c.south east) to ++(0,-1) node[state,scale=0.75]  {\normalsize$g$};
    \draw ([xshift=-4.25pt]a.south west) to ++(0,-1.5) to[out=-90,in=180] ++(1.5,-1.5) node[dot,fill=\copycolor] (e) {};
    \draw ([xshift=-4.25pt]c.south west) to ++(0,-1.5) to[out=-90,in=0] (e);
    \draw (e) to ++(0,-1.5) node[dot,fill=\multcolor] (f) {};
    \draw (f) to[out=0,in=90] ++(1,-1) node[state,scale=0.75] {\normalsize$\rand$};
    \draw (f) to[out=180,in=90] ++(-1,-1) node[state,scale=0.75] {\normalsize$\rand$};
  \end{pic}
      \overset{\eqref{eq:uniformly_random}}{=}
    \begin{pic}[scale=.4]
    \setlength\minimummorphismwidth{7mm}
    \node[morphism,scale=.75] (a) at (0,3.6) {\normalsize$act$};
    \node[morphism,scale=.75] (c) at (3,3.6) {\normalsize$act$};
    \draw (a.north) to ++(0,1);
    \draw ([xshift=4.25pt]a.south east) to ++(0,-1) node[state,scale=0.75]  {\normalsize$g$};
    \draw (c.north) to ++(0,1); 
    \draw ([xshift=4.25pt]c.south east) to ++(0,-1) node[state,scale=0.75]  {\normalsize$g$};
    \draw ([xshift=-4.25pt]a.south west) to ++(0,-1.5) to[out=-90,in=180] ++(1.5,-1.5) node[dot,fill=\copycolor] (e) {};
    \draw ([xshift=-4.25pt]c.south west) to ++(0,-1.5) to[out=-90,in=0] (e);
    \draw (e) to ++(0,-1) node[state,scale=0.75] {\normalsize$\rand$};
    \draw[color=white] (c.south) to ++(0,-3) node[dot,fill=\copycolor] (f) {};
    \draw (f) to ++(0,-1) node[state,scale=0.75] {\normalsize$\rand$};
  \end{pic}
  \]
  \[
   \overset{\eqref{eq:deletion}}{=}
    \begin{pic}[scale=.4]
    \setlength\minimummorphismwidth{7mm}
    \node[morphism,scale=.75] (a) at (0,6) {\normalsize$act$};
    \node[morphism,scale=.75] (c) at (3,6) {\normalsize$act$};
    \draw (a.north) to ++(0,1);
    \draw ([xshift=4.25pt]a.south east) to ++(0,-1) node[state,scale=0.75]  {\normalsize$g$};
    \draw (c.north) to ++(0,1); 
    \draw ([xshift=4.25pt]c.south east) to ++(0,-1) node[state,scale=0.75]  {\normalsize$g$};
    \draw ([xshift=-4.25pt]a.south west) to ++(0,-1.5) to[out=-90,in=180] ++(1.5,-1.5) node[dot,fill=\copycolor] (e) {};
    \draw ([xshift=-4.25pt]c.south west) to ++(0,-1.5) to[out=-90,in=0] (e);
    \draw (e) to ++(0,-1) node[state,scale=0.75] {\normalsize$\rand$};
  \end{pic}
   \overset{\eqref{eq:determinism}}{=}
     \begin{pic}[scale=.4]
    \setlength\minimummorphismwidth{12mm}
    \node[morphism,scale=.75] (c) at (0,0) {\normalsize$act$}; 
    \draw (c.north) to ++(0,1) node[dot,fill=\copycolor] (d){};
    \draw ([xshift=-4.25pt]c.south west)to ++(0,-.5) to ++(0,-.5)node[state,scale=0.75]  {\normalsize$\rand$};
    \draw ([xshift=4.25pt]c.south east) to ++(0,-1) node[state,scale=0.75]  {\normalsize$g$};
    \draw (d) to[in=-90,out=180] ++(-1,1);
    \draw (d) to[in=-90,out=0] ++(1,1);
    \end{pic}
  \]

This shows that when everyone follows the protocol, the end result is Alice and Bob sharing an element of the form $g^a$ for $a$ chosen uniformly from $\{1,\dots,\left\vert G\right\vert\}$. As this gives the uniform distribution on $G$, the protocol is indeed correct.

We will now discuss the security of the protocol. The usual formulation of the DH-assumption is given by saying that
when $a,b,c$ are sampled uniformly, we have $(g^a,g^b,g^{ab})\approx (g^a,g^b,g^c)$. For us, it's easier to work with an equivalent formulation of assuming that  $(g^a,g^b,g^{ab},g^{ab})\approx (g^a,g^b,g^c,g^c)$, as either joint probability distribution can be derived from the other efficiently, just by copying or discarding the last element as appropriate.

Now, when one turns a claim like  $(g^a,g^b,g^{ab},g^{ab})\approx (g^a,g^b,g^c,g^c)$ into an axiom relating two pictures, one has some choices to make: after all, there are many possible diagrams that depict the probability distributions appearing in this claim. One could create a long-winded pictorial proof akin to the correctness proof by choosing appropriate starting pictures. However, we can make things easier for us, and formalize the DH-assumption by stating that
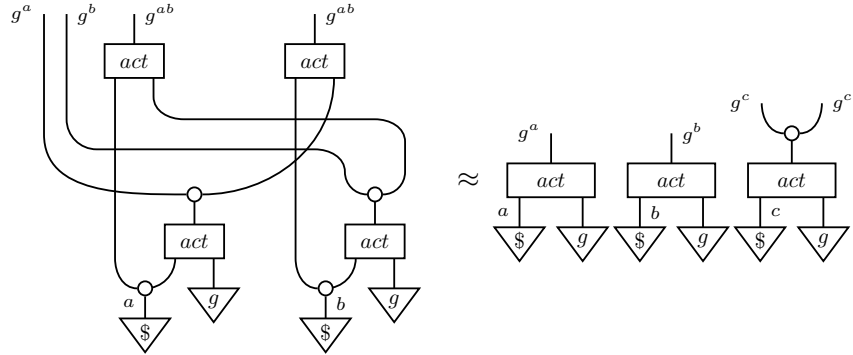
\[
  \begin{pic}[scale=.4]
    \setlength\minimummorphismwidth{7mm}
    \node[morphism,scale=.75] (a1) at (0,0) {\normalsize$act$};
    \node[morphism,scale=.75] (b1) at (-2,6) {\normalsize$act$};
    \node[morphism,scale=.75] (a2) at (6,0) {\normalsize$act$}; 
    \node[morphism,scale=.75] (b2) at (4,6) {\normalsize$act$}; 
    \draw ([xshift=-3.25pt]a1.south west) to[in=0,out=-90] ++(-1,-1) node[dot,fill=\copycolor]  {};
    \draw (d1) to ++(0,-.5) node[left] {$a$} to ++(0,-.5) node[state,scale=0.75] {\normalsize$\rand$};
    \draw (b1.north) to ++(0,1) node[right] {$g^{ab}$};
    \draw (d1) to[out=180,in=-90] ++(-1,1) to ([xshift=-3.25pt]b1.south west);
    \draw ([xshift=3.25pt]a1.south east) to ++(0,-1) node[state,scale=0.75] {\normalsize$g$};
    \draw (a1.north) to ++(0,1) node[dot,fill=\copycolor] (e1) {};
    \draw ([xshift=-3.25pt]a2.south west) to[in=0,out=-90] ++(-1,-1) node[dot,fill=\copycolor] (d2) {};
    \draw (d2) to ++(0,-.5) node[right] {$b$} to ++(0,-.5)  node[state,scale=0.75] {\normalsize$\rand$};
    \draw (b2.north) to ++(0,1) node[right] {$g^{ab}$};
    \draw (d2) to[out=180,in=-90] ++(-1,1) to ([xshift=-3.25pt]b2.south west);
    \draw ([xshift=3.25pt]a2.south east) to ++(0,-1) node[state,scale=0.75] {\normalsize$g$};
    \draw (a2.north) to ++(0,1) node[dot,fill=\copycolor] (e2) {};
    \draw (e1) to[out=0,in=-90] ([xshift=3.25pt]b2.south east);
    \draw (e1) to[out=180,in=-90] ++(-5,2) to ++(0,4) node[left] {$g^a$}; 
     \draw (e2) to[out=180,in=-90] ++(-1,.75) to[out=90,in=0] ++(-1,.75) to ++(-7.25,0) to[out=180,in=-90] ++(-1,1) to ++(0,3.5)  node[right] {$g^b$}; 
    \draw (e2) to[out=0,in=-90] ++(1,.75) to ++(0,1) to[out=90,in=0] ++(-1,.75) to ++(-6.36,0) to[in=-90,out=180] ++(-1,.75) to ([xshift=3.25pt]b1.south east);
  \end{pic}
  \quad \approx
     \begin{pic}[scale=.4]
    \setlength\minimummorphismwidth{12mm}
    \node[morphism,scale=.75] (a) at (0,0) {\normalsize$act$}; 
    \draw (a.north) to ++(0,1) node[left] {$g^a$};
    \draw ([xshift=-4.25pt]a.south west) to ++(0,-.5) node[left] {$a$} to ++(0,-.5) node[state,scale=0.75]  {\normalsize$\rand$};
    \draw ([xshift=4.25pt]a.south east) to ++(0,-1) node[state,scale=0.75]  {\normalsize$g$};
    \node[morphism,scale=.75] (b) at (4,0) {\normalsize$act$};  
    \draw (b.north) to ++(0,1) node[right] {$g^b$};
    \draw ([xshift=-4.25pt]b.south west) to ++(0,-.5) node[right] {$b$} to ++(0,-.5) node[state,scale=0.75]  {\normalsize$\rand$};
    \draw ([xshift=4.25pt]b.south east) to ++(0,-1) node[state,scale=0.75]  {\normalsize$g$};
    \node[morphism,scale=.75] (c) at (8,0) {\normalsize$act$};  
    \draw (c.north) to ++(0,1) node[dot,fill=\copycolor] (d){};
    \draw ([xshift=-4.25pt]c.south west)to ++(0,-.5) node[right] {$c$} to ++(0,-.5)node[state,scale=0.75]  {\normalsize$\rand$};
    \draw ([xshift=4.25pt]c.south east) to ++(0,-1) node[state,scale=0.75]  {\normalsize$g$};
    \draw (d) to[in=-90,out=180] ++(-1,1) node[left]  {$g^c$};
    \draw (d) to[in=-90,out=0] ++(1,1) node[right] {$g^c$};
  \end{pic}
  \]
as the left-hand side indeed depicts the distribution $(g^a,g^b,g^{ab},g^{ab})$ and the right hand side the distribution $(g^a,g^b,g^c,g^c)$. From this assumption, the security of DHKE is immediate: the left-hand side depicts the initial attack by Eve, and up to $\approx$, this is the same result as Eve sampling $g^a$ and $g^b$ on her own while Alice and Bob share $g^c$, which is uniformly distributed. Consequently, DHKE is a secure protocol that takes two broadcast channels (one from Alice and one from Bob) and constructs a shared secret key out of them. In pictures, one might then say that DHKE is a transformation

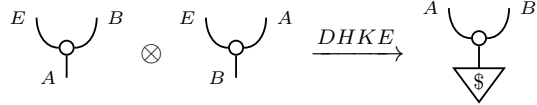
\[ \begin{pic}[scale=.4,yscale=-1]
    \node[dot,fill=\copycolor] (d) {};
    \draw (d) to +(0,1) node [left] {$A$};
    \draw (d) to[out=0,in=90] +(1,-1) node[right] {$B$};
    \draw (d) to[out=180,in=90] +(-1,-1) node[left] {$E$};
  \end{pic}
  \otimes
    \begin{pic}[scale=.4,yscale=-1]
    \node[dot,fill=\copycolor] (d) {};
    \draw (d) to +(0,1) node [left] {$B$};
    \draw (d) to[out=0,in=90] +(1,-1) node[right] {$A$};
    \draw (d) to[out=180,in=90] +(-1,-1) node[left] {$E$};
  \end{pic}
  \xrightarrow{DHKE}
    \begin{pic}[scale=.4]
    \node[dot,fill=\copycolor] (d) {};
    \draw (d) to +(0,-1) node[state,scale=0.75] {\normalsize$\rand$};
    \draw (d) to[out=180,in=-90] +(-1,1) node[left] {$A$};
    \draw (d) to[out=0,in=-90] +(1,1) node[right] {$B$};
  \end{pic}
  \]
Alternatively, one could have formulated the DH-assumption as stating that
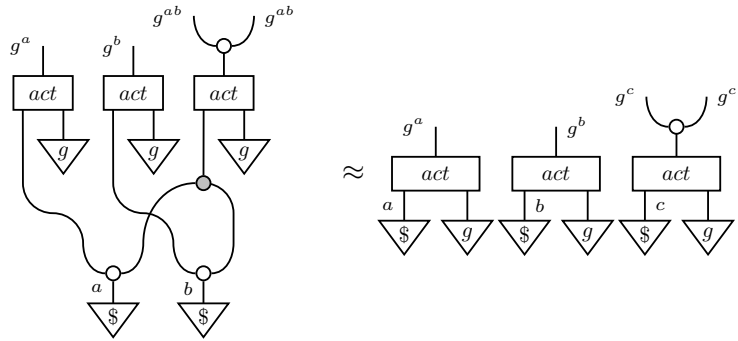
\[
     \begin{pic}[scale=.4]
    \setlength\minimummorphismwidth{7mm}
    \node[morphism,scale=.75] (a) at (-2.33,6) {\normalsize$act$};
    \draw (a.north) to ++(0,1) node[left] {$g^a$};
    \draw ([xshift=4.25pt]a.south east) to ++(0,-1) node[state,scale=0.75]  {\normalsize$g$};
    \node[morphism,scale=.75] (b) at (0.67,6) {\normalsize$act$}; 
    \draw (b.north) to ++(0,1) node[left] {$g^b$};
    \draw ([xshift=4.25pt]b.south east) to ++(0,-1) node[state,scale=0.75]  {\normalsize$g$};
    \node[dot,fill=\copycolor] (e) at (3,0) {}; 
    \draw (e) to ++(0,-.5) node[left] {$b$} to ++(0,-.5) node[state,scale=0.75] {\normalsize$\rand$};
    \draw (e) to[out=180,in=-90] ++(-1,1) to[out=90,in=0] ++(-1,1) to[out=180,in=-90] ++(-1,1) to ([xshift=-4.25pt]b.south west);
    \draw (e) to[out=0,in=-90] ++(1,1) to ++(0,1) to[out=90,in=0] ++(-1,1) node[dot,fill=\multcolor] (f) {};
    \node[dot,fill=\copycolor] (d) at (0,0) {}; 
    \draw (d)to ++(0,-.5) node[left] {$a$} to ++(0,-.5) node[state,scale=0.75] {\normalsize$\rand$};
    \draw (d) to[out=180,in=-90] ++(-1,1) to[out=90,in=0] ++(-1,1) to[out=180,in=-90] ++(-1,1) to ([xshift=-4.25pt]a.south west);
    \draw (d) to[out=0,in=-90] ++(1,1) to[in=180,out=90] (f);
    \node[morphism,scale=.75] (c) at (3.68,6) {\normalsize$act$}; 
    \draw ([xshift=-4.25pt]c.south west) to (f);
    \draw ([xshift=4.25pt]c.south east) to ++(0,-1) node[state,scale=0.75]  {\normalsize$g$};
    \draw (c.north) to ++(0,1) node[dot,fill=\copycolor] (g) {};
    \draw (g) to[in=-90,out=180] ++(-1,1) node[left]  {$g^{ab}$};
    \draw (g) to[in=-90,out=0] ++(1,1) node[right] {$g^{ab}$};
  \end{pic}
  \quad \approx
     \begin{pic}[scale=.4]
    \setlength\minimummorphismwidth{12mm}
    \node[morphism,scale=.75] (a) at (0,0) {\normalsize$act$}; 
    \draw (a.north) to ++(0,1) node[left] {$g^a$};
    \draw ([xshift=-4.25pt]a.south west) to ++(0,-.5) node[left] {$a$} to ++(0,-.5) node[state,scale=0.75]  {\normalsize$\rand$};
    \draw ([xshift=4.25pt]a.south east) to ++(0,-1) node[state,scale=0.75]  {\normalsize$g$};
    \node[morphism,scale=.75] (b) at (4,0) {\normalsize$act$};  
    \draw (b.north) to ++(0,1) node[right] {$g^b$};
    \draw ([xshift=-4.25pt]b.south west) to ++(0,-.5) node[right] {$b$} to ++(0,-.5) node[state,scale=0.75]  {\normalsize$\rand$};
    \draw ([xshift=4.25pt]b.south east) to ++(0,-1) node[state,scale=0.75]  {\normalsize$g$};
    \node[morphism,scale=.75] (c) at (8,0) {\normalsize$act$}; 
    \draw (c.north) to ++(0,1) node[dot,fill=\copycolor] (d){};
    \draw ([xshift=-4.25pt]c.south west)to ++(0,-.5) node[right] {$c$} to ++(0,-.5)node[state,scale=0.75]  {\normalsize$\rand$};
    \draw ([xshift=4.25pt]c.south east) to ++(0,-1) node[state,scale=0.75]  {\normalsize$g$};
    \draw (d) to[in=-90,out=180] ++(-1,1) node[left]  {$g^c$};
    \draw (d) to[in=-90,out=0] ++(1,1) node[right] {$g^c$};
  \end{pic}
  \]
For this formulation, security is no longer immediate: to conclude that the protocol is secure, one must first transform the LHS of our prior security definition into the LHS of this equation. We leave developing such a pictorial proof for the interested reader.

It is well-known that DHKE is vulnerable to an attacker-in-the-middle, where Eve pretends to be Bob when communicating with Alice and vice versa, establishing a shared secret key with both. How does this look like in our approach, especially since DHKE as analyzed above was deemed secure? The reason for this is that in our simplistic setting, the channels used by Alice and Bob are authenticated, in that both parties know that any message sent along the channel goes to the other party as-is. Moreover, these starting resources do not allow Eve to send messages to Alice and Bob. Consequently, in our model Eve simply cannot perform such an attack. However, we could instead have our starting resources to be channels between the honest parties and Eve. One could then consider the DHKE protocol in this setting, where Eve is supposed to just wire the messages onward passively. In this case, the protocol is no longer secure. For instance, if Eve behaves as an attacker-in-the-middle, both Alice and Bob end up sharing a key with Eve and not with each other, whereas there is no attack Eve can perform on the ideal key resource with an indistinguishable end result.

\subsection{Composing with other protocols}

Now, DHKE gives a transformation
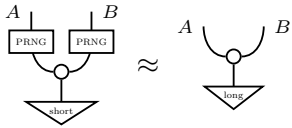
\[ \begin{pic}[scale=.4,yscale=-1]
    \node[dot,fill=\copycolor] (d) {};
    \draw (d) to +(0,1) node [left] {$A$};
    \draw (d) to[out=0,in=90] +(1,-1) node[right] {$B$};
    \draw (d) to[out=180,in=90] +(-1,-1) node[left] {$E$};
  \end{pic}
  \otimes
    \begin{pic}[scale=.4,yscale=-1]
    \node[dot,fill=\copycolor] (d) {};
    \draw (d) to +(0,1) node [left] {$B$};
    \draw (d) to[out=0,in=90] +(1,-1) node[right] {$A$};
    \draw (d) to[out=180,in=90] +(-1,-1) node[left] {$E$};
  \end{pic}
  \xrightarrow{DHKE}
    \begin{pic}[scale=.4]
    \node[dot,fill=\copycolor] (d) {};
    \draw (d) to +(0,-1) node[state,scale=0.75] {\normalsize$\rand$};
    \draw (d) to[out=180,in=-90] +(-1,1) node[left] {$A$};
    \draw (d) to[out=0,in=-90] +(1,1) node[right] {$B$};
  \end{pic}
  \]
and the OTP from the previous section gives a map
\[ \begin{pic}[scale=.4,yscale=-1]
    \node[dot,fill=\copycolor] (d) {};
    \draw (d) to +(0,1) node [left] {$A$};
    \draw (d) to[out=0,in=90] +(1,-1) node[right] {$B$};
    \draw (d) to[out=180,in=90] +(-1,-1) node[left] {$E$};
  \end{pic}
  \otimes
    \begin{pic}[scale=.4]
    \node[dot,fill=\copycolor] (d) {};
    \draw (d) to +(0,-1) node[state,scale=0.75] {\normalsize$\rand$};
    \draw (d) to[out=180,in=-90] +(-1,1) node[left] {$A$};
    \draw (d) to[out=0,in=-90] +(1,1) node[right] {$B$};
  \end{pic}
  \xrightarrow{OTP}
    \begin{pic}[scale=.4] \draw (0,0) node[left] {$A$} to (0,3) node[right] {$B$}; \end{pic}
  \]
so one would like to form the composite
\[\begin{pic}[scale=.4,yscale=-1]
    \node[dot,fill=\copycolor] (d) {};
    \draw (d) to +(0,1) node [left] {$A$};
    \draw (d) to[out=0,in=90] +(1,-1) node[right] {$B$};
    \draw (d) to[out=180,in=90] +(-1,-1) node[left] {$E$};
  \end{pic}
  \otimes
\begin{pic}[scale=.4,yscale=-1]
    \node[dot,fill=\copycolor] (d) {};
    \draw (d) to +(0,1) node [left] {$A$};
    \draw (d) to[out=0,in=90] +(1,-1) node[right] {$B$};
    \draw (d) to[out=180,in=90] +(-1,-1) node[left] {$E$};
  \end{pic}
  \otimes
    \begin{pic}[scale=.4,yscale=-1]
    \node[dot,fill=\copycolor] (d) {};
    \draw (d) to +(0,1) node [left] {$B$};
    \draw (d) to[out=0,in=90] +(1,-1) node[right] {$A$};
    \draw (d) to[out=180,in=90] +(-1,-1) node[left] {$E$};
  \end{pic}
  \xrightarrow{\id\otimes DHKE}
  \begin{pic}[scale=.4,yscale=-1]
    \node[dot,fill=\copycolor] (d) {};
    \draw (d) to +(0,1) node [left] {$A$};
    \draw (d) to[out=0,in=90] +(1,-1) node[right] {$B$};
    \draw (d) to[out=180,in=90] +(-1,-1) node[left] {$E$};
  \end{pic}
  \otimes
    \begin{pic}[scale=.4]
    \node[dot,fill=\copycolor] (d) {};
    \draw (d) to +(0,-1) node[state,scale=0.75] {\normalsize$\rand$};
    \draw (d) to[out=180,in=-90] +(-1,1) node[left] {$A$};
    \draw (d) to[out=0,in=-90] +(1,1) node[right] {$B$};
    \end{pic}
    \xrightarrow{OTP}
    \begin{pic}[scale=.4] \draw (0,0) node[left] {$A$} to (0,3) node[right] {$B$}; \end{pic}
  \]
There is only one problem: in the previous section we worked with the resource theory induced by  $\ncomb(\cat{FinStoch}^3)\xrightarrow{\ncomb(\otimes)}\ncomb(\cat{FinStoch})\xrightarrow{\hom(I,-)}\Set$, whereas in this section we replaced $\cat{FinStoch}$ with $\Eff$ and moved to security up to $\approx$. However, it is easy to lift the OTP to $\Eff$ for any efficient sequence of groups $G$. Indeed, if $G$ is our sequence of groups, it satisfies all the equations we used in the security proof of the OTP and hence results in a secure protocol in the resource theory of this section. Concretely, this amounts to thinking of the one-time pad as follows: Alice and Bob agree on an efficient sequence of groups $G$, and then OTP is a protocol that transforms a broadcasting map (of type $G$) and shared uniformly random key over $G$ into a secure communication channel (of type $G$). As a result, the above composite makes sense and is computationally secure by our composition theorems.

We now illustrate a further use of the composition theorems, similar to the example in~\cite{Mau11}. A major drawback of OTP, despite its perfect security, is the fact that one needs a key that is as long as the message. Now, Alice and Bob might want to use DHKE to obtain a key in the first place, but perhaps they want to communicate a longer message than allowed by the key.  If they agree on a pseudo-random number generator (PRNG) with their key as the seed, they can map the short key to a longer key. If the PRNG is computationally secure, then the end-result is computationally indistinguishable from a long key, depicted by
\[  \begin{pic}[scale=.4,yscale=-1]
    \node[dot,fill=\copycolor] (d) {};
    \draw (d) to +(0,1) node[state,scale=0.5] {short};
    \draw (d) to[out=0,in=90] ++(1,-1) node[morphism,scale=0.5] {PRNG} to ++(0,-1) node[right] {$B$};
    \draw (d) to[out=180,in=90] ++(-1,-1)  node[morphism,scale=0.5] {PRNG} to ++(0,-1) node[left] {$A$};
  \end{pic}\approx
  \begin{pic}[scale=.4,yscale=-1]
    \node[dot,fill=\copycolor] (d) {};
    \draw (d) to +(0,1) node[state,scale=0.5] {long};
    \draw (d) to[out=0,in=90] ++(1,-1) node[right] {$B$};
    \draw (d) to[out=180,in=90] ++(-1,-1) node[left] {$A$};
  \end{pic}
  \]
One could then form a long composite of protocols, where one first uses DHKE to obtain a key, then PRNG to expand it, and finally, the OTP to use it.

There is one difference with our analysis of OTP to that of~\cite{Mau11} which we should comment on: in~\cite{Mau11}, the resulting secret channel is slightly weaker, as it gives some information to Eve---namely, the length of the message. Concretely, this makes perfect sense, since Eve learns the length of the message upon seeing the ciphertext. Similarly, one might model the ideal shared key resulting from DHKE as a slightly weaker resource, that gives a uniformly random shared key to Alice and Bob, but lets Eve know the size of the key space.
Why is this feature seemingly absent from our analysis, even when we work with $\Eff$, where the message is no longer over a fixed message space? In a sense, this feature is implicit in our setting, as Eve's behavior can depend on the security parameter and consequently on the message/key length. If one thinks of Alice and Bob choosing the concrete value of the security parameter they actually use, this dependency lets Eve's behavior depend on the message length. However, it might be desirable to work in a slightly different category where this information leakage is made more explicit. We leave developing such an analysis for future work.

\section{No-go results}\label{sec:no-go}

Composable security is a stronger constraint than stand-alone security, and indeed many cryptographic functionalities are known to be impossible to achieve ``in the plain model'', \ie without set-up assumptions. A case in point is bit commitment, which was shown to be impossible in the UC-framework in~\cite{CF01}. This result was later generalized in~\cite{PR08} to show that any two-party functionality that can be realized in the plain UC-framework is ``splittable''. While the authors of~\cite{PR08} remark that their result applies more generally than just to the UC-framework, this wasn't made precise until~\cite{MR11}\footnote{Except that in their framework the 2-party case seems to require security constraints also when both parties cheat.}. We present a categorical proof of this result in our framework, which promotes the pictures ``illustrating the proof'' in~\cite{PR08} into a full proof---the main difference is that in~\cite{PR08} the pictures explicitly keep track of an environment trying to distinguish between different functionalities, whereas we prove our result in the case of perfect security and then deduce the asymptotic claim.

We now assume that $\CC$, our ambient category of interactive computations is compact closed\footnote{We do not view this as overtly restrictive, as many theoretical models of concurrent interactive (probabilistic/quantum) computation are compact closed~\cite{winskel:game,clairambault:gamesforquantum,clairambaultetal:gamesforquantum2}.}. As we are in the 2-party setting, we take our free computations to be given by $\CC^2$, and we consider two attack models: one where Alice cheats and Bob is honest, and one where Bob cheats and Alice is honest. We think of $\tinycup$ as representing a two-way communication channel, but this interpretation is not needed for the formal result.
\begin{thm}\label{thm:bipartite} For Alice and Bob (one of whom might cheat),
if a bipartite functionality $r$ can be securely realized from a communication channel between them, \ie from $\tinycup$, then there is a $g$ such that
\begin{equation}\label{eq:splittable}\tag{$*$}
\begin{pic}
    \node[state] (x) at (0,0) {$r$};
    \draw (x.A) to ++(0,.4) node[left] {$A$};
    \draw (x.B) to ++(0,.4) node[right] {$B$};
  \end{pic}=\begin{pic}\setlength\morphismheight{3mm}
  \setlength\minimummorphismwidth{3mm}
    \node[state] (x) at (0,0) {$r$};
     \node[state] (y) at (1,0) {$r$};
    \node[morphism] (f) at (0.5,0.4) {$g$};
    \draw (x.A) to ++(0,.6);
    \draw (x.B) to ([xshift=-4.25pt] f.south west);
    \draw (y.A) to ([xshift=4.25pt] f.south east);
    \draw (y.B) to ++(0,.6);
  \end{pic}.
  \end{equation}
\end{thm}
\begin{proof}
A secure protocol transforming $\tinycup$ to $r$ consists of maps $f_A,f_B$ satisfying \[\begin{pic}
    \node[morphism,scale=.5,font=\normalsize] (f) at (0,0) {$f_A$};
    \node[morphism,scale=.5,font=\normalsize] (g) at (.5,0) {$f_B$};
    \draw (f.north) to +(0,.3);
    \draw (g.north) to +(0,.3);
    \draw (f.south)  to[out=-90,in=180] ++(0.25,-.25) to[out=0,in=-90] (g.south);
    \end{pic}=\begin{pic}
    \node[state] (x) at (0,0) {$r$};
    \draw (x.A) to ++(0,.3);
    \draw (x.B) to ++(0,.3);
  \end{pic}\]
and the security constraints against attacks by Alice or Bob. In particular, security against Alice's and Bob's initial attacks give us $s_A$ and $s_B$ such that
    \[\begin{pic}
    \node[morphism,scale=.5,font=\normalsize] (f) at (0,0) {$f_A$};
    \draw (f.north) to ++(0,.3);
    \draw (f.south)  to[out=-90,in=-180] ++(0.25,-.25) to[out=0,in=-90] ++(.25,.25) to ++(0,.6);
    \end{pic}=\begin{pic}
     \node[morphism,scale=.5,font=\normalsize] (g) at (.225,.35) {$s_B$};
    \node[state] (x) at (0,0) {$r$};
    \draw (x.A) to ++(0,.75);
    \draw (g.south) to  (x.B);
    \draw (g.north) to ++(0,.25);
  \end{pic}
 \quad \text{and} \quad
  \begin{pic}[xscale=-1]
    \node[morphism,scale=.5,font=\normalsize] (f) at (0,0) {$f_B$};
    \draw (f.north) to ++(0,.3);
    \draw (f.south)  to[out=-90,in=-180] ++(0.25,-.25) to[out=0,in=-90] ++(.25,.25) to ++(0,.6);
    \end{pic}=\begin{pic}[xscale=-1]
     \node[morphism,scale=.5,font=\normalsize] (g) at (.225,.35) {$s_A$};
    \node[state] (x) at (0,0) {$r$};
    \draw (x.B) to ++(0,.75);
    \draw (g.south) to (x.A);
    \draw (g.north) to ++(0,.25);
  \end{pic}   \quad \text{so that} \quad\begin{pic}
    \node[state] (x) at (0,0) {$r$};
    \draw (x.A) to ++(0,.3);
    \draw (x.B) to ++(0,.3);
  \end{pic}=\begin{pic}
    \node[morphism,scale=.5,font=\normalsize] (f) at (0,0) {$f_A$};
    \node[morphism,scale=.5,font=\normalsize] (g) at (.5,0) {$f_B$};
    \draw (f.north) to ++(0,.3);
    \draw (g.north) to ++(0,.3);
    \draw (f.south)  to[out=-90,in=180] ++(0.25,-.25) to[out=0,in=-90] (g.south);
    \end{pic}
  =\begin{pic}
    \node[morphism,scale=.5,font=\normalsize] (f) at (0,0) {$f_A$};
    \node[morphism,scale=.5,font=\normalsize] (g) at (1.5,0) {$f_B$};
    \draw (f.north) to ++(0,.3);
    \draw (g.north) to ++(0,.3);
    \draw (f.south)  to[out=-90,in=-180] ++(.25,-.25) to[out=0,in=-90] ++(.25,.25) to[out=90,in=180] ++(.25,.25) to [out=0,in=90] ++(.25,-.25) to[out=-90,in=180] ++(.25,-.25) to[out=0,in=-90] (g.south);
    \end{pic}
  =
\begin{pic}
     \node[morphism,scale=.5,font=\normalsize] (g) at (.23,.35) {$s_B$};
    \node[state] (x) at (0,0) {$r$};
    \node[morphism,scale=.5,font=\normalsize] (f) at (0.78,.35) {$s_A$};
    \node[state] (y) at (1,0) {$r$};
    \draw (x.A) to ++(0,.75);
    \draw (y.B) to ++(0,.75);
    \draw (g.south) to (x.B);
    \draw (g.north) to[out=90,in=-180] ++(0.275,.25) to [out=0,in=90] (f.north);
    \draw (f.south) to (y.A);
  \end{pic}\qedhere
  \]
\end{proof}

We now wish to use this result to rule out functionalities that do not satisfy this equation for any $g$. Notable examples of this include bit commitment, which lets Alice commit to a bit that will be revealed to Bob later, and oblivious transfer, which lets Alice choose two bits, of which Bob can choose to learn exactly one, with Alice not learning Bob's choice.

\begin{cor} Given a compact closed $\CC$ modeling computation in which wires model communication channels, (composable) bit commitment and oblivious transfer are impossible in that model without setup, even asymptotically in terms of distinguisher advantage.
\end{cor}
\begin{proof}
If $r$ represents bit commitment from Alice to Bob, it does not satisfy the equation required by Theorem~\ref{thm:bipartite} for any $g$, and the two sides of~\eqref{eq:splittable} can be distinguished efficiently with at least probability $1/2$. Indeed, take any $g$ and let us compare the two sides of~\eqref{eq:splittable}: if the distinguisher commits to a random bit $b$, then Bob gets a notification of this on the left hand-side, so that $g$ has to commit to a bit on the right side of~\eqref{eq:splittable} to avoid being distinguished from the left side. But this bit coincides with $b$ with probability at most $1/2$, so that the difference becomes apparent at the reveal stage. The case of OT is similar.
\end{proof}
We now discuss a similar result in the tripartite case, which rules out building a broadcasting channel from pairwise channels securely against any single party cheating. In~\cite{MMP+18} comparable pictures are used to illustrate the official, symbolically rather involved, proof, whereas in our framework the pictures are the proof. Another key difference is that~\cite{MMP+18} rules out broadcasting directly, whereas we show that any tripartite functionality realizable from pairwise channels satisfies some equations, and then use these equations to rule out broadcasting.

Formally, we are working with the resource theory given by $\CC^3\xrightarrow{\otimes}\CC\xrightarrow{\hom(I,-)}\Set$ where $\CC$ is an SMC, and reason about protocols that are secure against three kinds of attacks: one for each party behaving dishonestly while the rest obey the protocol. Note that we do not need to assume compact closure for this result, and the result goes through for any state on $A\otimes A$ shared between each pair of parties: we will denote such a state by $\tinycup$ by convention.

\begin{thm} If a tripartite functionality $r$ can be realized from each pair of parties sharing a state
 $\tinycup$, securely against any single party, then there are simulators $s_A,s_B,s_C$ such that
\[\begin{pic}
    \setlength\minimumstatewidth{15mm}
    \node[state] (r) at (0,0) {$r$};
    \node[morphism,scale=.5,font=\normalsize] (f) at (-.5,.35) {$s_A$};
    \draw ([xshift=-3.5pt]r.A) to (f.south);
    \draw ([xshift=3.5pt]r.B) to ++(0,.75);
    \draw ([xshift=-.5pt]f.north west) to ++(0,.25);
    \draw ([xshift=.5pt]f.north east) to ++(0,.25);
    \draw (r.center) to ++(0,.75);
  \end{pic}\quad=\quad
  \begin{pic}
    \setlength\minimumstatewidth{15mm}
    \node[state] (r) at (0,0) {$r$};
    \node[morphism,scale=.5,font=\normalsize] (h) at (0,.35) {$s_B$};
    \draw ([xshift=-3.5pt]r.A) to ++(0,.75);
    \draw ([xshift=3.5pt]r.B) to ++(0,.75);
    \draw ([xshift=-.5pt]h.north west) to ++(0,.25);
    \draw ([xshift=.5pt]h.north east) to ++(0,.25);
    \draw (r.center) to (h.south);
  \end{pic}\quad=\quad
  \begin{pic}
    \setlength\minimumstatewidth{15mm}
    \node[state] (r) at (0,0) {$r$};
    \node[morphism,scale=.5,font=\normalsize] (g) at (.5,.35) {$s_C$};
    \draw ([xshift=-3.5pt]r.A) to ++(0,.75);
    \draw ([xshift=3.5pt]r.B) to (g.south);
    \draw ([xshift=-.5pt]g.north west) to ++(0,.25);
    \draw ([xshift=.5pt]g.north east) to ++(0,.25);
    \draw (r.center) to ++(0,.75);
  \end{pic}\ .
  \]
\end{thm}

\begin{proof}
Any tripartite protocol building on top of each pair of parties sharing $\tinycup$ can be drawn as in the left side of
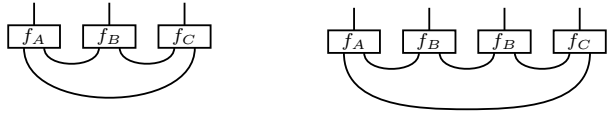
\[
\begin{pic}
    \setlength\morphismheight{3mm}
    \node[morphism] (f) at (0,0) {$f_A$};
    \node[morphism] (g) at (1,0) {$f_B$};
    \node[morphism] (h) at (2,0) {$f_C$};
    \draw (f.north) to ++(0,.3);
    \draw (g.north) to ++(0,.3);
    \draw (h.north) to ++(0,.3);
    \draw (f.south east)  to[out=-90,in=-90] (g.south west);
    \draw (g.south east)  to[out=-90,in=-90] (h.south west);
    \draw (f.south west)  to[out=-90,in=-90] (h.south east);
    \end{pic} \qquad\qquad
\begin{pic}
    \setlength\morphismheight{3mm}
    \coordinate (a) at (0.5,-.9);
    \node[morphism] (e) at (-1,0) {$f_A$};
    \node[morphism] (f) at (0,0) {$f_B$};
    \node[morphism] (g) at (1,0) {$f_B$};
    \node[morphism] (h) at (2,0) {$f_C$};
    \draw (e.north) to ++(0,.3);
    \draw (f.north) to ++(0,.3);
    \draw (g.north) to ++(0,.3);
    \draw (h.north) to ++(0,.3);
    \draw (e.south east)  to[out=-90,in=-90] (f.south west);
    \draw (f.south east)  to[out=-90,in=-90] (g.south west);
    \draw (g.south east)  to[out=-90,in=-90] (h.south west);
    \draw (e.south west)  to[out=-90,in=180] (a) to[out=0,in=-90] (h.south east);
    \end{pic}
   \]
Consider now the morphism in $\CC$ depicted on the right: it can be seen as the result of three different attacks on the protocol $(f_A,f_B,f_C)$ in $\CC^3$: one where Alice cheats and performs $f_A$ and $f_B$ (and the wire connecting them), one where Bob performs $f_B$ twice, and one where Charlie performs $f_B$ and $f_C$. The security of $(f_A,f_B,f_C)$ against each of these gives the required simulators.
\end{proof}

\begin{cor} Given an SMC $\CC$ modeling interactive computation, and a state $\tinycup$ on $A\otimes A$ modeling pairwise communication, it is impossible to build broadcasting channels securely (even asymptotically in terms of distinguisher advantage) from pairwise channels.
\end{cor}

\begin{proof} We show that a channel $r$ that enables Bob to broadcast an input bit to Alice and Charlie never satisfies the required equations for any $s_A,s_B,s_C$. Indeed, assume otherwise and let the environment
plug ``broadcast $0$'' and ``broadcast $1$'' to the two wires in the middle. The leftmost picture then says that Charlie receives $1$, the rightmost picture implies that Alice gets $0$ and the middle picture that Alice and Bob get the same output (if anything at all)---a contradiction. Indeed, one cannot satisfy all of these simultaneously with high probability, which rules out an asymptotic transformation.
\end{proof}

\section{On the choice of a model}\label{sec:models}

In this section we will discuss the choice of the category $\CC$ in which to work in. As we believe that one of the advantages of the current framework is how it allows one to separate the concern for composability from the choice of a model of computation, we do not think that choosing or building suitable models $\CC$ is central to the work at hand. Moreover, we do not believe that there is a single model suitable for all cryptographic purposes, as some assumptions that simplify the model mathematically might be justified for some tasks but not for others. Consequently, we will refer to the existing literature to explain where to find some models that are expressive enough for cryptography and result in an SMC, even when the original works are not expressed in categorical terms.

Note that in Section~\ref{sec:dhke}, we worked with $\ncomb(\Eff^3)$ equipped with the equivalence relation of computational indistinguishability. We believe that this a reasonable setting to study many protocols in classical cryptography, as many functionalities of interest can be modelled by shared (probabilistic) functions, and many protocols can be modelled as (shared) $n$-combs. Similarly, one could get an analogous setting for $k$-partite quantum cryptography by starting e.g. from efficient sequences of quantum channels before applying the $\ncomb$-construction. However, both of these categories have some limitations, of which we'll mention some:

\begin{itemize}
  \item In this setting, resources consist of shared functions. However, some cryptographic functionalities of interest fall outside of this: for instance, bit commitment is more naturally modelled as a two-round functionality.
  \item The shared resources used as building blocks are used in synchronized rounds.
  \item The number of rounds for a protocol is fixed in advance, as opposed to being decided \eg probabilistically during the protocol.
  \item The number of parties is fixed in advance, as opposed to some real-world protocols where the number of parties can change dynamically during the protocol.
\end{itemize}

Let us now explain why choosing (let alone building) a starting category is not trivial.
In general, the following desiderata pull in opposite directions
\begin{itemize}
  \item The computational model should be mathematically tractable.
  \item The computational model should be sufficiently expressive. One should be able to model both cryptographic protocols and realistic attacks in the model. In particular, any adversarial behavior simply assumed away creates room for side-channel attacks.
\end{itemize}
To elaborate on the required expressiveness, cryptographically reasonable models of computation should have most of the following features:
\begin{itemize}
  \item computation should be probabilistic (or quantum) as there's no security without randomness,
  \item in general, concurrent computation should be asynchronous: after all cryptography behaves differently in synchronous and asynchronous settings~\cite{BCG93,KMTZ13} and we cannot assume synchronicity throughout,
  \item cryptographic protocols need not have a fixed number of communication rounds, and might instead be repeated until a success condition occurs,
  \item the number of parties need not be fixed in advance, but can change during the protocol.
\end{itemize}

In particular, the requirement of asynchronous probabilistic computation causes some difficulties for modelling cryptography, as discussed in~\cite{MT13}. To paraphrase, the issue is that traditionally concurrency in asynchronous systems is modelled by nondeterminism, so that a system describes the set of all possible behaviors. Unfortunately, this does not work too well for cryptography, where one wants to bound the probability of a successful attack. Cryptographers often solve this by letting an adversary be responsible for scheduling. This is not always a reasonable assumption, and makes the resulting models inherently less compositional. However, one can still achieve categorical models that achieve most of these features, even if many existing models are not phrased categorically. Indeed, the aforementioned~\cite{MT13} builds a model of interactive, asynchronous probabilistic computation, where composition of computations is associative (resulting in a category), with the cartesian product of sets inducing a parallel composition operation (resulting in an SMC). The resulting framework is rich enough to study many cryptographic protocols such as broadcasting and secret sharing.

Another possible source of a model stems from cryptHOL~\cite{BLS20}, an approach to formally verified cryptographic protocols using higher-order-logic. While~\cite{BLS20} does not give an explicit category of probabilistic and stateful computations, we expect it to induce one despite not having verified this in detail. We base our optimism on the fact that the monadic and coalgebraic techniques used in~\cite{BLS20} are inherently categorical. Moreover cryptHOL has also been used to formalize constructive cryptography~\cite{LSB+19}, and as a consequence one would expect that the same model could be phrased in categorical terms. We leave verifying this for future work as it is outside the scope of this paper.

Additional models can be found in the literature on programming language theory, and particularly from game semantics~\cite{winskel:game}. While these models are more often expressed categorically, and even result in compact closed categories so that Theorem~\ref{thm:bipartite} applies directly, they run the risk of being too involved mathematically, at least in order to gain mainstream traction among cryptographers.

We now turn our attention to modelling quantum cryptography. Again, one can use models coming from game semantics~\cite{clairambault:gamesforquantum,clairambaultetal:gamesforquantum2} with the aforementioned caveats. Another model, this time built by and for quantum cryptographers, can be obtained from~\cite{portmann:causal}. The model is rather general, as it is intended to be also suitable for relativistic quantum protocols, and indeed it has been used for such purposes in~\cite{VPD19}. As the constructions are rather involved, we restrict ourselves to giving a high-level description. As~\cite{portmann:causal} is not expressed in categorical words, we will explain why their model can equally well be interpreted as a category.

In this work, the authors work over a countable but otherwise arbitrary partial order~$T$, and then  define for each (finite-dimensional) Hilbert space $A$ a new Hilbert-space $F(A)$ modelling a wire with $A$-messages being sent/received at times given by $t\in T$. They then define causal maps  $A\to B$ as certain families of CPTP-maps $\{F(A)\to F(B)_C\}$ where $C$ ranges over downward-closed subsets of $T$---the intuition being that for a causal map the output at time $t\in T$ depends only on the past of $t$ with some delay. The authors then define two ways of combining causal boxes:
\begin{itemize}
\item  parallel composition, which we will denote by $\otimes$,  which takes  causal boxes $A\to B$ and $C\to D$ and produces a causal box $A\otimes C\to B\otimes D$.
\item An internal wiring operation, which takes as input a causal box $A\otimes B\to B\otimes C$, and produces a causal box $A\otimes C$. The intuition is that the resulting causal box is obtained by wiring the output $B$-wire into the the input $B$-wire.
\end{itemize}

Moreover, the authors prove that these composition operations satisfy ``composition order invariance'', so that adding loops commutes with parallel composition, and the order in which loops was added does not affect the end result. While the authors phrase their constructions in terms of the algebraic theory of systems of~\cite{MMP+18}, the parallel composition operation they define is in fact a total operation, whereas in a system algebra $f\otimes f$ is never defined~\cite[Definition 3.1]{MMP+18}. This allows us to extract a category from~\cite{portmann:causal}. The objects of this category are given by (finite-dimensional) message spaces $A,B,C,\dots$, and morphisms $A\to B$ are given by $T$-causal maps $A\to B$. The composite of $f\colon A\to B$ and $g\colon B\to C$ is given internally wiring the $B$-ports together in $g\otimes f$, pictorially represented by
 \[\begin{pic}
    \node[morphism] (f) {$g \circ f$};
    \draw (f.south) to ++(0,-.65) node[right] {$A$};
    \draw (f.north) to ++(0,.65) node[right] {$C$};
  \end{pic}
  \enspace:=\enspace
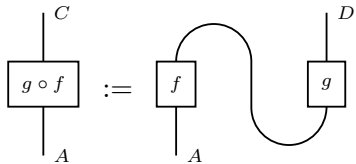
  \begin{pic}
    \node[morphism] (f) at (0,0) {$f$};
    \node[morphism] (g) at (2,0) {$g$};
    \draw (f.south) to ++(0,-.65) node[right] {$A$};
    \draw (f.north) to[in=180,out=90] ++(.5,.5) to[in=90,out=0] ++(.5,-.5) to ++(0,-.61) to[in=180,out=-90] ++(.5,-.5) to[in=-90,out=0] (g.south);
    \draw (g.north) to ++(0,.65) node[right] {$D$};
  \end{pic}\]

Now, composition order invariance implies that this composition operation is associative. As identity-maps are \emph{not} causal (informally this is because they're delay-free), this results only in a semicategory---\ie a category-like structure without identities. However, we can formally add identities, resulting in a symmetric monoidal category.\footnote{ One might be tempted to guess that \emph{any} model of abstract cryptography yields a category in an analogous manner. We make the more reserved but less precise guess that this holds for those models of abstract cryptography that are ``reasonable'' or arise ``naturally''. However, we don't think there exists a simple construction of an SMC directly from any (composition-order-invariant) system algebra in the sense of~\cite{MMP+18}: for one, in a system algebra $f\otimes f$ is never defined, whereas in a monoidal category $f\otimes f$ is always defined. Rather, we believe that natural sources of system algebras are also natural sources of SMCs.} Moreover, the authors equip these causal boxes with an explicit pseudometric. This pseudometric is a cryptographically well-motivated one, as the distance between $f$ and $g$ is defined in terms of the ability of an environment to guess whether it interacts with $f$ or $g$. Consequently one can also apply the asymptotic definitions of Sections~\ref{sec:metric} and~\ref{sec:metricsecurity} in this category.

\section{Outlook}\label{sec:outlook}

We have presented a categorical framework providing a general, flexible and mathematically robust way of reasoning about composability in cryptography. Besides contributing a further approach to composable cryptography and potentially helping with cross-talk and comparisons between existing approaches~\cite{CKL+19}, we believe that the current work opens the door for several further questions.

First, due to the generality of our approach we hope that one can, besides honest and malicious participants, reason about more refined kinds of adversaries composably. Indeed, we expect that Definition~\ref{def:attack}
is general enough to capture \eg honest-but-curious adversaries\footnote{Heuristically speaking this is the case: an honest-but-curious attack on $g\circ f$ should be factorizable as one on $g$ and one on $f$, and similarly an honest-but-curious attack on $g\otimes f$ should be factorizable into ones on $g$ and~$f$ that then forward their transcripts to an attack on $\id\otimes\id$.}. It would also be interesting to see if this captures even more general attacks, \eg situations where the sets of participants and dishonest parties can change during the protocol. This might require understanding our axiomatization of attack models more structurally and perhaps generalizing it. Does this structure (or a variant thereof) already arise in category theory? While we define an attack model on a category, perhaps one could define an attack model on a (strong) monoidal functor~$F$, the current definition being recovered when $F=\id$. Another approach would be to generalize the definitions and results of Section~\ref{sec:extensions} into general enriched notions of an attack model and of security, with Section~\ref{sec:extensions} then giving the special cases enriched over $\Equ$ and $\Met$.

Second, we expect that rephrasing cryptographic questions categorically would enable more cross-talk between cryptography and other fields already using category theory as an organizing principle. For instance, many existing approaches to composable cryptography develop their own models of concurrent, asynchronous, probabilistic and interactive computations. As categorical models of such computation exist in the context of game semantics~\cite{winskel:game,clairambault:gamesforquantum,clairambaultetal:gamesforquantum2}, one is left wondering whether the models of the semanticists' could be used to study and answer cryptographic questions, or conversely if the models developed by cryptographers contain valuable insights for programming language semantics.

Besides working inside concrete models---which ultimately blends into ``just doing composable cryptography''---one could study axiomatically how properties of a category relate to cryptographic properties in it. As a specific conjecture in this direction, if one has an environment structure~\cite{coecke:environment}, \ie coherent families of maps $\tinyground_A$ for each $A$ that axiomatize the idea of deleting a system, one might be able to talk about honest-but-curious adversaries at an abstract level. Similarly, having agents purify their actions is an important tool in quantum cryptography~\cite{LC97}---can categorical accounts of purification~\cite{chiribella:purification,cunningham:purity,coecke:environment} be used to elucidate this?

Finally, we hope to get more mileage out of the tools brought in with the categorical viewpoint. For instance, can one prove further no-go results pictorially? More specifically, given the impossibility results for two and three parties, one wonders if the ``only topology matters'' approach of string diagrams can be used to derive general impossibility results for $n$ parties sharing pairwise channels. Similarly, while diagrammatic languages have been used to reason about positive cryptographic results in the stand-alone setting~\cite{kissinger2017picture,breiner:graphicaldicrypto,breiner:selftesting}, can one push such approaches further now that composable security definitions have a clear categorical meaning? Besides the graphical methods, thinking of cryptography as a resource theory suggests using resource-theoretic tools such as monotones. While monotones have already been applied in cryptography~\cite{WW08:monotones}, a full understanding of cryptographically relevant monotones is still lacking.

\bibliographystyle{alphaurl}
\bibliography{categoricalcryptolmcs}

\newcommand{\etalchar}[1]{$^{#1}$}
\begin{thebibliography}{CDVW19b}

\bibitem[ABKM19]{abramskyetal:comonadicview}
Samson Abramsky, Rui~Soares Barbosa, Martti Karvonen, and Shane Mansfield.
\newblock A comonadic view of simulation and quantum resources.
\newblock In {\em 2019 34th Annual ACM/IEEE Symposium on Logic in Computer
  Science (LICS)}. IEEE, 2019.
\newblock \href {https://doi.org/10.1109/LICS.2019.8785677}
  {\path{doi:10.1109/LICS.2019.8785677}}.

\bibitem[Awo10]{awodey:categorytheory}
S.~Awodey.
\newblock {\em Category theory}.
\newblock Oxford University Press, 2010.

\bibitem[BB84]{BB84}
Charles~H. Bennett and Gilles Brassard.
\newblock Quantum cryptography: Public key distribution and coin tossing.
\newblock In {\em International Conference on Computers, Systems and Signal
  Processing}, pages 175--179, 1984.

\bibitem[BBB{\etalchar{+}}00]{BBB+00}
Eli Biham, Michel Boyer, P.~Oscar Boykin, Tal Mor, and Vwani Roychowdhury.
\newblock A proof of the security of quantum key distribution (extended
  abstract).
\newblock In {\em 32nd Annual ACM Symposium on Theory of Computing---STOC
  2000}, pages 715 -- 724, 2000.
\newblock \href {https://doi.org/10.1145/335305.335406}
  {\path{doi:10.1145/335305.335406}}.

\bibitem[BFK09]{BFK09}
Anne Broadbent, Joseph Fitzsimons, and Elham Kashefi.
\newblock Universal blind quantum computation.
\newblock In {\em 50th Annual Symposium on Foundations of Computer
  Science---FOCS 2009}, pages 517--526, 2009.
\newblock \href {https://doi.org/10.1109/FOCS.2009.36}
  {\path{doi:10.1109/FOCS.2009.36}}.

\bibitem[BK22]{BK22}
Anne Broadbent and Martti Karvonen.
\newblock Categorical composable cryptography.
\newblock In {\em Foundations of Software Science and Computation Structures
  (FoSSaCS)}, volume 13242 of {\em Lecture Notes in Computer Science}, pages
  161--183. Springer, 2022.
\newblock \href {https://doi.org/10.1007/978-3-030-99253-8_9}
  {\path{doi:10.1007/978-3-030-99253-8_9}}.

\bibitem[BKM18]{breiner:selftesting}
Spencer Breiner, Amir Kalev, and Carl~A. Miller.
\newblock Parallel self-testing of the {GHZ} state with a proof by diagrams.
\newblock In {\em Proceedings of QPL 2018}, volume 287 of {\em Electronic
  Proceedings in Theoretical Computer Science}, pages 43--66, 2018.
\newblock \href {https://doi.org/10.4204/eptcs.287.3}
  {\path{doi:10.4204/eptcs.287.3}}.

\bibitem[BLS20]{BLS20}
David~A. Basin, Andreas Lochbihler, and S.~Reza Sefidgar.
\newblock {CryptHOL}: {G}ame-{B}ased {P}roofs in {H}igher-{O}rder {L}ogic.
\newblock {\em Journal of Cryptology}, 33(2):494--566, January 2020.
\newblock \href {https://doi.org/10.1007/s00145-019-09341-z}
  {\path{doi:10.1007/s00145-019-09341-z}}.

\bibitem[BMR19]{breiner:graphicaldicrypto}
Spencer Breiner, Carl~A. Miller, and Neil~J. Ross.
\newblock Graphical methods in device-independent quantum cryptography.
\newblock {\em Quantum}, 3:146, 2019.
\newblock \href {https://doi.org/10.22331/q-2019-05-27-146}
  {\path{doi:10.22331/q-2019-05-27-146}}.

\bibitem[BOCG93]{BCG93}
Michael Ben-Or, Ran Canetti, and Oded Goldreich.
\newblock Asynchronous secure computation.
\newblock In {\em Proceedings of the twenty-fifth annual ACM symposium on
  Theory of computing}, pages 52--61, 1993.
\newblock \href {https://doi.org/10.1145/167088.167109}
  {\path{doi:10.1145/167088.167109}}.

\bibitem[BOHL{\etalchar{+}}05]{BHL+05}
Michael Ben-Or, Micha{\l} Horodecki, Debbie~W Leung, Dominic Mayers, and
  Jonathan Oppenheim.
\newblock The universal composable security of quantum key distribution.
\newblock In {\em 2nd Theory of Cryptography Conference---TCC 2005}, pages
  386--406, 2005.
\newblock \href {https://doi.org/10.1007/978-3-540-30576-7\_21}
  {\path{doi:10.1007/978-3-540-30576-7\_21}}.

\bibitem[BOM04]{BM04arxiv}
Michael Ben-Or and Dominic Mayers.
\newblock General security definition and composability for quantum \&
  classical protocols, 2004.

\bibitem[Bor94a]{borceux:vol1}
F.~Borceux.
\newblock {\em Handbook of Categorical Algebra 1: Basic Category Theory}.
\newblock Cambridge University Press, 1994.
\newblock \href {https://doi.org/10.1017/cbo9780511525858}
  {\path{doi:10.1017/cbo9780511525858}}.

\bibitem[Bor94b]{borceux:vol2}
F.~Borceux.
\newblock {\em Handbook of Categorical Algebra 2: Categories and Structures}.
\newblock Cambridge University Press, 1994.
\newblock \href {https://doi.org/10.1017/CBO9780511525865}
  {\path{doi:10.1017/CBO9780511525865}}.

\bibitem[BPW04]{BPW04}
Michael Backes, Birgit Pfitzmann, and Michael Waidner.
\newblock A general composition theorem for secure reactive systems.
\newblock In {\em 1st Theory of Cryptography Conference---TCC 2004}, pages
  336--354, 2004.
\newblock \href {https://doi.org/10.1007/978-3-540-24638-1\_19}
  {\path{doi:10.1007/978-3-540-24638-1\_19}}.

\bibitem[BPW07]{BPW07}
Michael Backes, Birgit Pfitzmann, and Michael Waidner.
\newblock The reactive simulatability {(RSIM)} framework for asynchronous
  systems.
\newblock {\em Information and Computation}, 205(12):1685--1720, 2007.
\newblock \href {https://doi.org/10.1016/j.ic.2007.05.002}
  {\path{doi:10.1016/j.ic.2007.05.002}}.

\bibitem[Can01]{Can01}
Ran Canetti.
\newblock Universally composable security: {A} new paradigm for cryptographic
  protocols.
\newblock In {\em 42nd Annual Symposium on Foundations of Computer
  Science---FOCS 2001}, pages 136--145, 2001.
\newblock \href {https://doi.org/10.1109/SFCS.2001.959888}
  {\path{doi:10.1109/SFCS.2001.959888}}.

\bibitem[CDP10]{chiribella:purification}
Giulio Chiribella, Giacomo~Mauro D'Ariano, and Paolo Perinotti.
\newblock Probabilistic theories with purification.
\newblock {\em Physical Review A}, 81(6), June 2010.
\newblock \href {https://doi.org/10.1103/physreva.81.062348}
  {\path{doi:10.1103/physreva.81.062348}}.

\bibitem[CdVW19a]{clairambaultetal:gamesforquantum2}
Pierre Clairambault, Marc de~Visme, and Glynn Winskel.
\newblock Concurrent quantum strategies.
\newblock In {\em International Conference on Reversible Computation}, pages
  3--19. Springer, 2019.
\newblock \href {https://doi.org/10.1007/978-3-030-21500-2\_1}
  {\path{doi:10.1007/978-3-030-21500-2\_1}}.

\bibitem[CDVW19b]{clairambault:gamesforquantum}
Pierre Clairambault, Marc De~Visme, and Glynn Winskel.
\newblock Game semantics for quantum programming.
\newblock {\em Proceedings of the ACM on Programming Languages}, 3(POPL):1--29,
  2019.
\newblock \href {https://doi.org/10.1145/3290345} {\path{doi:10.1145/3290345}}.

\bibitem[CF01]{CF01}
Ran Canetti and Marc Fischlin.
\newblock Universally composable commitments.
\newblock In {\em Advances in cryptology---CRYPTO 2001}, pages 19--40.
  Springer, 2001.
\newblock \href {https://doi.org/10.1007/3-540-44647-8\_2}
  {\path{doi:10.1007/3-540-44647-8\_2}}.

\bibitem[CFS16]{CFS16}
Bob Coecke, Tobias Fritz, and Robert~W Spekkens.
\newblock A mathematical theory of resources.
\newblock {\em Information and Computation}, 250:59--86, 2016.
\newblock \href {https://doi.org/10.1016/j.ic.2016.02.008}
  {\path{doi:10.1016/j.ic.2016.02.008}}.

\bibitem[CG19]{chitambar:resource}
Eric Chitambar and Gilad Gour.
\newblock Quantum resource theories.
\newblock {\em Reviews of Modern Physics}, 91(2):025001, 2019.
\newblock \href {https://doi.org/10.1103/revmodphys.91.025001}
  {\path{doi:10.1103/revmodphys.91.025001}}.

\bibitem[CGG{\etalchar{+}}22]{cruttwell2022categorical}
Geoffrey S.~H. Cruttwell, Bruno Gavranovi{\'{c}}, Neil Ghani, Paul Wilson, and
  Fabio Zanasi.
\newblock Categorical foundations of gradient-based learning.
\newblock In {\em Programming Languages and Systems ({ESOP})}, volume 13240 of
  {\em Lecture Notes in Computer Science}, pages 1--28. Springer, 2022.
\newblock \href {https://doi.org/10.1007/978-3-030-99336-8_1}
  {\path{doi:10.1007/978-3-030-99336-8_1}}.

\bibitem[CH17]{cunningham:purity}
Oscar Cunningham and Chris Heunen.
\newblock Purity through factorisation.
\newblock In {\em Proceedings of QPL 2017}, volume 266 of {\em Electronic
  Proceedings in Theoretical Computer Science}, pages 315--328, 2017.
\newblock \href {https://doi.org/10.4204/EPTCS.266.20}
  {\path{doi:10.4204/EPTCS.266.20}}.

\bibitem[CK17]{coeckekissinger:picturing}
Bob Coecke and Aleks Kissinger.
\newblock {\em Picturing Quantum Processes}.
\newblock Cambridge University Press, 2017.
\newblock \href {https://doi.org/10.1017/9781316219317}
  {\path{doi:10.1017/9781316219317}}.

\bibitem[CKLS19]{CKL+19}
Jan Camenisch, Ralf K{\"u}sters, Anna Lysyanskaya, and Alessandra Scafuro.
\newblock {Practical Yet Composably Secure Cryptographic Protocols (Dagstuhl
  Seminar 19042)}.
\newblock {\em Dagstuhl Reports}, 9(1):88--103, 2019.
\newblock \href {https://doi.org/10.4230/DagRep.9.1.88}
  {\path{doi:10.4230/DagRep.9.1.88}}.

\bibitem[CLM{\etalchar{+}}14]{Chitambaretal:LOCC}
Eric Chitambar, Debbie Leung, Laura Man{\v{c}}inska, Maris Ozols, and Andreas
  Winter.
\newblock Everything you always wanted to know about {LOCC} (but were afraid to
  ask).
\newblock {\em Communications in Mathematical Physics}, 328(1):303--326, 2014.
\newblock \href {https://doi.org/10.1007/s00220-014-1953-9}
  {\path{doi:10.1007/s00220-014-1953-9}}.

\bibitem[CP10]{coecke2010categories}
Bob Coecke and Eric~Oliver Paquette.
\newblock Categories for the practising physicist.
\newblock In {\em New Structures for Physics}, pages 173--286. Springer, 2010.
\newblock \href {https://doi.org/10.1007/978-3-642-12821-9\_3}
  {\path{doi:10.1007/978-3-642-12821-9\_3}}.

\bibitem[CP12]{coecke:environment}
Bob Coecke and Simon Perdrix.
\newblock {Environment and classical channels in categorical quantum
  mechanics}.
\newblock {\em {Logical Methods in Computer Science}}, {Volume 8, Issue 4},
  2012.
\newblock \href {https://doi.org/10.2168/LMCS-8(4:14)2012}
  {\path{doi:10.2168/LMCS-8(4:14)2012}}.

\bibitem[CWW{\etalchar{+}}11]{coecke:graphicalqkd}
Bob Coecke, Quanlong Wang, Baoshan Wang, Yongjun Wang, and Qiye Zhang.
\newblock Graphical calculus for quantum key distribution (extended abstract).
\newblock {\em Electronic Notes in Theoretical Computer Science},
  270(2):231--249, 2011.
\newblock \href {https://doi.org/10.1016/j.entcs.2011.01.034}
  {\path{doi:10.1016/j.entcs.2011.01.034}}.

\bibitem[DDMP03a]{DDMP03a}
Anupam Datta, Ante Derek, John~C Mitchell, and Dusko Pavlovic.
\newblock A derivation system for security protocols and its logical
  formalization.
\newblock In {\em 16th {IEEE} Computer Security Foundations Workshop, 2003.
  Proceedings.}, pages 109--125. IEEE, 2003.
\newblock \href {https://doi.org/10.1109/csfw.2003.1212708}
  {\path{doi:10.1109/csfw.2003.1212708}}.

\bibitem[DDMP03b]{DDMP03b}
Anupam Datta, Ante Derek, John~C Mitchell, and Dusko Pavlovic.
\newblock Secure protocol composition.
\newblock {\em Electronic Notes in Theoretical Computer Science}, 83:201--226,
  2003.
\newblock \href {https://doi.org/10.1016/s1571-0661(03)50011-1}
  {\path{doi:10.1016/s1571-0661(03)50011-1}}.

\bibitem[DDMP05]{DDMP05}
Anupam Datta, Ante Derek, John~C. Mitchell, and Dusko Pavlovic.
\newblock A derivation system and compositional logic for security protocols.
\newblock {\em Journal of Computer Security}, 13(3):423–482, August 2005.
\newblock \href {https://doi.org/10.3233/JCS-2005-13304}
  {\path{doi:10.3233/JCS-2005-13304}}.

\bibitem[DDMR07]{DDMR07}
Anupam Datta, Ante Derek, John~C. Mitchell, and Arnab Roy.
\newblock Protocol composition logic ({PCL}).
\newblock {\em Electronic Notes in Theoretical Computer Science}, 172:311--358,
  April 2007.
\newblock \href {https://doi.org/10.1016/j.entcs.2007.02.012}
  {\path{doi:10.1016/j.entcs.2007.02.012}}.

\bibitem[DMP01]{DMP01}
N.~Durgin, J.~Mitchell, and D.~Pavlovic.
\newblock A compositional logic for protocol correctness.
\newblock In {\em Proceedings. 14th {IEEE} Computer Security Foundations
  Workshop, 2001.} {IEEE}, 2001.
\newblock \href {https://doi.org/10.1109/csfw.2001.930150}
  {\path{doi:10.1109/csfw.2001.930150}}.

\bibitem[DMP03]{DMP03}
Nancy Durgin, John Mitchell, and Dusko Pavlovic.
\newblock A compositional logic for proving security properties of protocols.
\newblock {\em Journal of Computer Security}, 11(4):677–721, October 2003.
\newblock \href {https://doi.org/10.3233/JCS-2003-11407}
  {\path{doi:10.3233/JCS-2003-11407}}.

\bibitem[Fri15]{Fritz2015}
Tobias Fritz.
\newblock Resource convertibility and ordered commutative monoids.
\newblock {\em Mathematical Structures in Computer Science}, 27(6):850--938,
  2015.
\newblock \href {https://doi.org/10.1017/s0960129515000444}
  {\path{doi:10.1017/s0960129515000444}}.

\bibitem[Fri20]{fritz:markovcats}
Tobias Fritz.
\newblock A synthetic approach to {M}arkov kernels, conditional independence
  and theorems on sufficient statistics.
\newblock {\em Advances in Mathematics}, 370:107239, 2020.
\newblock \href {https://doi.org/10.1016/j.aim.2020.107239}
  {\path{doi:10.1016/j.aim.2020.107239}}.

\bibitem[FS19]{FS19}
Brendan Fong and David~I. Spivak.
\newblock {\em An Invitation to Applied Category Theory: Seven Sketches in
  Compositionality}.
\newblock Cambridge University Press, 2019.
\newblock \href {https://doi.org/10.1017/9781108668804}
  {\path{doi:10.1017/9781108668804}}.

\bibitem[FST19]{fongetal:backprop}
Brendan Fong, David Spivak, and Remy Tuyeras.
\newblock Backprop as functor: A compositional perspective on supervised
  learning.
\newblock In {\em 2019 34th Annual {ACM}/{IEEE} Symposium on Logic in Computer
  Science ({LICS})}, 2019.
\newblock \href {https://doi.org/10.1109/lics.2019.8785665}
  {\path{doi:10.1109/lics.2019.8785665}}.

\bibitem[Heu08]{heunen:qkd}
Chris Heunen.
\newblock Compactly accessible categories and quantum key distribution.
\newblock {\em Logical Methods in Computer Science}, 4(4), 2008.
\newblock \href {https://doi.org/10.2168/lmcs-4(4:9)2008}
  {\path{doi:10.2168/lmcs-4(4:9)2008}}.

\bibitem[Hil11]{hillebrand:superdense}
Anne Hillebrand.
\newblock Superdense coding with {GHZ} and quantum key distribution with {W} in
  the {ZX}-calculus.
\newblock In {\em Proceedings of QPL 2011}, volume~95 of {\em Electronic
  Proceedings in Theoretical Computer Science}, pages 103--121, 2011.
\newblock \href {https://doi.org/10.4204/EPTCS.95.10}
  {\path{doi:10.4204/EPTCS.95.10}}.

\bibitem[Hin20]{hines:diagrammaticrypto}
Peter~M. Hines.
\newblock A diagrammatic approach to information flow in encrypted
  communication, 2020.
\newblock \href {https://doi.org/10.1007/978-3-030-62230-5\_9}
  {\path{doi:10.1007/978-3-030-62230-5\_9}}.

\bibitem[HO13]{horodecki:resource}
Michal Horodecki and Jonathan Oppenheim.
\newblock ({Q}uantumness in the context of) {R}esource {T}heories.
\newblock {\em International Journal of Modern Physics B}, 27(01n03):1345019,
  2013.
\newblock \href {https://doi.org/10.1142/s0217979213450197}
  {\path{doi:10.1142/s0217979213450197}}.

\bibitem[HS15]{HS15}
Dennis Hofheinz and Victor Shoup.
\newblock {GNUC}: A new universal composability framework.
\newblock {\em Journal of Cryptology}, 28(3):423--508, 2015.
\newblock \href {https://doi.org/10.1007/s00145-013-9160-y}
  {\path{doi:10.1007/s00145-013-9160-y}}.

\bibitem[HV19]{heunenvicary:categories}
Chris Heunen and Jamie Vicary.
\newblock {\em Categories for Quantum Theory: an introduction}.
\newblock Oxford University Press, USA, 2019.

\bibitem[KL15]{KL15}
Jonathan Katz and Yehuda Lindell.
\newblock {\em Introduction to Modern Cryptography}.
\newblock CRC Press, 2\textsuperscript{nd} edition, 2015.

\bibitem[KMTZ13]{KMTZ13}
Jonathan Katz, Ueli Maurer, Bj\"{o}rn Tackmann, and Vassilis Zikas.
\newblock Universally composable synchronous computation.
\newblock In {\em Theory of Cryptography}, pages 477--498. Springer, 2013.
\newblock \href {https://doi.org/10.1007/978-3-642-36594-2\_27}
  {\path{doi:10.1007/978-3-642-36594-2\_27}}.

\bibitem[KRBM07]{KRBM07}
Robert K{\"o}nig, Renato Renner, Andor Bariska, and Ueli Maurer.
\newblock Small accessible quantum information does not imply security.
\newblock {\em Physical Review Letters}, 98(14):140502, 2007.
\newblock \href {https://doi.org/10.1103/PhysRevLett.98.140502}
  {\path{doi:10.1103/PhysRevLett.98.140502}}.

\bibitem[KTR20]{KTR20}
Ralf K{\"u}sters, Max Tuengerthal, and Daniel Rausch.
\newblock The {IITM} model: a simple and expressive model for universal
  composability.
\newblock {\em Journal of Cryptology}, 33(4):1461--1584, 2020.
\newblock \href {https://doi.org/10.1007/s00145-020-09352-1}
  {\path{doi:10.1007/s00145-020-09352-1}}.

\bibitem[KTW17]{kissinger2017picture}
Aleks Kissinger, Sean Tull, and Bas Westerbaan.
\newblock Picture-perfect {Q}uantum {K}ey {D}istribution, 2017.

\bibitem[LC97]{LC97}
Hoi-Kwong Lo and H~.F. Chau.
\newblock Is quantum bit commitment really possible?
\newblock {\em Physical Review Letters}, 78(17):3410--3413, 1997.
\newblock \href {https://doi.org/10.1103/PhysRevLett.78.3410}
  {\path{doi:10.1103/PhysRevLett.78.3410}}.

\bibitem[Lei04]{leinster2004}
Tom Leinster.
\newblock {\em Higher Operads, Higher Categories}.
\newblock Cambridge University Press, 2004.
\newblock \href {https://doi.org/10.1017/cbo9780511525896}
  {\path{doi:10.1017/cbo9780511525896}}.

\bibitem[Lei14]{leinster:basicCT}
Tom Leinster.
\newblock {\em Basic category theory}, volume 143.
\newblock Cambridge University Press, 2014.
\newblock \href {https://doi.org//10.1017/CBO9781107360068}
  {\path{doi:/10.1017/CBO9781107360068}}.

\bibitem[LHM19]{LHM19}
Kevin Liao, Matthew~A. Hammer, and Andrew Miller.
\newblock {ILC}: a calculus for composable, computational cryptography.
\newblock In {\em Proceedings of the 40th {ACM} {SIGPLAN} Conference on
  Programming Language Design and Implementation}, pages 640--654. {ACM}, June
  2019.
\newblock \href {https://doi.org/10.1145/3314221.3314607}
  {\path{doi:10.1145/3314221.3314607}}.

\bibitem[Lin17]{L17}
Yehuda Lindell.
\newblock How to simulate it {\textendash} a tutorial on the simulation proof
  technique.
\newblock In {\em Tutorials on the Foundations of Cryptography}, pages
  277--346. Springer International Publishing, 2017.
\newblock \href {https://doi.org/10.1007/978-3-319-57048-8_6}
  {\path{doi:10.1007/978-3-319-57048-8_6}}.

\bibitem[LSBM19]{LSB+19}
Andreas Lochbihler, S.~Reza Sefidgar, David Basin, and Ueli Maurer.
\newblock Formalizing {C}onstructive {C}ryptography using {CryptHOL}.
\newblock In {\em 2019 {IEEE} 32nd Computer Security Foundations Symposium
  ({CSF})}. {IEEE}, June 2019.
\newblock \href {https://doi.org/10.1109/csf.2019.00018}
  {\path{doi:10.1109/csf.2019.00018}}.

\bibitem[{Mac}71]{maclane:categories}
S.~{Mac Lane}.
\newblock {\em Categories for the Working Mathematician}.
\newblock Springer, 2nd edition, 1971.

\bibitem[Mau11]{Mau11}
Ueli Maurer.
\newblock Constructive cryptography--a new paradigm for security definitions
  and proofs.
\newblock In {\em Joint Workshop on Theory of Security and Applications---TOSCA
  2011}, pages 33--56, 2011.
\newblock \href {https://doi.org/10.1007/978-3-642-27375-9\_3}
  {\path{doi:10.1007/978-3-642-27375-9\_3}}.

\bibitem[May96]{May96}
Dominic Mayers.
\newblock The trouble with quantum bit commitment, 1996.

\bibitem[May01]{May01}
Dominic Mayers.
\newblock Unconditional security in quantum cryptography.
\newblock {\em Journal of the ACM}, 48(3):351--406, 2001.
\newblock \href {https://doi.org/10.1145/382780.382781}
  {\path{doi:10.1145/382780.382781}}.

\bibitem[Mel06]{mellies2006functorial}
Paul-Andr{\'e} Melli{\`e}s.
\newblock Functorial boxes in string diagrams.
\newblock In {\em Computer Science Logic}, Lecture Notes in Computer Science,
  pages 1--30. Springer, 2006.
\newblock \href {https://doi.org/10.1007/11874683\_1}
  {\path{doi:10.1007/11874683\_1}}.

\bibitem[MMP95]{Mifsudetal:controlstructures}
A.~Mifsud, R.~Milner, and J.~Power.
\newblock Control structures.
\newblock In {\em Proceedings of Tenth Annual {IEEE} Symposium on {L}ogic in
  {C}omputer {S}cience}, pages 188--198. {IEEE}, 1995.
\newblock \href {https://doi.org/10.1109/lics.1995.523256}
  {\path{doi:10.1109/lics.1995.523256}}.

\bibitem[MMP{\etalchar{+}}18]{MMP+18}
Christian Matt, Ueli Maurer, Christopher Portmann, Renato Renner, and Bj{\"o}rn
  Tackmann.
\newblock Toward an algebraic theory of systems.
\newblock {\em Theoretical Computer Science}, 747:1--25, 2018.
\newblock \href {https://doi.org/10.1016/j.tcs.2018.06.001}
  {\path{doi:10.1016/j.tcs.2018.06.001}}.

\bibitem[MQR09]{MR09}
J{\"o}rn M{\"u}ller-Quade and Renato Renner.
\newblock Composability in quantum cryptography.
\newblock {\em New Journal of Physics}, 11(8):085006, 2009.
\newblock \href {https://doi.org/10.1088/1367-2630/11/8/085006}
  {\path{doi:10.1088/1367-2630/11/8/085006}}.

\bibitem[MR11]{MR11}
Ueli Maurer and Renato Renner.
\newblock Abstract cryptography.
\newblock In {\em Innovations in Computer Science---ICS 2011}, 2011.

\bibitem[MT13]{MT13}
Daniele Micciancio and Stefano Tessaro.
\newblock An equational approach to secure multi-party computation.
\newblock In {\em 4th Conference on Innovations in Theoretical Computer
  Science---ITCS 2013}, pages 355--372, 2013.
\newblock \href {https://doi.org/10.1145/2422436.2422478}
  {\path{doi:10.1145/2422436.2422478}}.

\bibitem[MV20]{moeller:monoidalgrothendieck}
Joe Moeller and Christina Vasilakopoulou.
\newblock Monoidal {G}rothendieck construction.
\newblock {\em Theory and Applications of Categories}, 35(31):1159--1207, 2020.

\bibitem[Pav97]{pavlovic1997}
Dusko Pavlovic.
\newblock Categorical logic of names and abstraction in action calculi.
\newblock {\em Mathematical Structures in Computer Science}, 7(6):619–637,
  1997.
\newblock \href {https://doi.org/10.1017/S0960129597002296}
  {\path{doi:10.1017/S0960129597002296}}.

\bibitem[Pav12]{pavlovic2012tracing}
Dusko Pavlovic.
\newblock Tracing the man in the middle in monoidal categories.
\newblock In {\em Coalgebraic Methods in Computer Science}, pages 191--217.
  Springer, 2012.
\newblock \href {https://doi.org/10.1007/978-3-642-32784-1\_11}
  {\path{doi:10.1007/978-3-642-32784-1\_11}}.

\bibitem[Pav14]{dusko:crypto}
Dusko Pavlovic.
\newblock Chasing diagrams in cryptography.
\newblock In Claudia Casadio, Bob Coecke, Michael Moortgat, and Philip Scott,
  editors, {\em Categories and Types in Logic, Language, and Physics: Essays
  Dedicated to Jim Lambek on the Occasion of His 90th Birthday}, pages
  353--367. Springer Berlin Heidelberg, Berlin, Heidelberg, 2014.
\newblock \href {https://doi.org/10.1007/978-3-642-54789-8\_19}
  {\path{doi:10.1007/978-3-642-54789-8\_19}}.

\bibitem[PMM{\etalchar{+}}17]{portmann:causal}
Christopher Portmann, Christian Matt, Ueli Maurer, Renato Renner, and Bj{\"o}rn
  Tackmann.
\newblock Causal boxes: quantum information-processing systems closed under
  composition.
\newblock {\em IEEE Transactions on Information Theory}, 63(5):3277--3305,
  2017.
\newblock \href {https://doi.org/10.1109/TIT.2017.2676805}
  {\path{doi:10.1109/TIT.2017.2676805}}.

\bibitem[PR08]{PR08}
Manoj Prabhakaran and Mike Rosulek.
\newblock Cryptographic complexity of multi-party computation problems:
  Classifications and separations.
\newblock In {\em Advances in Cryptology---CRYPTO 2008}, pages 262--279, 2008.
\newblock \href {https://doi.org/10.1007/978-3-540-85174-5\_15}
  {\path{doi:10.1007/978-3-540-85174-5\_15}}.

\bibitem[PR14]{PR14arxiv}
Christopher Portmann and Renato Renner.
\newblock Cryptographic security of quantum key distribution, 2014.

\bibitem[PW00]{PW00}
Birgit Pfitzmann and Michael Waidner.
\newblock A model for asynchronous reactive systems and its application to
  secure message transmission.
\newblock In {\em 2001 IEEE Symposium on Security and Privacy---S\&P 2001},
  pages 184--200, 2000.
\newblock \href {https://doi.org/10.1109/SECPRI.2001.924298}
  {\path{doi:10.1109/SECPRI.2001.924298}}.

\bibitem[Ren05]{Ren05}
Renato Renner.
\newblock Security of quantum key distribution.
\newblock {\em International Journal of Quantum Information}, 06(01):1--127,
  2005.
\newblock \href {https://doi.org/10.1142/S0219749908003256}
  {\path{doi:10.1142/S0219749908003256}}.

\bibitem[Rie17]{riehl:categorytheory}
Emily Riehl.
\newblock {\em Category theory in context}.
\newblock Courier Dover Publications, 2017.

\bibitem[Sel10]{Sel10}
Peter Selinger.
\newblock A survey of graphical languages for monoidal categories.
\newblock In {\em New structures for physics}, pages 289--355. Springer, 2010.
\newblock \href {https://doi.org/10.1007/978-3-642-12821-9\_4}
  {\path{doi:10.1007/978-3-642-12821-9\_4}}.

\bibitem[SHW20]{sunwang:graphicalbc}
Xin Sun, Feifei He, and Quanlong Wang.
\newblock Impossibility of quantum bit commitment, a categorical perspective.
\newblock {\em Axioms}, 9(1):28, 2020.
\newblock \href {https://doi.org/10.3390/axioms9010028}
  {\path{doi:10.3390/axioms9010028}}.

\bibitem[SP00]{SP00}
Peter~W. Shor and John Preskill.
\newblock Simple proof of security of the {BB}84 quantum key distribution
  protocol.
\newblock {\em Physical Review Letters}, 85(2):441--444, 2000.
\newblock \href {https://doi.org/10.1103/physrevlett.85.441}
  {\path{doi:10.1103/physrevlett.85.441}}.

\bibitem[SV13]{stay:crypto}
Mike Stay and Jamie Vicary.
\newblock Bicategorical semantics for nondeterministic computation.
\newblock In {\em Proceedings of the Twenty-ninth Conference on the
  Mathematical Foundations of Programming Semantics, MFPS XXIX}, volume 298 of
  {\em Electronic Notes in Theoretical Computer Science}, pages 367 -- 382,
  2013.
\newblock \href {https://doi.org/10.1016/j.entcs.2013.09.022}
  {\path{doi:10.1016/j.entcs.2013.09.022}}.

\bibitem[Swe69]{sweedler:integrals}
Moss~Eisenberg Sweedler.
\newblock Integrals for hopf algebras.
\newblock {\em Annals of Mathematics}, 89(2):323--335, 1969.
\newblock \href {https://doi.org/10.2307/1970672} {\path{doi:10.2307/1970672}}.

\bibitem[TLGR12]{TLGR12}
Marco Tomamichel, Charles Ci~Wen Lim, Nicolas Gisin, and Renato Renner.
\newblock Tight finite-key analysis for quantum cryptography.
\newblock {\em Nature Communications}, 3:634, 2012.
\newblock \href {https://doi.org/10.1038/ncomms1631}
  {\path{doi:10.1038/ncomms1631}}.

\bibitem[Unr10]{Unr10}
Dominique Unruh.
\newblock Universally composable quantum multi-party computation.
\newblock In {\em Advances in Cryptology---EUROCRYPT 2010}, pages 486--505,
  2010.
\newblock \href {https://doi.org/10.1007/978-3-642-13190-5\_25}
  {\path{doi:10.1007/978-3-642-13190-5\_25}}.

\bibitem[VPDR19]{VPD19}
Venkatesh Vilasini, Christopher Portmann, and L{\'\i}dia Del~Rio.
\newblock Composable security in relativistic quantum cryptography.
\newblock {\em New Journal of Physics}, 21(4):043057, 2019.
\newblock \href {https://doi.org/10.1088/1367-2630/ab0e3b}
  {\path{doi:10.1088/1367-2630/ab0e3b}}.

\bibitem[Win13]{winskel:game}
Glynn Winskel.
\newblock Distributed probabilistic and quantum strategies.
\newblock {\em Electronic Notes in Theoretical Computer Science}, 298:403--425,
  2013.
\newblock \href {https://doi.org/10.1016/j.entcs.2013.09.024}
  {\path{doi:10.1016/j.entcs.2013.09.024}}.

\bibitem[WW08]{WW08:monotones}
Stefan Wolf and J\"{u}rg Wullschleger.
\newblock New monotones and lower bounds in unconditional two-party
  computation.
\newblock {\em {IEEE} Transactions on Information Theory}, 54(6):2792--2797,
  2008.
\newblock \href {https://doi.org/10.1109/tit.2008.921674}
  {\path{doi:10.1109/tit.2008.921674}}.

\end{thebibliography}

\end{document}